\documentclass[a4paper]{article}

\usepackage[english]{babel}
\usepackage[utf8]{inputenc}
\usepackage[T1]{fontenc}
\usepackage{csquotes}
\usepackage{mathtools}
\usepackage{amsthm}
\usepackage{amsfonts}
\usepackage{amssymb}
\usepackage[mathscr]{eucal}
\usepackage{multirow}
\usepackage{pdfpages}
\usepackage{color}
\usepackage[margin=3cm]{geometry}
\usepackage[pdftex,pdfborder={0 0 0},linktoc=all]{hyperref}

\theoremstyle{plain}
	\newtheorem{thm}{Theorem}[section]
	\newtheorem{cor}[thm]{Corollary}
	\newtheorem{lem}[thm]{Lemma}
	\newtheorem{prop}[thm]{Proposition}

\theoremstyle{definition}
	\newtheorem{dfn}[thm]{Definition}
	\newtheorem{dfns}[thm]{Definitions}

\theoremstyle{remark}
	\newtheorem{ex}[thm]{Examples}
	\newtheorem{rem}[thm]{Remark}
	\newtheorem{rems}[thm]{Remarks}

\numberwithin{equation}{section}
\numberwithin{figure}{section}

\newcommand{\C}{\mathbb{C}}
\newcommand{\N}{\mathbb{N}}
\newcommand{\Q}{\mathbb{Q}}
\newcommand{\R}{\mathbb{R}}
\newcommand{\T}{\mathbb{T}}
\newcommand{\Z}{\mathbb{Z}}

\newcommand{\dx}{\dmesure\!}
\newcommand{\deron}[2]{\frac{\partial #1}{\partial #2}}
\newcommand{\esp}[2][]{\mathbb{E}_{#1}\!\left[ #2 \right]}
\newcommand{\mvert}{\mathrel{}\middle|\mathrel{}}
\newcommand{\norm}[1]{\left\lvert #1 \right\rvert}
\newcommand{\Norm}[1]{\left\lVert #1 \right\rVert}
\newcommand{\prsc}[2]{\left\langle #1\,, #2 \right\rangle}
\newcommand{\var}[1]{\Var\!\left( #1 \right)}

\renewcommand{\L}{\mathcal{L}}
\renewcommand{\P}{\mathbb{P}}
\renewcommand{\S}{\mathbb{S}}

\renewcommand{\bar}{\overline}
\renewcommand{\epsilon}{\varepsilon}
\renewcommand{\geq}{\geqslant}
\renewcommand{\leq}{\leqslant}
\renewcommand{\tilde}{\widetilde}

\DeclareMathOperator{\card}{card}
\DeclareMathOperator{\dmesure}{d}
\DeclareMathOperator{\Id}{Id}

\DeclareMathOperator{\Sp}{Sp}
\DeclareMathOperator{\Var}{Var}

\author{Thomas Letendre\,\thanks{Thomas Letendre, Université Paris-Saclay, CNRS, Institut de Mathématiques d’Orsay, 91405 Orsay, France; e-mail: \url{letendre@math.cnrs.fr}. Thomas Letendre is supported by the French National Research Agency through the ANR grants SpInQS (ANR-17-CE40-0011) and UniRaNDom (ANR-17-CE40-0008).} \and Henrik Ueberschär\,\thanks{Henrik Ueberschär, Sorbonne Université, Institut de Mathématiques de Jussieu - Paris Rive Gauche, {F-75005} Paris, France; e-mail: \url{henrik.ueberschar@imj-prg.fr}. Henrik Ueberschär is supported by the French National Research Agency through the ANR grant SpInQS (ANR-17-CE40-0011).}}
\date{\today}
\title{Random moments for the new eigenfunctions of point scatterers on rectangular flat tori}

\begin{document}

\maketitle

\begin{abstract}
We define a random model for the moments of the new eigenfunctions of a point scatterer on a $2$-dimensional rectangular flat torus. In the deterministic setting, \u{S}eba conjectured these moments to be asymptotically Gaussian, in the semi-classical limit. This conjecture was disproved by Kurlberg--Ueberschär on Diophantine tori. In our model, we describe the accumulation points in distribution of the randomized moments, in the semi-classical limit. We prove that asymptotic Gaussianity holds if and only if some function, modeling the multiplicities of the Laplace eigenfunctions, diverges to $+\infty$.
\end{abstract}

\begin{description}
\item{\textbf{Keywords:}} Point scatterer, Poisson point process, Quantum chaos, Random wave model.
\item{\textbf{Mathematics Subject Classification 2010:}} 35P20, 58J50, 60G55, 81Q50.
\end{description}


\tableofcontents


\section{Introduction}
\label{sec introduction}


\subsection{Background}
\label{subsec background}

Let $(M,g)$ be a closed Riemannian manifold. A classical dynamical system on $M$ is given by a Hamiltonian vector field on the associated phase space $T^*M$. In this classical description, one tries to understand the flow of the vector field under study. Given such a classical system, it is possible to define a corresponding quantum dynamical system, which is an unbounded self-adjoint operator on $L^2(M)$, the space of square integrable functions on $M$. In the quantum description, one is interested in the spectrum of this operator and the properties of its eigenfunctions. A typical example of quantum system is the Laplace--Beltrami operator $\Delta$ on $M$, which is the quantum version of the geodesic vector field on $T^*M$.

Quantum chaos is concerned with the asymptotic behavior, in the semi-classical limit, of the eigenfunctions of quantum systems whose classical counterpart is chaotic. The term semi-classical limit means that we consider eigenfunctions in the regime where the associated eigenvalues tend to infinity. In this regime, one expects features of the classical dynamics to emerge in the properties of the eigenfunctions. An important idea in this field is Berry's Random Wave Conjecture (see~\cite{Ber1977}), which suggests that the eigenfunctions of quantum chaotic systems should behave like a random superposition of plane waves. Much effort has been made in recent years to develop a rigorous mathematical formulation for this conjecture and to prove its validity for various types of systems (see~\cite{ABLM2018} and~\cite{Ing2017}, for example). Let us mention that, after the present paper was written, Louis Gass proved that Berry's Conjecture holds almost surely for Gaussian Random Waves on~$M$, that is for random linear combinations of Laplace eigenfunctions with Gaussian coefficients. His result is valid for both band-limited and monochromatic random waves. For more details, we refer the interested reader to his beautiful paper~\cite{Gas2020}.

A weak form of Berry's Conjecture states that, for a quantum chaotic system, the value distribution of an $L^2$-normalized eigenfunction should converge to a standard Gaussian in the semi-classical limit. By this we mean that, if $\phi_\lambda$ denotes a normalized eigenfunction associated with the eigenvalue $\lambda$ and $X$ is a uniformly distributed point in $M$, then the random variable $\phi_\lambda(X)$ converges in distribution toward a standard Gaussian variable as $\lambda \to +\infty$. In the same spirit, for any $p \in \N$, the $p$-th central moment $\int_M \phi_\lambda(x)^p \dx x$ of this random variable ought to converge to the corresponding moment of the standard Gaussian distribution as $\lambda\to+\infty$. Here the integral is taken with respect to the normalized volume measure $\dx x$ induced by the Riemannian metric $g$. In this paper, we are concerned with this question of the asymptotic Gaussianity of the moments of normalized eigenfunctions for a model system: a point scatterer on a flat rectangular torus. Note that the simplest example of quantum chaotic system is the Laplacian on a surface of negative curvature. However, even for hyperbolic surfaces, the eigenfunctions of $\Delta$ are not explicit enough to compute with, and thus do not provide a good model system for this problem.

In the following, we consider the Laplacian perturbed by a Dirac mass on a flat rectangular torus. This kind of operator is usually referred to as a \emph{point scatterer}. The underlying classical system is a billiard on the torus, with an obstacle at the point where we put the Dirac mass. Its dynamics is close to integrable, since only a negligeable part of the trajectories is affected by the obstacle. However, the eigenfunctions of the point scatterer may display features typical of quantum chaotic systems, depending on the shape of the torus. For example, it is known that point scatterers on the square torus are quantum ergodic~\cite{KU2014}, while on Diophantine tori their semi-classical measures are scarred along a finite number of directions in phase space~\cite{KU2017}. The eigenfunctions of these point scatterers being reasonably explicit, they are good model systems to investigate quantum chaos.

Point scatterers were first introduced by Petr \u{S}eba~\cite{Seb1990}, who formulated several conjectures concerning the statistical properties of their spectra and the associated eigenfunctions. Since then, they have become a popular model to investigate the transition between chaotic and integrable quantum systems and have been studied in numerous papers~\cite{KU2014,KU2017,KU2019,RU2012,Yes2013,Yes2015,Yes2020}. In particular, \u{S}eba carried out a numerical investigation of the value distribution of the eigenfunctions on irrational tori, which seemed to suggest that it converged to a Gaussian in the semi-classical limit. Later, Keating, Marklof and Winn~\cite{KMW2003} argued that this value distribution ought to be non-Gaussian, by comparing this problem with a similar one for quantum star graphs. To the best of our knowledge, the only rigorous result in this direction was obtained by Kurlberg and Ueberschär. In~\cite{KU2019}, they proved that the fourth moment of the eigenfunctions of a point scatterer on a Diophantine torus can not converge to $3$, that is the fourth moment of the standard Gaussian. On the square torus, the asymptotics for general moments of the eigenfunctions for point scatterers seem difficult to obtain. This problem is already very hard in the case of Laplace eigenfunctions, where Bombieri and Bourgain managed to prove the asymptotic Gaussianity of the moments (cf.~Thm.~14, Rem.~15 and Thm.~17 in~\cite{BB2015}). A generalization of their methods to point scatterers has so far proved unsuccessful.

The purpose of the present article is to introduce a probabilistic model for the moments of the eigenfunctions of a point scatterer on a rectangular flat torus, which allows to compute the semi-classical limit of these moments. Our model and the results we obtain are described in Sect.~\ref{subsec random model and results} below. One of the key ideas in the definition of this random model is to replace the sequence of distinct Laplace eigenvalues on the ambient torus by the values of a Poisson point process. This idea originates in the paper~\cite{BT1977} by Berry and Tabor, in which they conjectured that these eigenvalues should behave like a Poisson point process. More precisely, they studied the empirical distribution of the spacings between consecutive Laplace eigenvalues, seen as a statistical series. They argued that it should converge weakly toward an exponential distribution, which is what we expect for a Poissonian sequence. See the survey~\cite{Mar2001} for more details on this Berry--Tabor Conjecture.


\subsection{Random model and main results}
\label{subsec random model and results}

Let us now describe the random model we are interested in and state our main results. We consider a point scatterer on a $2$-dimensional rectangular flat torus of unit area $\T_\alpha = \R^2 / \left(\alpha \Z \oplus \alpha^{-1} \Z\right)$, where $\alpha >0$. A point scatterer is a self-adjoint unbounded operator on $L^2(\T_\alpha)$ that can be intuitively thought of as ``$\Delta+\delta_0$'', where $\Delta$ is the usual Laplacian and $\delta_0$ is a Dirac potential, say at $0 \in \T_\alpha$. More details about the ambient space we consider and the definition of this operator are given in Sect.~\ref{subsec point scatterers on flat tori} below.

Non-zero eigenvalues of $\Delta$ remain eigenvalues of the point scatterer. In addition to these ``old'' eigenvalues, the point scatterer admits a sequence of ``new'' simple real eigenvalues. Let $\tau \in \R$ denote one of these new eigenvalues. An eigenfunction of the point scatterer associated with $\tau$ is:
\begin{equation}
\label{eq def G tau}
G_\tau : x \longmapsto -\frac{1}{\tau} + \sum_{k \geq 1} \frac{1}{\lambda_k - \tau} \sum_{\xi \in \Lambda_k} e^{2i\pi \prsc{\xi}{\cdot}},
\end{equation}
where $\prsc{\cdot}{\cdot}$ is the canonical Euclidean inner product in $\R^2$, $(\lambda_k)_{k \geq 1}$ is the sequence of distinct positive eigenvalues of $\Delta$ and, for each $k \geq 1$, $\Lambda_k = \frac{\sqrt{\lambda_k}}{2\pi} \S^1 \cap (\alpha^{-1}\Z \oplus \alpha \Z) \subset \R^2$ is the set of wave vectors associated with~$\lambda_k$. Note that $\Lambda_k$ is a finite set which is invariant under reflection with respect to the coordinate axes. Hence the function $G_\tau$ is well-defined and real-valued.

Since $\T_\alpha$ has unit area, $(\T_\alpha,\dx x)$ is a probability space, where $\dx x$ stands for the Lebesgue measure inherited from $\R^2$. We denote by $X$ a uniform random variable in $\T_\alpha$ (i.e.~the distribution of $X$ equals $\dx x$). We are interested in the central moments of the random variable $G_\tau(X)$. Its expectation is $\int_{\T_\alpha} G_\tau(x) \dx x = -\frac{1}{\tau}$. Denoting by $f_\tau = G_\tau +\frac{1}{\tau}$ and by
\begin{equation}
\label{eq moments G tau}
M^p_\tau = \int_{\T_\alpha} f_\tau(x)^p \dx x,
\end{equation}
the $p$-th normalized central moment of $G_\tau(X)$ equals $\left(M^2_\tau\right)^{-\frac{p}{2}}M^p_\tau$. It was conjectured by \u{S}eba~\cite{Seb1990} that, for every integer $p\geq 1$, this quantity converges as $\tau \to +\infty$ to $\mu_p$, the $p$-th moment of a standard Gaussian variable. In the following, we refer to $M^p_\tau$ as the $p$-th moment of~$f_\tau$. In this paper, we do not tackle this problem of computing the asymptotics of the deterministic normalized moments $\left(M^2_\tau\right)^{-\frac{p}{2}} M^p_\tau$. Instead, we consider the same question for a random model.

\begin{rem}
\label{rem notation expectation vs integral}
We do not use the classical notation $\esp{f_\tau(X)^p}$ for the $p$-th moment of $f_\tau$ and use integral notations instead. The reason is that we wish to avoid any confusion between the randomness coming from the choice of the random point $X$ (that we do not consider, except in the definition of $M^p_\tau$ above) and the randomness coming from the random model introduced below.
\end{rem}

As one can see from the definition, for any $\tau$, the function $f_\tau$ only depends on the positive spectrum $(\lambda_k)_{k \geq 1}$ of the usual Laplacian on $\T_\alpha$, and the associated wave vectors sets $(\Lambda_k)_{k \geq 1}$. In view of the Berry--Tabor Conjecture~\cite{BT1977}, a natural random model for the new centered eigenfunction~$f_\tau$ would be the following.
\begin{enumerate}
\item Choose $(\lambda_k)_{k \geq 1}$ to be the values of a Poisson point process on $[0,+\infty)$.
\item For each $k \geq 1$, define $\Lambda_k \cap [0,+\infty)^2$ as a random subset of $\frac{\sqrt{\lambda_k}}{2\pi} \S^1 \cap [0,+\infty)^2$, and define $\Lambda_k$ as the closure of this set under reflection with respect to the coordinate axes. For example, one could choose the directions of the points of $\Lambda_k \cap [0,+\infty)^2$ to be globally independent, independent of $(\lambda_k)_{k \geq 1}$, and uniformly distributed in $[0,\frac{\pi}{2}]$.
\item Define $f_\tau: x \longmapsto \displaystyle\sum_{k \geq 1} \frac{1}{\lambda_k - \tau} \sum_{\xi \in \Lambda_k} e^{2i\pi \prsc{\xi}{\cdot}}$.
\end{enumerate}
Then, one could consider the normalized moments of this random $f_\tau$ and study their asymptotics as $\tau \to +\infty$. Of course, one would need to make reasonable choices for the intensity of the Poisson point process $(\lambda_k)_{k \geq 1}$ and for the multiplicities of these mock eigenvalues (i.e.~the cardinalities of the random sets $(\Lambda_k)_{k \geq 1}$).

Unfortunately, this model is a bit too naive and has at least one major flaw. In the deterministic case, the $(\Lambda_k)_{k \geq 1}$ are subsets of $\alpha^{-1} \Z \oplus \alpha \Z$, which ensures that the complex exponentials appearing in the definition of $f_\tau$ (resp.~$G_\tau$) have the right periodicity and define functions on the torus $\T_\alpha$. In the model we just described, no such condition is required of the random sets $(\Lambda_k)_{k \geq 1}$. In particular, the complex exponentials appearing in the definition of the randomized $f_\tau$ do not define functions on $\T_\alpha$. Thus, $f_\tau$ is ill-defined as a function from $\T_\alpha$ to $\R$, and it makes no sense to consider its moments.

\begin{rem}
\label{rem rectangle}
One might be tempted to save this model by working in a rectangle of $\R^2$ instead of a torus. This allows to make sense of the randomized $f_\tau$ and its moments. However, when computing $M^p_\tau$ in the deterministic case, a lot of terms vanish because of the periodicity of the complex exponentials. These cancellations do not take place for the moments of the randomized~$f_\tau$. Because of this, their computation quickly becomes impractical.
\end{rem}

The random model we consider in this paper is built using the same basic ideas as the naive model we just discussed. In order to make these ideas work, we need to randomize the problem in steps. In the remainder of this section, we describe these steps and how they lead to our main result (Thm.~\ref{thm main} below). We discuss the various assumptions of our model in Sect.~\ref{subsec discussion of the random model}.

\paragraph{Step 1: deterministic expression of the moments.}

We first consider the moments of the deterministic functions $f_\tau$. In this first step, no randomness is involved. We still denote by $(\lambda_k)_{k \geq 0}$ the increasing sequence of distinct eigenvalues of $\Delta$ on $\T_\alpha$, and by $(\Lambda_k)_{k \geq 0}$ the associated wave vectors sets. For any $\tau \in \R \setminus \Sp(\Delta)$, we define as above $f_\tau = G_\tau + \frac{1}{\tau}$, where $G_\tau$ is the function defined by Eq.~\eqref{eq def G tau}. For any $p \in \N^*$, we denote by $M^p_\tau$ the $p$-th moment of $f_\tau$, as in~Eq.~\eqref{eq moments G tau}. Note that we denote by $\N$ the set of non-negative integers and by $\N^* = \N \setminus \{0\}$, the set of positive integers. Let us introduce some additional notations.

\begin{dfn}
\label{def l1 sequence}
We denote by $\ell_0 = \left\{(a_k)_{k \geq 1} \in \N^{\N^*} \mvert \sum_{k \geq 1} a_k < +\infty \right\}$ the set of sequences of non-negative integers with finite support, indexed by $\N^*$. For any $a=(a_k)_{k \geq 1} \in \ell_0$, we denote by $a! = \prod_{k \geq 1} a_k !$ and by $\norm{a} = \sum_{k \geq 1} a_k$.
\end{dfn}

\begin{dfn}
\label{def N a}
Let $a=(a_k)_{k \geq 1} \in \ell_0$, we denote by:
\begin{equation*}
N_a = \card\left\{(\xi_{k,l})_{k \geq 1; 1 \leq l \leq a_k} \in \prod_{k \geq 1} (\Lambda_k)^{a_k} \mvert
\sum_{k \geq 1} \sum_{l=1}^{a_k} \xi_{k,l}=0 \right\}.
\end{equation*}
Here, we use the convention that $\left(\Lambda_k\right)^0 = \{0\}$ and $\prod_{k \geq 1} (\Lambda_k)^{a_k}$ is canonically identified with the finite product $\prod_{\{k \geq 1 \mid a_k \neq 0\}} (\Lambda_k)^{a_k}$. Thus, if $\mathbf{0}$ denotes the zero sequence, we have $N_\mathbf{0} = 1$.
\end{dfn}

The coefficient $N_a$ counts the number of $\norm{a}$-tuples of lattice points in $\alpha^{-1}\Z \oplus \alpha\Z$ summing up to $0$, with some additional constraints. In particular, the sequence $(N_a)_{a \in \ell_0}$ heavily depends on the number-theoretical properties of the aspect ratio $\alpha^2$. With these notations in mind, we have the following deterministic expression of the moments, which is proved in Sect.~\ref{subsec moments of the new eigenfunctions}.

\begin{prop}
\label{prop deterministic expression M p tau}
Let $(\lambda_k)_{k \geq 1}$ denote the increasing sequence of distinct positive Laplace eigenvalues on $\T_\alpha$. Let $p \in \N^*$, let $\tau \in \R \setminus \Sp(\Delta)$ and let $M^p_\tau$ be defined by Eq.~\eqref{eq moments G tau}. We have:
\begin{equation}
\label{eq prop deterministic expression M p tau}
M^p_\tau = p! \sum_{a \in \ell_0; \norm{a}=p} \frac{N_a}{a!} \prod_{k \geq 1} \left(\frac{1}{\lambda_k -\tau}\right)^{a_k}.
\end{equation}
\end{prop}

For any $p \in \N^*$ and $\tau \in \R\setminus \Sp(\Delta)$, the expression of $M^p_\tau$ only depends on the sequence $(\lambda_k)_{k \geq 1}$ of positive eigenvalues of $\Delta$, and on the wave vectors sets $(\Lambda_k)_{k \geq 1}$ through the coefficients $(N_a)_{a \in \ell_0}$. It is this expression that we randomize in the subsequent steps. In particular, we define a random model for the moments $M^p_\tau$, but these randomized moments are not defined as the moments of some random function on $\T_\alpha$.

\paragraph{Step 2: randomization of the wave vectors.} In this second step, we start to introduce some randomness by randomizing the wave vectors sets $(\Lambda_k)_{k \geq 1}$. We assume that we are given an increasing sequence $(\lambda_k)_{k \geq 1}$ of positive numbers and a sequence $(m_k)_{k \geq 1}$ of positive integers. For any $k \in \N^*$, we want to define $\Lambda_k$ as a random finite subset of $\frac{\sqrt{\lambda_k}}{2\pi}\S^1 \subset \R^2$, which is invariant by reflection with respect to the coordinate axes, and such that $\card\left(\Lambda_k \cap [0,+\infty)^2\right)=m_k$. This amounts to choosing random directions for the $m_k$ elements of $\Lambda_k \cap [0,+\infty)^2$.

Here, the sequence $(\lambda_k)_{k \geq 1}$ can be the positive spectrum of the Laplacian on some rectangular flat torus, or a realization of a Poisson point process on $[0,+\infty)$, or anything else. Similarly, the sequence $(m_k)_{k \geq 1}$ can be the one given by the deterministic wave vectors sets we want to model, or not.

\begin{dfn}
\label{def eta}
Let $\eta$ denote the product of the uniform probability measures on $[0,\frac{\pi}{2}]^{\N^* \times \N^*}$, that is the distribution of a sequence of independent uniform variables in $[0,\frac{\pi}{2}]$, indexed by $\N^* \times \N^*$.
\end{dfn}

\begin{dfn}
\label{def randomized wave vectors}
Let $(\lambda_k)_{k \geq 1}$ be an increasing sequence of positive numbers, let $(m_k)_{k \geq 1}$ be a sequence of positive integers and let $(\theta_{k,j})_{k,j \geq 1}$ be a sequence of random variables in $[0,\frac{\pi}{2}]$ whose distribution is absolutely continuous with respect to $\eta$. Then, for any $k\geq 1$, we set:
\begin{equation*}
\Lambda_k = \frac{\sqrt{\lambda_k}}{2\pi} \left\{\zeta^{(i)}(\theta_{k,j}) \mvert 1\leq j \leq m_k, 1 \leq i \leq 4 \right\},
\end{equation*}
where we denoted, for every $\theta \in [0,\frac{\pi}{2}]$:
\begin{align*}
\zeta^{(1)}(\theta) &= (\cos(\theta),\sin(\theta)) = - \zeta^{(3)}(\theta) & &\text{and} & \zeta^{(2)}(\theta) &= (-\cos(\theta),\sin(\theta))= -\zeta^{(4)}(\theta).
\end{align*}
\end{dfn}

Working with the randomized $(\Lambda_k)_{k \geq 1}$ introduced in Def.~\ref{def randomized wave vectors}, we define a random sequence $(N_a)_{a \in \ell_0}$ as in Def.~\ref{def N a}. Similarly, for any $p \in \N^*$ and $\tau \in \R \setminus \left\{\lambda_k \mvert k \geq 1\right\}$, we define the randomized moment $M^p_\tau$ by the formula appearing in Prop.~\ref{prop deterministic expression M p tau}, see Eq.~\eqref{eq prop deterministic expression M p tau}. Our choice of probability distribution allows us to derive an almost sure expression for $(N_a)_{a \in \ell_0}$, see Lem.~\ref{lem computation Na}. This expression is purely combinatorial and only depends on the sequence $(m_k)_{k \geq 1}$. In particular, there is no longer any number theory involved in the problem at this stage. From the almost sure expression of $(N_a)_{a \in \ell_0}$, we deduce an almost sure expression of the randomized moments $M^p_\tau$. In order to give a precise statement, we will need the following definitions.

\begin{dfn}
\label{def partition}
Let $p \in \N$, we denote by $\mathcal{P}(p)= \left\{ (\alpha_k)_{k \geq 1} \in \ell_0 \mvert \sum_{k \geq 1} k \alpha_k  = p\right\}$ the set of \emph{partitions} of $p$.
\end{dfn}

\begin{dfn}
\label{def P and Q}
Let $p \in \N^*$, we denote by $P_p$ and $Q_p$ the following polynomials in $p$ variables:
\begin{align*}
P_p(X_1,\dots,X_p) &= \sum_{\alpha \in \mathcal{P}(p)} (-1)^{p - \norm{\alpha}} \frac{p!}{\alpha!} \prod_{q\geq 1}\left(X_q\right)^{\alpha_q},\\
Q_p(X_1,\dots,X_p) &= \sum_{\alpha \in \mathcal{P}(p)} \frac{(-1)^{p - \norm{\alpha}}}{\norm{\alpha}} \frac{\norm{\alpha}!}{\alpha!} \prod_{q\geq 1}\left(\frac{X_q}{q!}\right)^{\alpha_q}.
\end{align*}
\end{dfn}

\begin{rem}
\label{rem P and Q}
The set $\mathcal{P}(p)$ is finite. Moreover, if $\alpha = (\alpha_q)_{q \geq 1} \in \mathcal{P}(p)$, then for all $q >p$, we have $\alpha_q=0$. Hence, we can index the products appearing in Def.~\ref{def P and Q} by $1 \leq q \leq p$ instead of $q \geq 1$. Thus, $P_p$ and $Q_p$ are indeed polynomials and their total degree is at most $p$. In fact, for all $p \geq 1$, the total degree of $P_p$ and $Q_p$ equals $p$.
\end{rem}

\begin{dfn}
\label{def Ap}
For any $p \in \N^*$, we denote by $A_p = Q_p\left(1,\frac{1}{2!},\dots,\frac{1}{p!}\right)$.
\end{dfn}

\begin{dfn}
\label{def S p lambda}
Let $(\lambda_k)_{k \geq 1}$ be an increasing sequence of positive numbers and let $(m_k)_{k \geq 1}$ be a sequence of non-negative numbers. For any $p \in \N^*$ and $\tau \in \R$, we denote by $S^p_\tau$ the \emph{spectral sum}:
\begin{equation*}
S^p_\tau = \sum_{k \geq 1} \frac{m_k}{(\lambda_k-\tau)^{2p}},
\end{equation*}
where the sum makes sense in $[0,+\infty]$ as a sum of non-negative (possibly infinite) terms.
\end{dfn}

\begin{prop}
\label{prop as expression of M p tau}
Let $(\lambda_k)_{k \geq 1}$ be an increasing sequence of positive numbers, let $(m_k)_{k \geq 1}$ be a sequence of positive integers. Let $(\Lambda_k)_{k \geq 1}$ be defined by Def.~\ref{def randomized wave vectors} and let the corresponding randomized moments $(M^p_\tau)_{p \geq 1}$ be defined by Eq.~\eqref{eq prop deterministic expression M p tau}. Then, almost surely, the following holds for all $p \in \N^*$ and $\tau \in \R \setminus \left\{ \lambda_k \mvert k \geq 1 \right\}$.
\begin{enumerate}
\item \label{point 1} We have $M^{2p-1}_\tau = 0$.
\item \label{point 2} If $S^q_\tau < +\infty$ for all $q \in \{1,\dots,p\}$, then we have:
\begin{equation}
\label{eq prop as expression of M p tau}
M^{2p}_\tau = \frac{(2p)!}{p!} P_p(2 A_1 S^1_\tau,\dots,2 A_p S^p_\tau),
\end{equation}
where the polynomial  $P_p$ and the coefficients $A_1,\dots,A_p$ are defined respectively by Def.~\ref{def P and Q} and~\ref{def Ap}. In particular, $M^{2p}_\tau < +\infty$.
\end{enumerate}
\end{prop}

Prop.~\ref{prop as expression of M p tau} shows that, in our model, the odd moments vanish almost surely. This is consistent with the idea that the moments should be asymptotically those of a standard Gaussian. In the following, we only consider even moments. Almost surely, the moment of order $2p$ is the value of some fixed polynomial in $p$ variables, whose coefficients only depend on $p$, evaluated at $(S^1_\tau,\dots,S^p_\tau)$. Thus, $M^{2p}_\tau$ only depends on $\tau$, $(\lambda_k)_{k \geq 1}$ and $(m_k)_{k \geq 1}$, and this only through the spectral sums $(S^q_\tau)_{1 \leq q \leq p}$.

\begin{rem}
\label{rem square torus}
In the case of the square torus ($\alpha=1$), the deterministic wave vectors sets $(\Lambda_k)_{k \geq 1}$ are also invariant under $(x_1,x_2) \mapsto (x_2,x_1)$. In Sect.~\ref{subsubsec case of the square torus}, we define a variation on our random model such that the randomized $(\Lambda_k)_{k \geq 1}$ also have this additional symmetry (see Def.~\ref{def randomized wave vectors bis}). It turns out that the previous results can be adapted to this alternative model. In particular, the conclusion of Prop.~\ref{prop as expression of M p tau} remains valid.
\end{rem}

\paragraph{Step 3: randomization of the spectrum of the Laplacian.}

In this final step, we start from the expressions of the moments derived in Prop.~\ref{prop as expression of M p tau}, and we randomize the sequences $(\lambda_k)_{k \geq 1}$ and $(m_k)_{k \geq 1}$ in these expressions. Inspired by the Berry--Tabor Conjecture, we replace $(\lambda_k)_{k \geq 1}$ by the values of a Poisson point process on $[0,+\infty)$ (see Sect.~\ref{sec reminder on Poisson point processes} for a quick reminder about these processes). Since it is meant to model the spectrum of the Laplacian on some~$\T_\alpha$, we tune the intensity of this point process so that it satisfies the Weyl Law in the mean and to leading order. Let us recall this classical result.

\begin{thm}[Weyl Law]
\label{thm Weyl Law}
On $\T_\alpha$, let $\mathcal{N}(\lambda)$ denote the number of Laplace eigenvalues smaller than or equal to $\lambda$, counted with multiplicities. Then, we have $\mathcal{N}(\lambda)  = \frac{\lambda}{4\pi} +O(\sqrt{\lambda})$ as $\lambda \to +\infty$.
\end{thm}

Note that, with the random model for the wave vectors introduced in Step~2, what plays the role of the multiplicity of $\lambda_k$ is almost surely $4m_k$. Hence, the intensity of our Poisson process must somehow be related to the sequence $(m_k)_{k \geq 1}$.

\begin{dfn}
\label{def multiplicity function}
Let $m:[0,+\infty) \to [1,+\infty)$ be a function of class $\mathcal{C}^1$. We say that $m$ is a \emph{multiplicity function} if there exists $\beta >0$ such that $m'(t) = O(t^{-\beta})$ as $t \to +\infty$.
\end{dfn}

\begin{dfn}
\label{def nu m}
Let $m$ be a multiplicity function, we denote by $\nu_m$ the measure on $[0,+\infty)$ admitting the density $(16\pi m)^{-1}$ with respect to the Lebesgue measure.
\end{dfn}

\begin{dfn}
\label{def random spectrum}
Let $m$ be a multiplicity function. We define the random sequence $(\lambda_k)_{k \geq 1}$ as the increasing sequence of the values of a Poisson point process on $[0,+\infty)$ of intensity~$\nu_m$. Moreover, for all $k \geq 1$, we set $m_k = m(\lambda_k)$.
\end{dfn}

It is a classical fact that the sequence $(\lambda_k)_{k \geq 1}$ defined by Def.~\ref{def random spectrum} is almost surely a sequence of positive numbers that diverges to $+\infty$ as $k \to +\infty$. For the reader's convenience, we included a proof of this fact in App.~\ref{sec reminder on Poisson point processes}, see Lem.~\ref{lem random spectrum}. Note that, as a consequence of Def.~\ref{def multiplicity function} and~\ref{def random spectrum}, almost surely the sequence $(m_k)_{k \geq 1}$ is no longer integer-valued at this stage. We will comment further on this issue in Sect.~\ref{subsec discussion of the random model}, where we discuss the features of our random model.

Thanks to Prop.~\ref{prop as expression of M p tau}.\ref{point 1}, we are now only interested in moments of even order. Let $m$ be a multiplicity function and let $\tau \in \R$. We sample random sequences $(\lambda_k)_{k \geq 1}$ and $(m_k)_{k \geq 1}$ as in Def.~\ref{def random spectrum}. Then, for all $p \in \N^*$, we define the randomized spectral sum~$S^p_\tau$ as in Def.~\ref{def S p lambda} and the randomized even moment~$M^{2p}_\tau$ by Eq.~\eqref{eq prop as expression of M p tau}. It turns out that these random variables are almost surely well-defined. The proof of the following result is given at the end of Sect.~\ref{subsec randomization of the spectrum of the Laplacian}.

\begin{lem}
\label{lem as def Sq}
Let $\tau \in \R$. Let $m$ be a multiplicity function, let $(\lambda_k)_{k \geq 1}$ be the associated randomized spectrum, and let $m_k = m(\lambda_k)$ for all $k \geq 1$. Then, almost surely, for all $p \geq 1$, the spectral sum $S^p_\tau$ is finite, and the randomized moment $M^{2p}_\tau$ is well-defined and finite.
\end{lem}

\begin{rem}
\label{rem order of sampling}
There is a small subtlety here. If we sample the sequence $(\lambda_k)_{k \geq 1}$ once and work with the same realization for all $\tau \in \R$ then, whenever $\tau$ is equal to one of the $\lambda_k$, the random variables $(S^p_\tau)_{p \geq 1}$ diverge and the $(M^{2p}_\tau)_{p \geq 1}$ are ill-defined. Here, we work with a fixed $\tau$, for which the random variables we are interested in are almost surely well-defined by Lem.~\ref{lem as def Sq}. In particular, we can consider the distribution of these random variables for any $\tau$, and study their limit in distribution as $\tau \to +\infty$.
\end{rem}

We can now state our main result, which describes the asymptotic joint distribution of the even normalized moments $(M^2_\tau)^{-p} M^{2p}_\tau$ for $p \geq 2$, as $\tau \to +\infty$. Note that for $p=1$, for all $\tau \in  \R$, we have $(M^2_\tau)^{-p} M^{2p}_\tau = 1$ deterministically.

\begin{dfn}
\label{def mu p}
We denote by $(\mu_p)_{p \geq 0}$ the sequence of moments of a standard real Gaussian variable. Recall that for any $p \in \N$, we have $\mu_{2p} = \frac{(2p)!}{2^pp!}$ and $\mu_{2p+1}=0$.
\end{dfn}

\begin{thm}
\label{thm main}
Let $m$ be a multiplicity function and let $(\lambda_k)_{k \geq 1}$ and $(m_k)_{k \geq 1}$ be associated random sequences as in Def.~\ref{def random spectrum}. For any $\tau \in \R$, let the randomized even moments $(M^{2p}_\tau)_{p \geq 1}$ be defined by Eq.~\eqref{eq prop as expression of M p tau}.

For any integer $p \geq 2$, there exists a one-parameter family $\left(R_2(l),\dots,R_p(l)\right)_{l \in (0,+\infty]}$ of random vectors in $\R^{p-1}$ such that, if $(\tau_n)_{n \geq 0}$ is a sequence of real numbers satisfying $\tau_n \xrightarrow[n \to +\infty]{} +\infty$ and $m(\tau_n) \xrightarrow[n \to +\infty]{} l$, then the following holds in distribution:
\begin{equation*}
\left(\frac{M^4_{\tau_n}}{(M^2_{\tau_n})^2},\dots,\frac{M^{2p}_{\tau_n}}{(M^2_{\tau_n})^p}\right) \xrightarrow[n \to +\infty]{} \left(\mu_4 R_2(l),\dots,\mu_{2p}R_p(l)\right).
\end{equation*}
Moreover, the distribution of $\left(R_2(l),\dots,R_p(l)\right)$ only depends on $p$ and $l$ and satisfies the following.
\begin{itemize}
\item If $l=+\infty$, then we have $\left(R_2(l),\dots,R_p(l)\right)=(1,\dots,1)$ almost surely, so that the previous convergence holds in probability.
\item If $l \in (0,+\infty)$, then $\left(R_2(l),\dots,R_p(l)\right)$ admits a smooth density with respect to the Lebesgue measure of $\R^{p-1}$.
\item The random vectors $\left(R_2(l),\dots,R_p(l)\right)$ and $\left(R_2(l'),\dots,R_p(l')\right)$ are equal in distribution if and only if $l=l'$.
\end{itemize}
\end{thm}

\begin{rem}
\label{rem limit distribution main thm}
If $2 \leq q \leq p$ then, for any $l \in (0,+\infty]$, the random vector $\left(R_2(l),\dots,R_q(l)\right)$ is distributed as the first $q-1$ components of $\left(R_2(l),\dots,R_p(l)\right)$. In particular the distribution of $\left(R_2(l),\dots,R_q(l)\right)$ is uniquely defined and does not depend on $p \geq q$.
\end{rem}

\begin{rem}
\label{rem main thm CV tau}
The same proof as that of Thm.~\ref{thm main} shows that if $m(\tau) \xrightarrow[\tau \to +\infty]{} l$ then we have:
\begin{equation*}
\left(\frac{M^4_\tau}{(M^2_\tau)^2},\dots,\frac{M^{2p}_\tau}{(M^2_\tau)^p}\right) \xrightarrow[\tau \to +\infty]{} \left(\mu_4 R_2(l),\dots,\mu_{2p}R_p(l)\right)
\end{equation*}
in distribution, see Sect.~\ref{subsec proof of the main theorem}, and especially the proof of Lem.~\ref{lem proof of CV in distribution}.
\end{rem}

Thm.~\ref{thm main} is consistent with the deterministic result of Kurlberg and Ueberschär~\cite{KU2019}. They proved that on Diophantine tori, the $4$-th normalized moment of the new eigenfunctions of a point scatterer is bounded away from $\mu_4$ along a full-density subsequence of new eigenvalues. Our result suggests that, in the deterministic setting, the asymptotic behavior of the normalized moments of the new eigenfunctions mostly depends on the asymptotics of the multiplicities of the eigenvalues of the Laplacian.

We conclude this section by stating one last result. Let $m$ be a multiplicity function (see Def.~\ref{def multiplicity function}) and let $(\lambda_k)_{k \geq 1}$ and $(m_k)_{k \geq 1}$ be defined as in Def.~\ref{def random spectrum}. For any $\lambda \geq 0$, we denote by: 
\begin{equation}
\label{eq def Nm lambda}
\mathcal{N}_m(\lambda) = 1 + \sum_{\{k \geq 1 \mid \lambda_k \leq \lambda\}} 4m_k.
\end{equation}
The function $\mathcal{N}_m$ is the analogue in our random setting of the counting function $\mathcal{N}$ appearing in the Weyl Law (Thm.~\ref{thm Weyl Law}). We tuned the intensity $\nu_m$ (see Def.~\ref{def nu m}) of the Poisson point process $(\lambda_k)_{k \geq 1}$ so that $\esp{\mathcal{N}_m(\lambda)} \sim \frac{\lambda}{4\pi}$ as $\lambda \to +\infty$, in agreement with the Weyl Law. In fact, we have a much stronger result.

\begin{prop}[Random Weyl Law]
\label{prop random Weyl Law}
Let $m$ be a multiplicity function and let $\mathcal{N}_m$ be defined by Eq.~\eqref{eq def Nm lambda}. For all $\lambda \geq 0$, we have $\esp{\mathcal{N}_m(\lambda)} = 1 + \frac{\lambda}{4\pi}$ and $\var{\mathcal{N}_m(\lambda)} = \frac{1}{\pi} \int_0^\lambda m(t) \dx t$. Moreover, the following holds.
\begin{enumerate}
\item (Strong Law of Large Numbers) Almost surely $\displaystyle\frac{1}{\lambda}\mathcal{N}_m(\lambda) \xrightarrow[\lambda \to +\infty]{} \frac{1}{4\pi}$.
\item (Central Limit Theorem) As $\lambda \to +\infty$, the random variable $\displaystyle\frac{\mathcal{N}_m(\lambda) - \esp{\mathcal{N}_m(\lambda)}}{\var{\mathcal{N}_m(\lambda)}^\frac{1}{2}}$ 
converges in distribution toward a standard real Gaussian.
\end{enumerate}
\end{prop}

\begin{rem}
\label{rem random Weyl Law}
In the usual Weyl Law, the error term is of order $O(\sqrt{\lambda})$. Here, the fluctuations around the leading term are of order $\left(\int_0^\lambda m(t) \dx t\right)^\frac{1}{2}$. Since the values of $m$ are larger than $1$, this term is greater than $\sqrt{\lambda}$. On the other hand, using Def.~\ref{def multiplicity function}, one can prove that this term is of order $O(\lambda^\frac{1+\alpha}{2})$ for some $\alpha \in [0,1)$. In fact, for any $\alpha \in [0,1)$, we can build a multiplicity function $m$ such that $\left(\int_0^\lambda m(t) \dx t\right)^\frac{1}{2} \sim \lambda^\frac{1+\alpha}{2}$ as $\lambda \to +\infty$.
\end{rem}


\subsection{Discussion of the random model}
\label{subsec discussion of the random model}

Let us discuss the various assumptions in our model. The first thing to recall is that we defined a random model for the moments of some function $f_\tau:\T_\alpha \to \R$, but the random moments $(M^p_\tau)$ that we study are not the moments of some randomized version of $f_\tau$. In fact, trying to define a random model for $f_\tau$ using the same ingredients as our model leads to some serious difficulties, as discussed in Sect.~\ref{subsec random model and results}.

In Step~1, we compute the deterministic moments of the function $f_\tau$. As one expects, in the course of this computation, a lot of terms vanish for periodicity reasons. This is a typical feature of the problem we wish to model, and it is naturally present in our random model since we start from the deterministic expression obtained at the end of Step~1, in Prop.~\ref{prop deterministic expression M p tau}.

In Step~2, we assume that the sequences $(\lambda_k)_{k \geq 1}$ and $(m_k)_{k \geq 1}$ are fixed, and we randomize the directions $(\theta_{k,j})$ of the wave vectors. We only ask that the distribution of these directions admits a density with respect to some natural measure. Though this assumption does not seem very restrictive, it is enough to kill all the number-theoretic aspects of the problem. On the one hand, this is what allows us to say something about our model. But on the other hand this means that our model is too crude to capture any of the arithmetic subtleties of the deterministic setting. Still, the almost sure result of Prop.~\ref{prop as expression of M p tau} is very robust, in the sense that it does not depend on the actual density of $(\theta_{k,j})$. This shows that there is a lot of flexibility in this random model, and we hope that it is enough for it to say something of the deterministic case.

Let us give some examples of admissible distributions for the sequence $(\theta_{k,j})_{k,j \geq 1}$. In view of Def.~\ref{def randomized wave vectors}, we are only interested in the distribution of the subsequence $(\theta_{k,j})_{k \geq 1; 1 \leq j \leq m_k}$, which must admit a density with respect to the product of the Lebesgue measures. In the following examples, we describe the joint distribution of the random variables $(\theta_{k,j})_{k \geq 1; 1 \leq j \leq m_k}$. One can build admissible distributions for $(\theta_{k,j})_{k,j \geq 1}$ from these, by choosing the remaining terms to be independent uniform variables in $[0,\frac{\pi}{2}]$.

\begin{ex}
\begin{enumerate}
\item The $(\theta_{k,j})_{k \geq 1; 1 \leq j \leq m_k}$ are independent uniform variables in $[0,\frac{\pi}{2}]$. The distribution of the corresponding sequence $(\theta_{k,j})_{k,j \geq 1}$ is then $\eta$.

\item Let us assume that $(\lambda_k)_{k \geq 1}$ is the spectrum of $\Delta$ on some $\T_\alpha$, and that the associated wave vectors are known, up to some error. Let us say, for example, that the deterministic $\Lambda_k$ is of the form $\frac{\sqrt{\lambda_k}}{2\pi}\left\{\zeta^{(i)}(\beta_{k,j}) \mvert 1\leq j \leq m_k, 1 \leq i \leq 4 \right\}$, where we know that $\beta_{k,j}$ belongs to some interval~$I_{k,j} \subset [0,\frac{\pi}{2}]$, for any $k \geq 1$ and $j \in \{1,\dots,m_k\}$. In this case, we can define $(\theta_{k,j})_{k \geq 1; 1 \leq j \leq m_k}$ as a sequence of independent variables admitting a density with respect to the Lebesgue measure on $[0,\frac{\pi}{2}]$, and such that $\theta_{k,j}$ is supported on $I_{k,j}$.

\item One can introduce some repulsion between the directions $(\theta_{k,j})_{k \geq 1; 1 \leq j \leq m_k}$. A naive example is to define the density of the distribution of $(\theta_{k,j})_{1 \leq j \leq m_k}$ with respect to the Lebesgue measure on $[0,\frac{\pi}{2}]^{m_k}$ as:
\begin{equation*}
Z_k \exp\left(-\sum_{1 \leq i < j \leq m_k} \frac{1}{\norm{\theta_{k,i}- \theta_{k,j}}^2}\right),
\end{equation*}
where $Z_k$ is some normalizing constant. Then take the product distribution over $k \geq 1$.

\item One can mix the previous two examples and build admissible distributions such that the $(\theta_{k,j})_{k \geq 1; 1 \leq j \leq m_k}$ are both localized and repel one another.
\end{enumerate}
\end{ex}

Note that the density of $(\theta_{k,j})$ can depend on the sequences $(\lambda_k)_{k \geq 1}$ and $(m_k)_{k \geq 1}$. Since we obtain an almost sure result independent of this density in Prop.~\ref{prop as expression of M p tau}, there is no coupling issue when we randomize $(\lambda_k)_{k \geq 1}$ and $(m_k)_{k \geq 1}$ in Step~3. Note also that our assumption that $(\theta_{k,j})$ admits a density with respect to $\eta$ (cf.~Def.~\ref{def eta}) is a little stronger than what we actually need. It is enough to assume that all the $p$-dimensional marginal distributions of $(\theta_{k,j})_{k \geq 1; 1 \leq j \leq m_k}$ admit a density with respect to Lebesgue to prove Prop.~\ref{prop as expression of M p tau} for the $p$-th moment.

In Step~3, we finally replace the sequence $(\lambda_k)_{k \geq 1}$ by the values of a Poisson point process. Since we want this sequence to model the spectrum of $\Delta$, it is natural to ask that it satisfies a version of the Weyl Law. As explained in Sect.~\ref{subsec random model and results}, this means that the definitions of the intensity measure of the point process and of the sequence $(m_k)_{k \geq 1}$ must be intertwined. We define $m_k$ as $m(\lambda_k)$ for all $k \geq 1$, where $m:[0,+\infty) \to [1,+\infty)$ is some fixed function. That is, the mock multiplicities $(m_k)_{k \geq 1}$ depend deterministically on the randomized spectrum $(\lambda_k)_{k \geq 1}$. This point is fundamental in our approach. It allows us to write the randomized spectral sums $S^p_\tau$ (see Def.~\ref{def S p lambda}), which are the relevant random variables to consider in this problem, as $\sum_{k \geq 1} g_p(\lambda_k)$ for some function $g_p$. Then we can study these quantities using  Campbell's Theorem (recalled in App.~\ref{sec reminder on Poisson point processes}, Thm.~\ref{thm Campbell}). The same tools also show that, with our definition of  $\nu_m$ (see Def.~\ref{def nu m}), the sequence $(\lambda_k)_{k \geq 1}$ satisfies a Random Weyl Law (see Prop.~\ref{prop random Weyl Law}).

Let us now discuss our choice of the multiplicity function $m$. One could think of taking $m$ to be measurable and integer-valued. Instead, we ask our multiplicity functions to be regular enough (cf.~Def.~\ref{def multiplicity function}), in order to use some analytic tools. This part of the model is probably the one that seems less natural, because we replace a sequence of integers by the values of a smooth function evaluated at random points. In particular, generically, none of the $(m(\lambda_k))_{k \geq 1}$ is an integer. This is not an issue, because Prop.~\ref{prop as expression of M p tau}  gives an almost sure expression of $M^p_\tau$, where the sequence $(m_k)_{k \geq 1}$ only appears through the spectral sums $(S^q_\tau)_{1 \leq q \leq p}$, and the definition of these sums (Def.~\ref{def S p lambda}) makes sense as soon as the $(m_k)_{k \geq 1}$ are non-negative. In fact, the expression of $S^q_\tau$ hints that what matters here is the asymptotic behavior of $(m_k)_{k \geq 1}$, or equivalently of $m$, rather than $(m_k)_{k \geq 1}$ being integer-valued.

The precise condition we ask our multiplicity function to satisfy (cf.~Def.~\ref{def multiplicity function}) may seem quite arbitrary. Indeed, to the best of our understanding, it is only a technical condition that allows our proofs to work. This condition can be weakened a bit. In most of this paper, we only need $m$ to be such that $m'(t) \xrightarrow[t \to +\infty]{} 0$. The only place where we need a stronger assumption is to prove the Law of Large Numbers part in Prop.~\ref{prop random Weyl Law}. There, we need to further assume that $m(t) = O(t^\alpha)$ as $t \to +\infty$, for some $\alpha <1$.

Let us give some examples of admissible multiplicity functions. The first two examples show that defining $(m_k)_{k \geq 1}$ as $(m(\lambda_k))_{k \geq 1}$ allows us to model the average behavior of the deterministic multiplicities of the Laplace eigenvalues on some flat tori.
\begin{ex}
\label{ex multiplicity function}
\begin{enumerate}
\item On an irrational torus, that is on $\T_\alpha$ with $\alpha^4 \notin \Q$, the multiplicity of a generic deterministic eigenvalue of $\Delta$ is $4$. This suggests to model this case by setting $m:t \mapsto 1$.

\item On the square torus $\T$, it is known that $\frac{1}{k}\sum_{i=1}^k r_i \sim 4C_0 \ln(\lambda_k)^\frac{1}{2}$ as $k \to +\infty$, where $(\lambda_k)_{k \geq 1}$ are the deterministic eigenvalues of the Laplacian, $(r_k)_{k \geq 1}$ are their multiplicities, and $C_0$ is some explicit constant. This follows from Landau's Theorem \cite{Lan1908}, which gives the counting asymptotics of the Laplace eigenvalues without multiplicities, and Weyl's Law (Thm.~\ref{thm Weyl Law}). This suggests to model this case using the multiplicity function $m:t \mapsto 1+C_0 \ln(1+t)^\frac{1}{2}$.

Note that $r_k=8$ for infinitely many $k \geq 1$. With our definition of multiplicity function, we can only model the average behavior of the multiplicities of the deterministic Laplace eigenvalues.

\item More generally, $t \mapsto 1+ C\ln(1+t)^\alpha$ is a multiplicity function if $C \geq 0$ and $\alpha \geq 0$, and $t \mapsto 1+Ct^\alpha$ is a multiplicity function if $C\geq 0$ and $0\leq \alpha < 1$.
\end{enumerate}
\end{ex}

We conclude this section with a last remark. In Step~1 and Step~2, the even moments $(M^{2p}_\tau)_{p \geq 1}$ are positive. In Step~3, we redefine these even moments using Eq.~\eqref{eq prop as expression of M p tau}. If the multiplicity function $m$ is constant, equal to some integer, then a computation similar to the proof of Prop.~\ref{prop as expression of M p tau} shows that: for all $\tau \in \R$, almost surely, for all $p \in \N^*$, $M^{2p}_\tau > 0$. Unfortunately, this kind of computation does not make sense if $m$ is not integer-valued. For a general $m$, it is not clear that the even randomized moments are almost surely positive. It is natural to conjecture that this is the case, but we do not know how to prove it at this point.


\subsection{Organization of the paper}
\label{subsec organization of the paper}

The paper is organized as follows. In Sect.~\ref{sec moments of deterministic eigenfunctions of point scatterers}, we study the deterministic situation. In Sect.~\ref{subsec point scatterers on flat tori}, we describe the new eigenfunctions of point scatterers on rectangular flat tori, then, in Sect.~\ref{subsec moments of the new eigenfunctions}, we study their moments. The main result of Sect.~\ref{sec moments of deterministic eigenfunctions of point scatterers} is the proof of Prop.~\ref{prop deterministic expression M p tau}.

In Sect.~\ref{sec random model for the moments of the new eigenfunctions}, we introduce and study our random model. Sect.~\ref{subsec randomization of the wave vectors} is concerned with what we called Step~2 in the Introduction (see Sect.~\ref{subsec random model and results}), that is the randomization of the wave vectors. In this section we prove Prop.~\ref{prop as expression of M p tau}, for our random model in Sect.~\ref{subsubsec almost sure expression of the moments} and for a variation adapted to the square torus in Sect.~\ref{subsubsec case of the square torus}. In Sect.~\ref{subsec randomization of the spectrum of the Laplacian} we conclude the definition of our random model as in Step~3 above. Then, we study the asymptotic distribution of the spectral sums (see Def.~\ref{def S p lambda}) in Sect.~\ref{subsec limit distribution of the spectral sums} and we prove Thm.~\ref{thm main} in Sect.~\ref{subsec proof of the main theorem}.

Sect.~\ref{sec random Weyl Law} is concerned with the proof of the Random Weyl Law (Prop.~\ref{prop random Weyl Law}). App.~\ref{sec reminder on Poisson point processes} surveys some basic facts about Poisson point processes, including Campbell's Theorem (Thm.~\ref{thm Campbell}), which is one of the key tools in the proof of our main results. Finally, we gathered the proofs of several technical lemmas in App.~\ref{sec technical lemmas}.

\paragraph{Acknowledgment.} The authors thank Stéphane Nonnenmacher for his remarks and comments, that helped to improve the exposition of the paper.


\section{Moments of deterministic eigenfunctions of point scatterers}
\label{sec moments of deterministic eigenfunctions of point scatterers}

The goal of this section is to describe precisely the deterministic setting that we wish to model. In Sect.~\ref{subsec point scatterers on flat tori}, we introduce our framework and recall the definition of a point scatterer on a flat torus. We also study the new eigenfunctions of point scatterers on flat tori and prove these functions are $L^p$ for any $p \in \N^*$. Then, we derive an expression for the moments of any order of these new eigenfunctions and we prove Prop.~\ref{prop deterministic expression M p tau} in Sect.~\ref{subsec moments of the new eigenfunctions}.


\subsection{Point scatterers on flat tori}
\label{subsec point scatterers on flat tori}

In this section, we introduce our framework and recall the definition of a point scatterer in this setting. Let $\alpha >0$, we denote by $\L_\alpha=\alpha \Z \oplus \alpha^{-1} \Z \subset \R^2$ and by $\L_\alpha^*= \alpha^{-1} \Z \oplus \alpha \Z$ the dual lattice. We denote by $\T_\alpha = \R^2 / \L_\alpha$, equipped with the metric induced by the Euclidean inner product on~$\R^2$. That is, $\T_\alpha$ is the $2$-dimensional rectangular flat torus with aspect ratio~$\alpha^2$ and total area equal to $1$. Here and in the rest of this paper, area is computed with respect to the Lebesgue measure on $\T_\alpha$, induced by the one on $\R^2$. When $\alpha =1$, we simply denote $\T = \T_1$ for the flat square torus. In the following, we constantly identify points of $\R^2$ with their projection in~$\T_\alpha$.

Let $\Delta$ denote the Laplace--Beltrami operator on $\T_\alpha$. Recall that $\Delta$ is an unbounded, positive, self-adjoint operator on the space $L^2(\T_\alpha)$ of square integrable functions from $\T_\alpha$ to $\C$. Note that, since we defined $\Delta$ as a positive operator, we have $\Delta = -\frac{\partial^2}{\partial x_1^2} -\frac{\partial^2}{\partial x_2^2}$ in the local coordinates induced by the usual Euclidean coordinates of $\R^2$.

Let $x_0 \in \T_\alpha$, and let us define rigorously what a point scatterer at $x_0$ is. Let $\Delta_{\vert D_0}$ denote the restriction of $\Delta$ to the space $D_0$ of smooth functions from $\T_\alpha$ to $\C$ that vanish in some neighborhood of $x_0$. By von~Neuman's theory of self-adjoint extensions, there exists a one-parameter family $(\Delta_\varphi)_{\varphi \in (-\pi,\pi]}$ of unbounded self-adjoint operators on $L^2(\T_\alpha)$ that extend $\Delta_{\vert D_0}$ (cf.~\cite{Zor1980}). For $\varphi=\pi$, we recover the usual Laplacian. Any element of the family $(\Delta_\varphi)_{\varphi \in (-\pi,\pi)}$ is called a \emph{point scatterer} at~$x_0$. The situation being translation-invariant, from now on we assume that $x_0 = 0$. For the purpose of the present paper, we do not need to go over the details of the definition of the operators $(\Delta_\varphi)_{\varphi \in (-\pi,\pi)}$. In particular, the choice of the parameter $\varphi \in (-\pi,\pi)$ will be irrelevant to us. We refer the interested reader to~\cite[Sect.~3]{RU2012} for details on this matter.

Let us review the spectral properties of the point scatterers $(\Delta_\varphi)_{\varphi \in (-\pi,\pi)}$. First, recall that the spectrum of the usual Laplacian is $\Sp(\Delta) = \left\{\frac{4\pi^2}{\alpha^2} (a^2 + \alpha^4b^2) \mvert (a,b) \in \N^2 \right\}$. Its distinct eigenvalues can be ordered into an increasing sequence $(\lambda_k)_{k \geq 0}$ going to infinity. Let $k \geq 0$, we denote by:
\begin{equation}
\label{eq def Lambda k}
\Lambda_k = \{\xi \in \L_\alpha^* \mid 4\pi^2\Norm{\xi}^2 = \lambda_k \} = \L_\alpha^* \cap \left(\frac{\sqrt{\lambda_k}}{2\pi}\S^1 \right)
\end{equation}
the set of \emph{wave vectors} associated with~$\lambda_k$, where $\S^1$ denotes the unit circle in $\R^2$. Note that $\Lambda_0 = \{0\}$ and $\L^*_\alpha = \bigsqcup_{k \geq 0} \Lambda_k$. The eigenspace $\ker(\Delta - \lambda_k \Id)$ is spanned by the functions $\left\{ e^{2i\pi \prsc{\xi}{\cdot}} \mvert \xi \in \Lambda_k \right\}$, where $\prsc{\cdot}{\cdot}$ stands for the canonical Euclidean inner product on $\R^2$, and the multiplicity of $\lambda_k$ is $r_k = \dim \ker\left(\Delta - \lambda_k \Id \right) = \card \Lambda_k$.

\begin{rem}
\label{rem multiplicities}
We have $r_0=1$. Besides, $\Lambda_k$ is invariant under the reflections $(x_1,x_2) \mapsto (-x_1,x_2)$ and $(x_1,x_2)\mapsto (x_1,-x_2)$. Hence $r_k$ is even if $k >0$, and $r_k \in 4 \N$ unless $\Lambda_k \cap (\R \times \{0\} \cup \{0\} \times \R) \neq \emptyset$. When $\alpha =1$, the set $\Lambda_k$ is moreover invariant by $(x_1,x_2) \mapsto (x_2,x_1)$, so that $r_k \in 4\N$ for any $k >0$ and $r_k \in 8\N$ generically.
\end{rem}

Let $\varphi \in (-\pi,\pi)$, in the remainder of this section, we describe the eigenvalues and eigenfunctions of $\Delta_\varphi$. See~\cite[Sect.~3 and App.~A]{RU2012} for more details. For any $k \geq 1$, $\lambda_k$ is an eigenvalue of $\Delta_\varphi$ with multiplicity $r_k -1$, and the associated eigenspace is:
\begin{equation*}
\ker\left(\Delta_\varphi - \lambda_k \Id\right) = \left\{ \phi \in \ker\left(\Delta - \lambda_k \Id\right) \mvert \phi(0) = 0 \right\}.
\end{equation*}
In addition to the old eigenvalues inherited from $\Delta$, there exists an increasing sequence $(\tau_k^\varphi)_{k \geq 0}$ of new simple real eigenvalues of $\Delta_\varphi$. That is $\Sp(\Delta_\varphi) = \left\{\lambda_k \mvert k \geq 1 \right\} \cup \left\{\tau_k^\varphi \mvert k \geq 0 \right\}$. Moreover, the sequences $(\lambda_k)_{k \geq 0}$ and $(\tau_k^\varphi)_{k \geq 0}$ are interlaced, in the sense that:
\begin{equation*}
\tau_0^\varphi < 0=\lambda_0 < \tau_1^\varphi < \lambda_1 < \tau_2^\varphi < \cdots < \lambda_{k-1} < \tau_k^\varphi < \lambda_k < \cdots .
\end{equation*}
The new eigenvalues of $\Delta_\varphi$ are characterized by the following spectral equation. Let $\tau \in \R \setminus \Sp(\Delta)$, then $\tau \in \Sp(\Delta_\varphi)$ if and only if:
\begin{equation}
\label{eq spectral equation}
\sum_{k \geq 0} r_k \left(\frac{1}{\lambda_k - \tau} - \frac{\lambda_k}{\lambda_k^2+1}\right) = \tan\left(\frac{\varphi}{2}\right) \sum_{k \geq 0} \frac{r_k}{\lambda_k^2 +1}.
\end{equation}
If $\tau \in \R \setminus \Sp(\Delta)$ is a solution of Eq.~\eqref{eq spectral equation}, then $\ker\left(\Delta_\varphi - \tau \Id\right)$ is spanned by the function $G_\tau$ defined below (Def.~\ref{def phi k G tau f tau}). We prove that $G_\tau \in L^2(\T_\alpha)$ in Lem.~\ref{lem CV in Lp}.

\begin{dfns}
\label{def phi k G tau f tau}
For any $k \in \N$, we denote by $\phi_k$ the function:
\begin{equation*}
\begin{array}{cccc}
\phi_k: & \T_\alpha & \longrightarrow & \R\\
& x & \longmapsto & \displaystyle\sum_{\xi \in \Lambda_k} e^{2i\pi \prsc{\xi}{x}},
\end{array}
\end{equation*}
where $\Lambda_k$ is defined by Eq.~\eqref{eq def Lambda k}. Moreover, for any $\tau \in \R \setminus \Sp(\Delta)$, we define two function $G_\tau: \T_\alpha \to \R$ and $f_\tau: \T_\alpha \to \R$ by:
\begin{align*}
G_\tau &= \sum_{k \geq 0} \frac{\phi_k}{\lambda_k - \tau} & &\text{and} & f_\tau & = G_\tau + \frac{1}{\tau} = \sum_{k \geq 1} \frac{\phi_k}{\lambda_k - \tau}.
\end{align*}
\end{dfns}

\begin{rem}
\label{rem phi k G tau f tau}
The sets $\Lambda_k$ are invariant by $(x_1,x_2) \mapsto (-x_1,-x_2)$. Hence, $\phi_k$ is indeed real-valued for any $k \in \N$, and $G_\tau$ and $f_\tau$ are real-valued for any $\tau \in \R \setminus \Sp(\Delta)$.
\end{rem}

In the following, we study the asymptotics as $\tau \to +\infty$ of the moments of $f_\tau$ (resp.~$G_\tau$). We will consider any $\tau \in \R \setminus \Sp(\Delta)$, not only the solutions of Eq.~\eqref{eq spectral equation}. This way, the problem no longer depends on the parameter $\varphi \in (-\pi,\pi)$, that is on the choice of a self-adjoint extension of~$\Delta_{\vert D_0}$. First, we must check that $f_\tau$ (resp.~$G_\tau$) is a well-defined element of $L^p(\T_\alpha)$, for any positive integer $p$. This is the purpose of Lem.~\ref{lem CV in Lp} below.

\begin{lem}
\label{lem CV in Lp}
Let $p \in \N^*$, for any $\tau \in \R \setminus \Sp(\Delta)$, we have $f_\tau \in L^p(\T_\alpha)$ (resp.~$G_\tau  \in L^p(\T_\alpha)$). More precisely, the truncated series $\displaystyle\sum_{k=1}^n \frac{\phi_k}{\lambda_k - \tau}$ converge in $L^p(\T_\alpha)$ as $n \to +\infty$.
\end{lem}

\begin{rems}
\label{rems CV in Lp}
\begin{itemize}
\item The proof is adapted from~\cite{KU2019}, where this result is proved for $p=4$.
\item For $p=2$, it is possible to give a simpler proof, using the Hilbert space structure of $L^2(\T_\alpha)$.
\item The series defining $f_\tau$ does not converge absolutely in $L^p(\T_\alpha)$ in general, even for $p = 2$.
\end{itemize}
\end{rems}

\begin{proof}
Let us fix $p \in \N^*$ and $\tau \in \R \setminus \Sp(\Delta)$. Let $\lambda >0$ and let $\lambda' \in [\lambda,2\lambda]$, we denote by $A(\lambda,\lambda') = \left\{\xi \in \L_\alpha^* \mvert \lambda < 4\pi^2\Norm{\xi}^2 \leq \lambda' \right\}$. This set is finite, and we denote by:
\begin{equation}
\label{eq def hTS}
h_{\lambda,\lambda'} = \sum_{\left\{k \geq 1 \mvert \lambda < \lambda_k \leq \lambda' \right\}} \frac{\phi_k}{\lambda_k - \tau} = \sum_{\xi \in A(\lambda,\lambda')} \frac{e^{2i\pi \prsc{\xi}{\cdot}}}{4\pi^2 \Norm{\xi}^2 - \tau}.
\end{equation}
Since $\T_\alpha$ has unit area, by Jensen's inequality we have $\Norm{h_{\lambda,\lambda'}}_p \leq \Norm{h_{\lambda,\lambda'}}_{2p}$, where $\Norm{\cdot}_p$ stands for the $L^p$ norm associated with the Lebesgue measure on $\T_\alpha$. Then, $\Norm{h_{\lambda,\lambda'}}_{2p}^{2p}$ equals
\begin{equation*}
\int_{\T_\alpha} h_{\lambda,\lambda'}(x)^{2p} \dx x = \sum_{\xi_1,\dots,\xi_{2p} \in A(\lambda,\lambda')} \left(\prod_{j=1}^{2p} (4\pi^2\Norm{\xi_j}^2-\tau)\right)^{-1} \int_{\T_\alpha} \exp\left(2i\pi \prsc{\sum_{j=1}^{2p} \xi_j}{x}\right) \dx x.
\end{equation*}
For any $\xi_1,\dots,\xi_{2p} \in A(\lambda,\lambda')$, we have:
\begin{equation*}
\int_{\T_\alpha} \exp\left(2i\pi \prsc{\sum_{j=1}^{2p} \xi_j}{x}\right) \dx x = \left\{\begin{aligned}
&1, & &\text{if } \sum_{j=1}^{2p} \xi_j=0,\\
&0, & &\text{otherwise.}
\end{aligned} \right.
\end{equation*}
Moreover, as $\lambda \to +\infty$, we have $\prod_{j=1}^{2p} (4\pi^2\Norm{\xi_j}^2-\tau) \sim \prod_{j=1}^{2p} 4\pi^2\Norm{\xi_j}^2 \geq \lambda^{2p}$. So that, for any $\lambda$ large enough, $\prod_{j=1}^{2p} (4\pi^2\Norm{\xi_j}^2-\tau) \geq \frac{1}{2}\lambda^{2p}$. Besides,
\begin{equation*}
\card \left\{(\xi_1,\dots,\xi_{2p}) \in A(\lambda,\lambda')^{2p} \mvert \sum_{j=1}^{2p} \xi_j = 0\right\} \leq \card(A(\lambda,\lambda'))^{2p-1}.
\end{equation*}
By Weyl's Law (Thm.~\ref{thm Weyl Law}), for any $\lambda$ large enough and $\lambda' \in [\lambda,2\lambda]$, we have $\card(A(\lambda,\lambda')) \leq \frac{\lambda}{2}$. Hence $\Norm{h_{\lambda,\lambda'}}_{2p}^{2p} \leq \lambda^{-1}$ and $\Norm{h_{\lambda,\lambda'}}_p \leq \lambda^{-\frac{1}{2p}}$. In particular, this shows that, for any $p \in \N^*$, for any $\tau \in \R \setminus \Sp(\Delta)$, we have:
\begin{equation}
\label{eq estimate hTS}
\Norm{h_{\lambda,\lambda'}}_p = O\left(\lambda^{-\frac{1}{2p}}\right)
\end{equation}
as $\lambda \to +\infty$, where the constant involved in the $O$ notation does not depend on $\lambda' \in [\lambda,2\lambda]$.

Let $n \geq 0$ and let us denote by $L = \lfloor \log_2(\lambda_n) \rfloor$. We have:
\begin{align*}
\Norm{\sum_{k > n} \frac{\phi_k}{\lambda_k-\lambda}}_p &= \Norm{\sum_{l \geq L}\sum_{\left\{k \geq 1 \mvert 2^l < \lambda_k \leq 2^{l+1} \right\}} \frac{\phi_k}{\lambda_k - \tau} - \sum_{\left\{k \geq 1 \mvert 2^L < \lambda_k \leq \lambda_n \right\}} \frac{\phi_k}{\lambda_k-\tau}}_p\\
&\leq \Norm{h_{2^L,\lambda_n}}_p + \sum_{l \geq L} \Norm{h_{2^l,2^{l+1}}}_p\\
&= O\left(\left(2^L\right)^{-\frac{1}{2p}}\right)=O\left(\left(\lambda_n\right)^{-\frac{1}{2p}}\right)\xrightarrow[n \to +\infty]{} 0.
\end{align*}
Thus, the sequence $\left(\sum_{k > n} \frac{\phi_k}{\lambda_k-\tau}\right)_{n \geq 0}$ converges in $L^p(\T_\alpha)$, hence is a Cauchy sequence. Therefore $\left(\sum_{k =1}^n \frac{\phi_k}{\lambda_k-\tau}\right)_{n \geq 0}$ is also a Cauchy sequence, and converges in $L^p(\T_\alpha)$.
\end{proof}


\subsection{Moments of the new eigenfunctions}
\label{subsec moments of the new eigenfunctions}

In this section, we study the central moments of the new eigenfunctions of a point scatterer at $0$ in $\T_\alpha$, where $\alpha >0$. More generally, we study the moments of the functions $f_\tau$, for $\tau \in \R \setminus \Sp(\Delta)$ (see Def.~\ref{def phi k G tau f tau}). The main point is the proof of Prop.~\ref{prop deterministic expression M p tau}, in which we derive an expression of these moments. The content of this section roughly corresponds to Step~1 in Sect.~\ref{subsec random model and results}.

\begin{lem}
\label{lem expectation variance}
For any $\tau \in \R\setminus \Sp(\Delta)$, we have:
\begin{align*}
\int_{\T_\alpha} f_\tau(x) \dx x &= 0, & \int_{\T_\alpha} f_\tau(x)^2 \dx x &= \sum_{k \geq 1} \frac{r_k}{(\lambda_k-\tau)^2},\\
\int_{\T_\alpha} G_\tau(x) \dx x &= -\frac{1}{\tau}, & \int_{\T_\alpha} G_\tau(x)^2 \dx x &= \sum_{k \geq 0} \frac{r_k}{(\lambda_k-\tau)^2},
\end{align*}
where, as before, $(\lambda_k)_{k \geq 0}$ is the increasing sequence of distinct eigenvalues of $\Delta$ and, for any $k \geq 0$, $r_k$ is the multiplicity of $\lambda_k$.
\end{lem}

\begin{proof}
We give the proof in the case of $f_\tau$, the proof is similar for $G_\tau$. Recall that the family $\left\{e^{2i\pi \prsc{\xi}{\cdot}}\mvert \xi \in \L^*_\alpha \right\}$ is a Hilbert basis of $L^2(\T_\alpha)$. Recall also that:
\begin{equation*}
f_\tau = \sum_{k \geq 1} \sum_{\xi \in \Lambda_k} \frac{e^{2i\pi \prsc{\xi}{\cdot}}}{\lambda_k-\tau}.
\end{equation*}
Denoting by $\prsc{\cdot}{\cdot}$ the $L^2$ inner product on $L^2(\T_\alpha)$ and by $\mathbf{1}$ the constant unit function, we have:
\begin{equation*}
\int_{\T_\alpha} f_\tau(x) \dx x = \prsc{f_\tau}{\mathbf{1}} = 0.
\end{equation*}
Besides, by Parseval's Identity,
\begin{equation*}
\int_{\T_\alpha} f_\tau(x)^2 \dx x = \Norm{f_\tau}_2^2 = \sum_{k \geq 1} \frac{\card(\Lambda_k)}{(\lambda_k-\tau)^2} = \sum_{k \geq 1} \frac{r_k}{(\lambda_k-\tau)^2}.\qedhere
\end{equation*}
\end{proof}

As we already explained in Sect.~\ref{subsec random model and results}, we are interested in the central moments of the random variable $G_\tau(X)$, where $X$ is a uniform random variable in $(\T_\alpha, \dx x)$. With this probabilistic point of view, Lem.~\ref{lem CV in Lp} states that, for any $p \in \N^*$ and $\tau \in \R \setminus \Sp(\Delta)$, both $G_\tau(X)$ and $f_\tau(X)$ are $L^p$ random variables. Moreover, by Lem.~\ref{lem expectation variance}, the expectation of $G_\tau(X)$ equals $-\frac{1}{\tau}$. Hence, $f_\tau(X) = G_\tau(X) + \frac{1}{\tau}$ is the centered random variable associated with $G_\tau(X)$. In particular, the central moments of $G_\tau(X)$ are simply the corresponding moments of $f_\tau(X)$. In the following, we focus on studying the moments of $f_\tau(X)$.

\begin{dfn}
\label{def M p tau}
For any $p \in \N^*$ and any $\tau \in \R \setminus \Sp(\Delta)$, we denote by $M^p_\tau = \int_{\T_\alpha} f_\tau(x)^p \dx x$, the $p$-th \emph{moment} of $f_\tau$.
\end{dfn}

Note that $f_\tau$ is not $L^2$-normalized. Indeed, $M^2_\tau = \Norm{f_\tau}_2^2$ is given by Lem.~\ref{lem expectation variance}, and is not equal to $1$ in general. We are ultimately interested in studying the asymptotics, as $\tau \to +\infty$, of the normalized moments $(M^2_\tau)^{-\frac{p}{2}}M^p_\tau$, for $p \in \N^*$, that is the moments of the normalized function $\Norm{f_\tau}_2^{-1}f_\tau$. It is more convenient to work with the sequence $(M^p_\tau)_{p \geq 1}$ for now and only consider the normalized moments later on.  

We conclude this section with the proof of Prop.~\ref{prop deterministic expression M p tau}.

\begin{proof}[Proof of Prop.~\ref{prop deterministic expression M p tau}]
Let $p \in \N^*$ and let $\tau \in \R \setminus \Sp(\Delta)$. For any $l \in \Z$, let us denote by $A_l$ the finite set $A(2^l,2^{l+1}) = \left\{\xi \in \L_\alpha^* \mvert 2^l < 4\pi^2\Norm{\xi}^2 \leq 2^{l+1} \right\}$ and by
\begin{equation*}
h_l = h_{2^l,2^{l+1}} = \sum_{\left\{k \geq 1 \mvert 2^l < \lambda_k \leq 2^{l+1} \right\}} \frac{\phi_k}{\lambda_k - \tau} = \sum_{\xi \in A_l} \frac{e^{2i\pi \prsc{\xi}{\cdot}}}{4\pi^2 \Norm{\xi}^2 - \tau}.
\end{equation*}
Note that these functions already appeared in the proof of Lem.~\ref{lem CV in Lp} (see Eq.~\eqref{eq def hTS}).

Let us denote by $L = \lfloor \log_2(\lambda_1) \rfloor -1$. Recalling Def.~\ref{def phi k G tau f tau}, we have $f_\tau = \sum_{l \geq L} h_l$, where the sum converges absolutely in $L^p(\T_\alpha)$ by Eq.~\eqref{eq estimate hTS}. Then, we have:
\begin{equation*}
M^p_\tau = \int_{\T_\alpha} \left(\sum_{l \geq L} h_l(x)\right)^p \dx x = \int_{\T_\alpha} \sum_{l_1,\dots,l_p \geq L} \left(\prod_{j=1}^p h_{l_j}(x)\right) \dx x = \sum_{l_1,\dots,l_p \geq L} \int_{\T_\alpha} \prod_{j=1}^p h_{l_j}(x) \dx x.
\end{equation*}
We obtain the last equality by exchanging the integral and the sum. This is possible because:
\begin{align*}
\sum_{l_1,\dots,l_p \geq L} \int_{\T_\alpha} \norm{\prod_{j=1}^p h_{l_j}(x)} \dx x = \sum_{l_1,\dots,l_p \geq L} \Norm{\prod_{j=1}^p h_{l_j}}_1 \leq \sum_{l_1,\dots,l_p \geq L} \prod_{j=1}^p \Norm{h_{l_j}}_p = \left(\sum_{l \geq L} \Norm{h_l}_p\right)^p < +\infty,
\end{align*}
where the middle inequality is Hölder's, and the finiteness of the last term is a consequence of Eq.~\eqref{eq estimate hTS}, as we already explained. For any $l_1,\dots,l_p \geq L$, we have:
\begin{equation*}
\prod_{j=1}^p h_{l_j} = \prod_{j=1}^p \sum_{\xi \in A_{l_j}} \frac{e^{2i\pi\prsc{\xi}{\cdot}}}{4\pi^2\Norm{\xi}^2-\tau} = \sum_{(\xi_1,\dots,\xi_p) \in A_{l_1}\times \dots\times A_{l_p}} \left( \prod_{j=1}^p \frac{e^{2i\pi\prsc{\xi_j}{\cdot}}}{4\pi^2\Norm{\xi_j}^2-\tau}\right).
\end{equation*}
Since these sums are indexed by finite sets, we have:
\begin{equation*}
M^p_\tau = \sum_{l_1,\dots,l_p \geq L} \int_{\T_\alpha} \prod_{j=1}^p h_{l_j}(x) \dx x = \sum_{l_1,\dots,l_p \geq L} \ \sum_{(\xi_1,\dots,\xi_p) \in A_{l_1}\times \dots\times A_{l_p}} \int_{\T_\alpha} \prod_{j=1}^p \frac{e^{2i\pi\prsc{\xi_j}{x}}}{4\pi^2\Norm{\xi_j}^2-\tau} \dx x.
\end{equation*}
Then,
\begin{equation*}
\bigsqcup_{l \geq L} A_l = \left\{\xi \in \L^*_\alpha \mvert 2^L \leq 4\pi^2 \Norm{\xi}^2 \right\} =  \L^*_\alpha \setminus \{0\} = \bigsqcup_{k \geq 1} \Lambda_k,
\end{equation*}
where the second equality comes from the fact that $\lambda_1 = \min\left\{4\pi^2\Norm{\xi}^2 \mvert \xi \in \L^*_\alpha \setminus \{0\} \right\}$. Hence,
\begin{align*}
M^p_\tau &= \sum_{\xi_1,\dots,\xi_p \in  \L^*_\alpha \setminus \{0\}} \int_{\T_\alpha} \prod_{j=1}^p \frac{e^{2i\pi\prsc{\xi_j}{x}}}{4\pi^2\Norm{\xi_j}^2-\tau} \dx x\\
&= \sum_{k_1,\dots,k_p \geq 1} \ \sum_{(\xi_1,\dots,\xi_p) \in \Lambda_{k_1}\times \dots\times \Lambda_{k_p}} \int_{\T_\alpha} \prod_{j=1}^p \frac{e^{2i\pi\prsc{\xi_j}{x}}}{4\pi^2\Norm{\xi_j}^2-\tau} \dx x\\
&= \sum_{k_1,\dots,k_p \geq 1} \left(\prod_{j=1}^p \frac{1}{\lambda_{k_j}-\tau}\right) \sum_{(\xi_1,\dots,\xi_p) \in \Lambda_{k_1}\times \dots\times \Lambda_{k_p}} \int_{\T_\alpha} \exp\left(2i\pi\prsc{\sum_{j=1}^p\xi_j}{x} \right)\dx x\\
&= \sum_{k_1,\dots,k_p \geq 1} \left(\prod_{j=1}^p \frac{1}{\lambda_{k_j}-\tau}\right) \card \left\{(\xi_1,\dots,\xi_p) \in \prod_{j=1}^p \Lambda_{k_j} \mvert \sum_{j=1}^p \xi_j=0\right\}.
\end{align*}

Let $k_1,\dots,k_p \geq 1$, for any $k \geq 1$, we denote by $a_k = \card\left\{j \in \{1,\dots,p\} \mvert k_j = k\right\}$. Then $a=(a_k)_{k \geq 1} \in \ell_0$ and $\norm{a}=p$. Moreover,
\begin{equation*}
\prod_{j=1}^p \frac{1}{\lambda_{k_j}-\tau} = \prod_{k \geq 1} \left(\frac{1}{\lambda_k -\tau}\right)^{a_k},
\end{equation*}
and, recalling Def.~\ref{def N a} and reordering the terms of the Cartesian product, we have:
\begin{equation*}
\card \left\{(\xi_1,\dots,\xi_p) \in \prod_{j=1}^p \Lambda_{k_j} \mvert \sum_{j=1}^p \xi_j=0\right\} = N_a.
\end{equation*}
Conversely, let $a= (a_k)_{k\geq 1} \in \ell_0$ be such that $\norm{a}=p$. The choice of $k_1,\dots,k_p \geq 1$ such that, for any $k \geq 1$, $a_k = \card\left\{j \in \{1,\dots,p\} \mvert k_j = k\right\}$ corresponds to the choice of a partition of $\{1,\dots,p\}$ into subsets of cardinalities $a_0,a_1,\dots,a_k,\dots$. The number of such partitions is $p!(a!)^{-1}$. Thus,
\begin{equation*}
M^p_\tau = p! \sum_{a \in \ell_0; \norm{a}=p} \frac{N_a}{a!}\prod_{k \geq 1} \left(\frac{1}{\lambda_k -\tau}\right)^{a_k}.\qedhere
\end{equation*}
\end{proof}


\section{Random model for the moments of the new eigenfunctions}
\label{sec random model for the moments of the new eigenfunctions}

In this section, we study the random model introduced in Sect.~\ref{subsec random model and results} for the central moments of the new eigenfunctions of a point scatterer on a rectangular flat torus $\T_\alpha$, where $\alpha >0$. We start from the expression of these moments derived in Prop.~\ref{prop deterministic expression M p tau}, and we randomize first the wave vectors sets $(\Lambda_k)_{k \geq 1}$, then the spectrum $(\lambda_k)_{k \geq 1}$ of $\Delta$. The content of Sect.~\ref{subsec randomization of the wave vectors} corresponds to what we called Step~2 in the Introduction (Sect.~\ref{subsec random model and results}), that is the randomization of the wave vectors. The main goal of Sect.~\ref{subsec randomization of the wave vectors} is the proof of Prop.~\ref{prop as expression of M p tau}. In Sect.~\ref{subsec randomization of the spectrum of the Laplacian}, we randomize the spectrum of the Laplacian (Step~3 in Sect.~\ref{subsec random model and results}). In Sect.~\ref{subsec limit distribution of the spectral sums}, we study the limit distribution of the spectral sums introduced in Def.~\ref{def S p lambda}. Finally, we prove our main result (Thm.~\ref{thm main}) in Sect.~\ref{subsec proof of the main theorem}.


\subsection{Randomization of the wave vectors}
\label{subsec randomization of the wave vectors}

This section corresponds to the second step in the construction of our random model, that is the randomization of the sets of wave vectors $(\Lambda_k)_{k \geq 1}$ sets (cf.~Step~2 in Sect.~\ref{subsec random model and results}). In Sect.~\ref{subsubsec almost sure expression of the moments}, we work in the setting described in the Introduction (Sect.~\ref{subsec random model and results}) and we prove Prop.~\ref{prop as expression of M p tau}. In Sect.~\ref{subsubsec case of the square torus}, we consider the special case of the square torus ($\alpha =1$). In this setting, we define a variation on our random model with some additional symmetries and show that the results of Sect.~\ref{subsubsec almost sure expression of the moments} remain valid for this alternative model.


\subsubsection{Almost sure expression of the moments}
\label{subsubsec almost sure expression of the moments}

The goal of this section is the proof of Prop.~\ref{prop as expression of M p tau} in the setting described in Step~2 in Sect.~\ref{subsec random model and results}. We assume that we are given an increasing sequence $(\lambda_k)_{k \geq 1}$ of positive numbers. This sequence can be thought of as the spectrum of the Laplacian on some $\T_\alpha$, or its randomized version. We also fix a sequence $(m_k)_{k \geq 1}$ of positive integers. For each $k \geq 1$, $m_k$ is meant to model the number of wave vectors associated with $\lambda_k$ that have non-negative coordinates. Let $(\theta_{k,j})_{k,j \geq 1}$ denote a sequence of random variables in $[0,\frac{\pi}{2}]$ whose distribution admits a density with respect to the measure $\eta$ (see Def.~\ref{def eta}). As in Sect.~\ref{subsec random model and results}, we define the randomized wave vectors sets $(\Lambda_k)_{k \geq 1}$ by Def.~\ref{def randomized wave vectors}. For all $k \geq 1$, we denote by $r_k=\card(\Lambda_k)$, as in the deterministic case.

\begin{rem}
\label{rem as multiplicity}
Almost surely, for all $k$ and $j \geq 1$, we have $\theta_{k,j} \in (0,\frac{\pi}{2})$. This is true for $\eta$, hence for any distribution which is absolutely continuous with respect to $\eta$. In particular, almost surely, for all $k \geq 1$, we have $r_k = 4m_k$.
\end{rem}

Now that we defined a sequence $(\Lambda_k)_{k \geq 1}$ of random finite subsets of $\R^2$, we can define a random sequence of integers $(N_a)_{a \in \ell_0}$ by Def.~\ref{def N a}. Recalling the deterministic expression of the moments obtained in Prop.~\ref{prop deterministic expression M p tau}, we define the randomized moments as follows.

\begin{dfn}
\label{def random Mp}
In the setting we just described, for any $p \in \N^*$ and $\tau \in \R \setminus \left\{ \lambda_k \mvert k \geq 1 \right\}$, we define the randomized moment $M^p_\tau$ by:
\begin{equation*}
M^p_\tau = p! \sum_{a \in \ell_0; \norm{a}=p} \frac{N_a}{a!} \prod_{k \geq 1} \left(\frac{1}{\lambda_k -\tau}\right)^{a_k}.
\end{equation*}
\end{dfn}

Note that it is not clear that the series defining $M^p_\tau$ makes sense almost surely. We prove that it is the case in Cor.~\ref{cor moments well defined} below. This result will be a corollary of the following almost sure expression of the sequence $(N_a)_{a \in \ell_0}$.

\begin{lem}
\label{lem computation Na}
In the setting of this section, almost surely, for all $a =(a_k)_{k \geq 1} \in \ell_0$ we have:
\begin{equation*}
N_a = a! \prod_{k \geq 1} \left(\sum_{\substack{(b_j) \in \N^{m_k}\\ 2\sum_{j=1}^{m_k} b_j = a_k}} \prod_{j=1}^{m_k} \sum_{c=0}^{b_j}\left(\frac{1}{c!\left(b_j-c\right)!}\right)^2\right).
\end{equation*}
In particular, almost surely, $N_a$ only depends on $(a_k)_{k \geq 1}$ and $(m_k)_{k \geq 1}$. Moreover, if there exists $k \geq 1$ such that $a_k$ is odd, then $N_a=0$.
\end{lem}

\begin{proof}
Let $a =(a_k)_{k \geq 1} \in \ell_0$. Let $(\xi_{k,l})_{k \geq 1; 1 \leq l \leq a_k}$ be such that, for all $k \geq 1$ and $l \in \{1,\dots,a_k\}$, $\xi_{k,l} \in \Lambda_k$. Then, for any $(k,l)$, there exists a unique $(i_{k,l},j_{k,l}) \in \{1,2,3,4\}\times \{1,\dots,m_k\}$ such that
\begin{equation*}
\xi_{k,l} = \frac{\sqrt{\lambda_k}}{2\pi} \zeta^{(i_{k,l})}\left(\theta_{k,j_{k,l}}\right).
\end{equation*}
Here we used the notation $\zeta^{(i)}(\theta)$ introduced in Def.~\ref{def randomized wave vectors}. Thus, we have:
\begin{equation*}
N_a = \card\left\{(i_{k,l},j_{k,l}) \in \prod_{k \geq 1} \left(\{1,2,3,4\}\times \{1,\dots,m_k\}\right)^{a_k} \mvert \sum_{k \geq 1} \sum_{l=1}^{a_k} \frac{\sqrt{\lambda_k}}{2\pi} \zeta^{(i_{k,l})}\left(\theta_{k,j_{k,l}}\right)=0 \right\}.
\end{equation*}

Let us consider some sequence $(i_{k,l},j_{k,l}) \in \prod_{k \geq 1} \left(\{1,2,3,4\}\times \{1,\dots,m_k\}\right)^{a_k}$. For any $k \geq 1$, any $j \in \{1,\dots,m_k\}$ and any $i \in \{1,2,3,4\}$, we denote by:
\begin{equation}
\label{eq def c ijk}
c_i(k,j) = \card\left\{l \in \{1,\dots,a_k\} \mvert (i_{k,l},j_{k,l})=(i,j) \right\}.
\end{equation}
Note that $\sum_{k \geq 1} \sum_{j=1}^{m_k}\sum_{i=1}^4 c_i(k,j) = \sum_{k \geq 1} a_k = \norm{a}< +\infty$. In particular, $c_i(k,j)=0$ for all $(k,j,i)$ but a finite number. With these notations, we have:
\begin{align*}
\sum_{k \geq 1} \sum_{l=1}^{a_k} \frac{\sqrt{\lambda_k}}{2\pi} \zeta^{(i_{k,l})}\left(\theta_{k,j_{k,l}}\right) &= \sum_{k \geq 1} \sum_{j=1}^{m_k} \sum_{i=1}^4 \frac{\sqrt{\lambda_k}}{2\pi} c_i(k,j) \zeta^{(i)}(\theta_{k,j})\\
&= \sum_{k \geq 1} \sum_{j=1}^{m_k} \frac{\sqrt{\lambda_k}}{2\pi} \begin{pmatrix}
\left(c_1(k,j)-c_2(k,j)-c_3(k,j)+c_4(k,j)\right)\cos(\theta_{k,j})\\
\left(c_1(k,j)+c_2(k,j)-c_3(k,j)-c_4(k,j)\right)\sin(\theta_{k,j})
\end{pmatrix}.
\end{align*}
If $(i_{k,l},j_{k,l})_{k \geq 1; 1 \leq l \leq a_k}$ is such that for any $k \geq 1$ and any $j \in \{1,\dots,m_k\}$ we have:
\begin{equation}
\label{eq condition c ijk}
\left\{ \begin{aligned}
c_1(k,j)-c_2(k,j)-c_3(k,j)+c_4(k,j) &= 0,\\
c_1(k,j)+c_2(k,j)-c_3(k,j)-c_4(k,j) &= 0,
\end{aligned} \right.
\end{equation}
then $\sum_{k \geq 1} \sum_{l=1}^{a_k} \frac{\sqrt{\lambda_k}}{2\pi} \zeta^{(i_{k,l})}\left(\theta_{k,j_{k,l}}\right)=0$ and the sequence $(i_{k,l},j_{k,l})_{k \geq 1; 1 \leq l \leq a_k}$ contributes to~$N_a$.

On the other hand, let us consider $(i_{k,l},j_{k,l})_{k \geq 1; 1 \leq l \leq a_k}$ such that one of the relations~\eqref{eq condition c ijk} is not satisfied. Let us prove that, in this case, $\sum_{k \geq 1} \sum_{l=1}^{a_k} \frac{\sqrt{\lambda_k}}{2\pi} \zeta^{(i_{k,l})}\left(\theta_{k,j_{k,l}}\right)\neq 0$ for almost every $(\theta_{k,j})_{k,j \geq 1}$, i.e. the sequence $(i_{k,l},j_{k,l})_{k \geq 1; 1 \leq l \leq a_k}$ almost surely does not contribute to $N_a$. Without loss of generality, we can assume that there exists $k_1 \geq 1$ and $j_1\in \{1,\dots,m_{k_1}\}$ such that:
\begin{equation*}
c_1(k_1,j_1)-c_2(k_1,j_1)-c_3(k_1,j_1)+c_4(k_1,j_1) \neq 0.
\end{equation*}
There are at most $\norm{a}$ couples $(k,j)$ such that $c_i(k,j) \neq 0$ for some $i \in \{1,2,3,4\}$. One of these couples is $(k_1,j_1)$. Let us denote by $(k_2,j_2)$, \dots, $(k_q,j_q)$ the others. Let $(\theta_{k,j})_{k,j \geq 1} \in [0,\frac{\pi}{2}]^{\N^* \times \N^*}$, for any $l \in \{1,\dots,q\}$, we denote $\theta_l = \theta_{k_l,j_l}$ for simplicity. If $\sum_{k \geq 1} \sum_{l=1}^{a_k} \frac{\sqrt{\lambda_k}}{2\pi} \zeta^{(i_{k,l})}\left(\theta_{k,j_{k,l}}\right)= 0$, then we have:
\begin{align*}
0 &= \sum_{k \geq 1} \sum_{j=1}^{m_k} \frac{\sqrt{\lambda_k}}{2\pi} \left(c_1(k,j)-c_2(k,j)-c_3(k,j)+c_4(k,j)\right)\cos(\theta_{k,j})\\
&= \sum_{l=1}^q \frac{\sqrt{\lambda_{k_l}}}{2\pi} \left(c_1(k_l,j_l)-c_2(k_l,j_l)-c_3(k_l,j_l)+c_4(k_l,j_l)\right)\cos(\theta_l)\\
&= F(\theta_1,\theta_2,\dots,\theta_q),
\end{align*}
where the last line defines $F:[0,\frac{\pi}{2}]^q \to \R$. For any $(\theta_1,\dots,\theta_q) \in (0,\frac{\pi}{2})^q$, we have:
\begin{equation*}
\deron{F}{\theta_1}(\theta_1,\dots,\theta_q) = -\frac{\sqrt{\lambda_{k_1}}}{2\pi}\left(c_1(k_1,j_1)-c_2(k_1,j_1)-c_3(k_1,j_1)+c_4(k_1,j_1)\right)\sin(\theta_1) \neq 0.
\end{equation*}
Thus $F$ is a smooth submersion on $(0,\frac{\pi}{2})^q$, and $F^{-1}(0) \cap (0,\frac{\pi}{2})^q$ is a smooth hypersurface in $(0,\frac{\pi}{2})^q$. Hence, $F^{-1}(0)$ is included in the union of the boundary of $[0,\frac{\pi}{2}]^q$ and a smooth hypersurface of $(0,\frac{\pi}{2})^q$. Since we assumed that the distribution of $(\theta_{k,j})_{k,j \geq 1}$ is absolutely continuous with respect to $\eta$, then the distribution of $(\theta_1,\dots,\theta_q)$ is absolutely continuous with respect to the Lebesgue measure on $[0,\frac{\pi}{2}]^q$. Hence, the measure of $F^{-1}(0)$ equals $0$, that is $F(\theta_1,\dots,\theta_q) \neq 0$ almost surely. This proves that, if one of the conditions~\eqref{eq condition c ijk} is not satisfied, then almost surely $\sum_{k \geq 1} \sum_{l=1}^{a_k} \frac{\sqrt{\lambda_k}}{2\pi} \zeta^{(i_{k,l})}\left(\theta_{k,j_{k,l}}\right)\neq 0$. Thus, almost surely, we have: for every $(i_{k,l},j_{k,l})_{k \geq 1; 1 \leq l \leq a_k}$,
\begin{equation*}
\sum_{k \geq 1} \sum_{l=1}^{a_k} \frac{\sqrt{\lambda_k}}{2\pi} \zeta^{(i_{k,l})}\left(\theta_{k,j_{k,l}}\right)= 0
\end{equation*}
if and only if the conditions~\eqref{eq condition c ijk} are satisfied.

Since $\ell_0$ is countable, almost surely, for every $a \in \ell_0$, we have:
\begin{equation*}
N_a = \card \left\{(i_{k,l},j_{k,l})_{k \geq 1; 1 \leq l \leq a_k} \mvert \rule{0pt}{22pt} \forall (k,j), \ \left\{ \begin{aligned}
c_1(k,j)-c_2(k,j)-c_3(k,j)+c_4(k,j) &= 0,\\
c_1(k,j)+c_2(k,j)-c_3(k,j)-c_4(k,j) &= 0
\end{aligned} \right.\right\},
\end{equation*}
where the $c_i(k,j)$ are defined by Eq.~\eqref{eq def c ijk}. In the following, we only consider the full-measure set on which the $(N_a)_{a \in \ell_0}$ are defined by the previous formula. Let $a = (a_k)_{k \geq 1} \in \ell_0$ and let $(i_{k,l},j_{k,l})_{(k,l)} \in \prod_{k \geq 1} \left(\{1,2,3,4\}\times \{1,\dots,m_k\}\right)^{a_k}$. For any $k \geq 1$ and $j \in \{1,\dots,m_k\}$, we set
\begin{equation*}
b_{k,j} = \card\left\{l \in \{1,\dots,a_k\} \mvert j_{k,l} = j\right\}.
\end{equation*}
Then, for any $(k,j)$, $(i_{k,l},j_{k,l})_{k\geq 1;1 \leq l \leq a_k}$ satisfies Eq.~\eqref{eq condition c ijk} if and only if the following holds:
\begin{equation*}
\left\{ \begin{aligned}
c_1(k,j)-c_2(k,j)-c_3(k,j)+c_4(k,j) &= 0,\\
c_1(k,j)+c_2(k,j)-c_3(k,j)-c_4(k,j) &= 0,\\
c_1(k,j)+c_2(k,j)+c_3(k,j)+c_4(k,j) &= b_{k,j},
\end{aligned} \right.
\iff
\left\{ \begin{aligned}
c_1(k,j)=c_3(k,j),\\
c_2(k,j)=c_4(k,j),\\
2(c_1(k,j)+c_2(k,j)) = b_{k,j}.
\end{aligned} \right.
\end{equation*}
Hence, the sequence $(i_{k,l},j_{k,l})_{k\geq 1;1 \leq l \leq a_k}$ contributes to the count of $N_a$ if and only, for all $k \geq 1$ and $j \in \{1,\dots,m_k\}$, the integer $b_{k,j}$ is even and there exists $c_{k,j} \in \N$ such that $c_1(k,j)=c_3(k,j)=c_{k,j}$ and $c_2(k,j) = c_4(k,j) = \frac{1}{2}b_{k,j} - c_{k,j}$. In particular, for all $k \geq 1$, $a_k = \sum_{j=1}^{m_k} b_{k,j}$ must be even.

This shows that, if there exists $k \geq 1$ such that $a_k$ is odd, then $N_a = 0$. Conversely, let $a$ be such that $a_k$ is even for any $k \geq 1$. For any $k \geq 1$, let $(b_{k,j})_{1 \leq j \leq m_k}$ be even integers such that $a_k = \sum_{j=1}^{m_k} b_{k,j}$. For any $k \geq 1$ and $j \in \{1,\dots,m_k\}$, let $c_{k,j} \in \{0,\dots, \frac{1}{2}b_{k,j}\}$. Then, the number of sequences $(i_{k,l},j_{k,l}) \in \prod_{k \geq 1} \left(\{1,2,3,4\}\times \{1,\dots,m_k\}\right)^{a_k}$ such that
\begin{align*}
c_1(k,j)&=c_3(k,j)=c_{k,j} & &\text{and} & c_2(k,j)&=c_4(k,j)=\frac{1}{2}b_{k,j}-c_{k,j}
\end{align*}
equals
\begin{equation*}
\prod_{k \geq 1} \left( \frac{a_k !}{\prod_{j=1}^{m_k} b_{k,j}!}
\prod_{j=1}^{m_k}\frac{b_{k,j}!}{(c_{k,j}!)^2\left(\left(\frac{1}{2}b_{k,j}-c_{k,j}\right)!\right)^2}\right) = a!\prod_{k \geq 1}\prod_{j=1}^{m_k}\left(\frac{1}{c_{k,j}!\left(\frac{1}{2}b_{k,j}-c_{k,j}\right)!}\right)^2.
\end{equation*}
Thus, if $a_k$ is even for any $k \geq 1$, we have:
\begin{align*}
N_a &= \sum_{\substack{(b_{k,j}) \in \prod_{k \geq 1} (2\N)^{m_k} \\ \sum_{j=1}^{m_k} b_{k,j}=a_k}} \quad \sum_{\substack{(c_{k,j}) \in \prod_{k \geq 1} \N^{m_k} \\ 0\leq c_{k,j} \leq \frac{1}{2}b_{k,j}}}a!\prod_{k \geq 1}\prod_{j=1}^{m_k}\left(\frac{1}{c_{k,j}!\left(\frac{1}{2}b_{k,j}-c_{k,j}\right)!}\right)^2\\
&= a! \sum_{\substack{(b_{k,j}) \in \prod_{k \geq 1} \N^{m_k} \\ 2 \sum_{j=1}^{m_k} b_{k,j}=a_k}} \quad \sum_{\substack{(c_{k,j}) \in \prod_{k \geq 1} \N^{m_k} \\ 0\leq c_{k,j} \leq b_{k,j}}} \quad \prod_{k \geq 1}\prod_{j=1}^{m_k}\left(\frac{1}{c_{k,j}!\left(b_{k,j}-c_{k,j}\right)!}\right)^2\\
&= a! \prod_{k \geq 1} \left(\sum_{\substack{(b_j) \in \N^{m_k}\\ 2\sum_{j=1}^{m_k} b_j = a_k}} \prod_{j=1}^{m_k} \sum_{c=0}^{b_j}\left(\frac{1}{c!\left(b_j-c\right)!}\right)^2\right).
\end{align*}
The last term still makes sense if one of the $(a_k)_{k \geq 1}$ is odd. In this case, one of the sums is indexed by the empty set, hence equals $0$. This yields the result.
\end{proof}

As a corollary, we can prove that Def.~\ref{def random Mp} almost surely makes sense.

\begin{cor}
\label{cor moments well defined}
Almost surely, for any $p \in \N^*$ and $\tau \in \R \setminus \left\{ \lambda_k \mvert k \geq 1 \right\}$ we have $M^{2p-1}_\tau = 0$ and
\begin{equation*}
M^{2p}_\tau = (2p)! \sum_{a \in \ell_0; \norm{a}=p} \ \prod_{k \geq 1} \left(\left(\frac{1}{\lambda_k -\tau}\right)^{2a_k}\sum_{\substack{(b_j) \in \N^{m_k}\\ \sum_{j=1}^{m_k} b_j = a_k}} \prod_{j=1}^{m_k} \sum_{c=0}^{b_j}\left(\frac{1}{c!\left(b_j-c\right)!}\right)^2\right),
\end{equation*}
where the sum make sense in $[0,+\infty]$ as a sum of non-negative terms.
\end{cor}

\begin{proof}
The sequence $(N_a)_{a \in \ell_0}$ is almost surely given by Lem.~\ref{lem computation Na}. From now on, we only consider the corresponding full-measure event.

Let $p \in \N^*$ and $\tau \in \R \setminus \left\{ \lambda_k \mvert k \geq 1 \right\}$. If $a \in \ell_0$ is such that $\norm{a}=2p-1$, then at least one of the terms $(a_k)_{k \geq 1}$ is odd and $N_a = 0$. Thus, all the terms in the sum defining $M^{2p-1}_\tau$ vanish (see Def.~\ref{def random Mp}) and $M^{2p-1}_\tau = 0$. By the same argument, we can rewrite $M^{2p}_\tau$ as:
\begin{equation*}
M^{2p}_\tau = (2p)! \sum_{a \in \ell_0; \norm{a}=2p} \frac{N_a}{a!} \prod_{k \geq 1} \left(\frac{1}{\lambda_k -\tau}\right)^{a_k} = (2p)! \sum_{a \in \ell_0; \norm{a}=p} \frac{N_{2a}}{(2a)!} \prod_{k \geq 1} \left(\frac{1}{\lambda_k -\tau}\right)^{2a_k}.
\end{equation*}
For any $a \in \ell_0$, the product $\prod_{k \geq 1}\left(\frac{1}{\lambda_k -\tau}\right)^{2a_k}$ only contains a finite number of non-trivial terms, hence it is well-defined in $[0,+\infty)$. Thus, the sum on the right-hand side above makes sense in $[0,+\infty]$, all the terms being non-negative. By Lem.~\ref{lem computation Na}, for any $a \in \ell_0$ we have:
\begin{equation*}
\frac{N_{2a}}{(2a)!} = \prod_{k \geq 1} \left(\sum_{\substack{(b_j) \in \N^{m_k}\\ \sum_{j=1}^{m_k} b_j = a_k}} \prod_{j=1}^{m_k} \sum_{c=0}^{b_j}\left(\frac{1}{c!\left(b_j-c\right)!}\right)^2\right),
\end{equation*}
which yields the result.
\end{proof}

We conclude this section with the proof of Prop.~\ref{prop as expression of M p tau}.

\begin{proof}[Proof of Prop.~\ref{prop as expression of M p tau}]
We already derived an almost sure expression of the randomized moments $(M^p_\tau)_{p \geq 1}$ in Cor.~\ref{cor moments well defined}. In particular, we already proved that the odd moments vanish almost surely, which is the first point in Prop.~\ref{prop as expression of M p tau}.

Recall that, for any $p \geq 1$, the polynomials $P_p$ and $Q_p$ were introduced in Def.~\ref{def P and Q}, and the coefficient $A_p$ was defined in Def.~\ref{def Ap}. By Cor.~\ref{cor moments well defined}, in order to prove the second point of Prop.~\ref{prop as expression of M p tau}, we only need to prove the following for any $p \geq 1$ and $\tau \in \R \setminus \left\{ \lambda_k \mvert k \geq 1 \right\}$: if $S^q_\tau < +\infty$ for any $q \in \{1,\dots,p\}$, then
\begin{equation}
\label{eq goal as expression of M p tau}
\sum_{a \in \ell_0; \norm{a}=p} \ \prod_{k \geq 1} \left(\left(\frac{1}{\lambda_k -\tau}\right)^{2a_k}\sum_{\substack{(b_j) \in \N^{m_k}\\ \sum_{j=1}^{m_k} b_j = a_k}} \prod_{j=1}^{m_k} \sum_{c=0}^{b_j}\left(\frac{1}{c!\left(b_j-c\right)!}\right)^2\right) = \frac{1}{p!} P_p\left(2A_1S^1_\tau,\dots,2A_pS^p_\tau\right).
\end{equation}
Indeed,
\begin{equation*}
\frac{1}{p!} P_p\left(2A_1S^1_\tau,\dots,2A_pS^p_\tau\right) = \sum_{\alpha \in \mathcal{P}(p)} (-1)^{p - \norm{\alpha}}\frac{2^{\norm{\alpha}}}{\alpha!} \prod_{q = 1}^p \left(A_q S^q_\tau\right)^{\alpha_q}.
\end{equation*}

Let $X = (X_k)_{k \geq 1}$, for any $p \in \N^*$, we denote by $F_p(X)$ the formal power series in infinitely many variables defined by:
\begin{equation*}
F_p(X) = \sum_{a \in \ell_0; \norm{a}=p} \left(\prod_{k \geq 1}  X_k^{2a_k}\right) \prod_{k \geq 1}\left(\sum_{\substack{(b_j) \in \N^{m_k}\\ \sum_{j=1}^{m_k} b_j = a_k}} \prod_{j=1}^{m_k} \sum_{c=0}^{b_j}\left(\frac{1}{c!\left(b_j-c\right)!}\right)^2\right).
\end{equation*}
We also define $F(X) = 1 + \sum_{p \geq 1} F_p(X)$. Note that, for all $p \geq 1$, $F_p$ is the homogeneous part of total degree $2p$ of $F$. In the following, we write that two formal power series in $(X_k)_{k\geq 1}$ are equal if, for any $(a_k)_{k \geq 1} \in \ell_0$, the coefficients in front of the monomial $\prod_{k \geq 1}X_k^{a_k}$ are equal. We have:
\begin{align*}
F(X) &= \sum_{a \in \ell_0} \ \prod_{k \geq 1} \sum_{b_1 + \dots + b_{m_k} = a_k} \ \prod_{j=1}^{m_k} \left(X_k^{2b_j} \sum_{c=0}^{b_j}\left(\frac{1}{c!\left(b_j-c\right)!}\right)^2\right)\\
&= \prod_{k \geq 1} \sum_{a_k \geq 0} \ \sum_{b_1 + \dots + b_{m_k} = a_k} \ \prod_{j=1}^{m_k} \left(X_k^{2b_j} \sum_{c=0}^{b_j}\left(\frac{1}{c!\left(b_j-c\right)!}\right)^2\right)\\
&= \prod_{k \geq 1} \sum_{b_1, \dots, b_{m_k} \in \N} \ \prod_{j=1}^{m_k} \left(X_k^{2b_j} \sum_{c=0}^{b_j}\left(\frac{1}{c!\left(b_j-c\right)!}\right)^2\right)\\
&= \prod_{k \geq 1} \left(\sum_{b \geq 0} X_k^{2b} \sum_{c=0}^{b}\left(\frac{1}{c!\left(b-c\right)!}\right)^2\right)^{m_k}\\
&= \prod_{k \geq 1} \left(\sum_{p \geq 0} \frac{X_k^{2p}}{(p!)^2}\right)^{2m_k}.
\end{align*}
Using Eq.~\eqref{eq formal series Q} with $T=X^2$ and $Y_p = \frac{1}{p!}$ for any $p \geq 1$, we have:
\begin{equation}
\label{eq generating function Ap}
\ln\left(1 + \sum_{p \geq 1} \frac{X^{2p}}{(p!)^2}\right) = \sum_{q \geq 1} (-1)^{q-1}Q_q\left(1,\frac{1}{2!},\dots,\frac{1}{q!}\right)X^{2q} = \sum_{q \geq 1} (-1)^{q-1}A_qX^{2q},
\end{equation}
by the definition of the $(A_p)_{p \geq 1}$ (Def.~\ref{def Ap}). Then, using Eq.~\eqref{eq formal series P} with $T=1$, we have:
\begin{align*}
F(X) &= \exp \left(\sum_{k \geq 1} 2m_k \ln\left(I_0(2X_k)\right)\right)\\
&= \exp\left(\sum_{q \geq 1} (-1)^{q-1} 2A_q \sum_{k \geq 1} m_k X_k^{2q}\right)\\
&= 1 + \sum_{p \geq 1} \frac{1}{p!}P_p\left(2A_1 \sum_{k \geq 1}m_k X_k^2,\dots,2A_p \sum_{k \geq 1}m_k X_k^{2p}\right).
\end{align*}
Let $p \in \N^*$, by definition of $P_p$ (see Def.~\ref{def P and Q}), $P_p\left(2A_1 \sum_{k \geq 1}m_k X_k^2,\dots,2A_p \sum_{k \geq 1}m_k X_k^{2p}\right)$ is homogeneous of total degree $2p$ in $X$. Hence,
\begin{equation}
\label{eq formal Fp}
F_p(X) = \frac{1}{p!}P_p\left(2A_1 \sum_{k \geq 1}m_k X_k^2,\dots,2A_p \sum_{k \geq 1}m_k X_k^{2p}\right).
\end{equation}

Let $\tau \in \R \setminus \left\{ \lambda_k \mvert k\geq 1 \right\}$ and let us specialize Eq.~\eqref{eq formal Fp} in the case where $X = \left(\frac{1}{\lambda_k-\tau}\right)_{k \geq 1}$. As we already explained (see the proof of Cor.~\ref{cor moments well defined}), the left-hand side always makes sense in $[0,+\infty]$. However, since some of the coefficients of $P_p$ are negative, if two of the $(S^q_\tau)_{1 \leq q \leq p}$ are infinite, it might not be possible to make sense of the right-hand side of Eq.~\eqref{eq formal Fp} evaluated on $\left(\frac{1}{\lambda_k - \tau}\right)_{k \geq 1}$. Assuming that $S^q_\tau < +\infty$ for all $q \in \{1, \dots,p\}$, we obtain precisely Eq.~\eqref{eq goal as expression of M p tau}, which concludes the proof.
\end{proof}


\subsubsection{Case of the square torus}
\label{subsubsec case of the square torus}

This section is concerned with the special case of the square flat torus $\T = \T_1$. The content of the other sections of this paper is, of course, still valid in this setting. Nonetheless, on $\T$, the deterministic wave vectors sets $(\Lambda_k)_{k \geq 1}$ (see Eq.~\eqref{eq def Lambda k}) have additional symmetries. Namely, they are invariant by $(x_1,x_2) \mapsto (x_2,x_1)$, as explained in Rem.~\ref{rem multiplicities}. In what follows, we define a variation on the model studied in the previous sections, such that the randomized $(\Lambda_k)_{k \geq 1}$ are also invariant by $(x_1,x_2) \mapsto (x_2,x_1)$. It turns out that the almost sure expression of the moments derived in Prop.~\ref{prop as expression of M p tau} is still valid in this case.

This section follows step by step what we did in Step~2 of Sect.~\ref{subsec random model and results} and in Sect.~\ref{subsubsec almost sure expression of the moments}. We assume that we are given an increasing sequence $(\lambda_k)_{k \geq 1}$ of positive numbers and a sequence $(n_k)_{k \geq 1}$ of positive integers. For any $k \geq 1$, $n_k$ will be the number of points in $\Lambda_k$ of the form $(x_1,x_2)$ with $0 \leq x_2 \leq x_1$.

\begin{dfn}
\label{def eta'}
Let $\eta'$ denote the product of the uniform probability measures on $[0,\frac{\pi}{4}]^{\N^* \times \N^*}$, that is the distribution of a sequence of independent uniform variables in $[0,\frac{\pi}{4}]$, indexed by $\N^* \times \N^*$.
\end{dfn}

\begin{dfn}
\label{def randomized wave vectors bis}
Let $(\lambda_k)_{k \geq 1}$ be an increasing sequence of positive numbers, let $(n_k)_{k \geq 1}$ be a sequence of positive integers and let $(\theta_{k,j})_{k,j \geq 1}$ be a sequence of random variables in $[0,\frac{\pi}{4}]$ whose distribution is absolutely continuous with respect to $\eta'$. Then, for any $k\geq 1$, we set:
\begin{equation}
\Lambda_k = \frac{\sqrt{\lambda_k}}{2\pi} \left\{\zeta^{(i)}(\theta_{k,j}) \mvert 1\leq j \leq n_k, 1 \leq i \leq 8 \right\}.
\end{equation}
where we denoted, for every $\theta \in [0,\frac{\pi}{4}]$:
\begin{align*}
\zeta^{(1)}(\theta) &= (\cos(\theta),\sin(\theta)) = - \zeta^{(3)}(\theta), & & & \zeta^{(2)}(\theta) &= (-\cos(\theta),\sin(\theta))= -\zeta^{(4)}(\theta),\\
\zeta^{(5)}(\theta) &= (\sin(\theta),\cos(\theta)) = - \zeta^{(7)}(\theta) & &\text{and} & \zeta^{(6)}(\theta) &= (-\sin(\theta),\cos(\theta))= -\zeta^{(8)}(\theta).
\end{align*}
\end{dfn}

For all $k \geq 1$, we denote by $r_k = \card(\Lambda_k)$ and by $m_k$ the cardinality of $\Lambda_k \cap [0,+\infty)^2$. As in Sect.~\ref{subsubsec almost sure expression of the moments}, almost surely, for any $k$ and $j \geq 1$, we have $\theta_{k,j} \in (0,\frac{\pi}{4})$, so that $r_k = 4m_k = 8 n_k$.

Having defined the randomized $(\Lambda_k)_{k \geq 1}$ by Def.~\ref{def randomized wave vectors bis}, we can define a random sequence of integers $(N_a)_{a \in \ell_0}$ by Def.~\ref{def N a}. Then, for any $p\in \N^*$ and $\tau \in \R \setminus \left\{ \lambda_k \mvert k \geq 1 \right\}$, we can still define the randomized moments $M^p_\tau$ as in Def.~\ref{def random Mp}.

For our original random model, introduced in Sect.~\ref{subsec random model and results}, the only thing we used in the proofs of Cor.~\ref{cor moments well defined} and Prop.~\ref{prop as expression of M p tau} is the almost sure expression of the $(N_a)_{a \in \ell_0}$ derived in Lem.~\ref{lem computation Na} (cf.~Sect.~\ref{subsubsec almost sure expression of the moments}). Our main point here is that Lem.~\ref{lem computation Na} remains valid for the alternate random model introduced in this section. Thus, proceeding as in Sect.~\ref{subsubsec almost sure expression of the moments}, Cor.~\ref{cor moments well defined} and Prop.~\ref{prop as expression of M p tau} hold true in the model with additional symmetry, with the exact same proofs (recall that, almost surely, for all $k \geq 1$, $m_k = 2n_k$). In the next sections, we study the randomized moments $M^p_\tau$ using the almost sure expression derived in Prop.~\ref{prop as expression of M p tau}. In the case of the square torus, our results hold both for the general model of Sect.~\ref{subsec random model and results} and for the model with additional symmetry of Sect.~\ref{subsubsec case of the square torus}.

We conclude this section by proving that Lem.~\ref{lem computation Na} remains valid for the variation on our random model considered in this section.

\begin{proof}[Proof of Lem.~\ref{lem computation Na} for the alternate model of this section]
The proof follows the lines of the one we gave for the original model. For any $a \in \ell_0$, we have:
\begin{equation*}
N_a = \card\left\{(i_{k,l},j_{k,l}) \in \prod_{k \geq 1} \left(\{1,\dots,8\}\times \{1,\dots,n_k\}\right)^{a_k} \mvert \sum_{k \geq 1} \sum_{l=1}^{a_k} \frac{\sqrt{\lambda_k}}{2\pi} \zeta^{(i_{k,l})}\left(\theta_{k,j_{k,l}}\right)=0 \right\}.
\end{equation*}
Let $(i_{k,l},j_{k,l}) \in \prod_{k \geq 1} \left(\{1,\dots,8\}\times \{1,\dots,n_k\}\right)^{a_k}$. For any $k \geq 1$, any $j \in \{1,\dots,n_k\}$ and any $i \in \{1,\dots,8\}$, we set:
\begin{equation*}
c_i(k,j) = \card\left\{l \in \{1,\dots,a_k\} \mvert (i_{k,l},j_{k,l})=(i,j) \right\}.
\end{equation*}
Then
\begin{align*}
\sum_{k \geq 1} \sum_{l=1}^{a_k} \frac{\sqrt{\lambda_k}}{2\pi} \zeta^{(i_{k,l})}\left(\theta_{k,j_{k,l}}\right) =& \sum_{k \geq 1} \sum_{j=1}^{n_k} \frac{\sqrt{\lambda_k}}{2\pi} \begin{pmatrix}
\left(c_1(k,j)-c_2(k,j)-c_3(k,j)+c_4(k,j)\right)\cos(\theta_{k,j})\\
\left(c_1(k,j)+c_2(k,j)-c_3(k,j)-c_4(k,j)\right)\sin(\theta_{k,j})
\end{pmatrix}\\
&+ \sum_{k \geq 1} \sum_{j=1}^{n_k} \frac{\sqrt{\lambda_k}}{2\pi} \begin{pmatrix}
\left(c_5(k,j)-c_6(k,j)-c_7(k,j)+c_8(k,j)\right)\sin(\theta_{k,j})\\
\left(c_5(k,j)+c_6(k,j)-c_7(k,j)-c_8(k,j)\right)\cos(\theta_{k,j})
\end{pmatrix}.
\end{align*}
If $(i_{k,l},j_{k,l})_{k \geq 1; 1 \leq l \leq a_k}$ is such that for any $k \geq 1$ and any $j \in \{1,\dots,n_k\}$ we have:
\begin{equation}
\label{eq condition c ijk bis}
\left\{ \begin{aligned}
c_1(k,j)-c_2(k,j)-c_3(k,j)+c_4(k,j) &= 0,\\
c_1(k,j)+c_2(k,j)-c_3(k,j)-c_4(k,j) &= 0,\\
c_5(k,j)-c_6(k,j)-c_7(k,j)+c_8(k,j) &= 0,\\
c_5(k,j)+c_6(k,j)-c_7(k,j)-c_8(k,j) &= 0
\end{aligned} \right.
\end{equation}
then the above sum equals $0$.

Conversely, let us assume that one of the relations~\eqref{eq condition c ijk bis} is not satisfied. Without loss of generality, we assume that there exists $k_1 \geq 1$ and $j_1 \in \{1,\dots,n_{k_1}\}$ such that:
\begin{equation*}
c_1(k_1,j_1)-c_2(k_1,j_1)-c_3(k_1,j_1)+c_4(k_1,j_1) \neq 0.
\end{equation*}
Let us denote by $(k_2,j_2),\dots,(k_q,j_q)$ the other couples of the form $(k,j)$ with $1 \leq j \leq n_k$ and such that there exists $i \in \{1,\dots,8\}$ such that $c_i(k,j) \neq 0$. If $\sum_{k \geq 1} \sum_{l=1}^{a_k} \frac{\sqrt{\lambda_k}}{2\pi} \zeta^{(i_{k,l})}\left(\theta_{k,j_{k,l}}\right) = 0$, then
\begin{align*}
0 = \tilde{F}(\theta_1,\dots,\theta_q) =& \sum_{l=1}^q \frac{\sqrt{\lambda_{k_l}}}{2\pi} \left(c_1(k_l,j_l)-c_2(k_l,j_l)-c_3(k_l,j_l)+c_4(k_l,j_l)\right)\cos(\theta_l)\\
&+ \sum_{l=1}^q \frac{\sqrt{\lambda_{k_l}}}{2\pi} \left(c_5(k_l,j_l)-c_6(k_l,j_l)-c_7(k_l,j_l)+c_8(k_l,j_l)\right)\sin(\theta_l),
\end{align*}
where we denoted $\theta_l = \theta_{k_l,j_l}$ for any $l \in \{1,\dots,q\}$, and the previous equation defines $\tilde{F}$. For any $(\theta_1,\dots,\theta_q) \in (0,\frac{\pi}{4})^q$, we have:
\begin{align*}
\deron{\tilde{F}}{\theta_1}(\theta_1,\dots,\theta_q) =& -\frac{\sqrt{\lambda_{k_1}}}{2\pi} \left(c_1(k_1,j_1)-c_2(k_1,j_1)-c_3(k_1,j_1)+c_4(k_1,j_1)\right)\sin(\theta_1) \\
&+\frac{\sqrt{\lambda_{k_1}}}{2\pi} \left(c_5(k_1,j_1)-c_6(k_1,j_1)-c_7(k_1,j_1)+c_8(k_1,j_1)\right)\cos(\theta_1),
\end{align*}
and this quantity vanishes if and only if
\begin{equation*}
\tan(\theta_1) = \frac{c_5(k_1,j_1)-c_6(k_1,j_1)-c_7(k_1,j_1)+c_8(k_1,j_1)}{c_1(k_1,j_1)-c_2(k_1,j_1)-c_3(k_1,j_1)+c_4(k_1,j_1)}.
\end{equation*}
Hence, $\tilde{F}$ is a smooth submersion on
\begin{equation*}
\left\{(\theta_1,\dots,\theta_q) \in \left(0,\frac{\pi}{4}\right)^q \mvert \theta_1 \neq \arctan\left(\frac{c_5(k_1,j_1)-c_6(k_1,j_1)-c_7(k_1,j_1)+c_8(k_1,j_1)}{c_1(k_1,j_1)-c_2(k_1,j_1)-c_3(k_1,j_1)+c_4(k_1,j_1)}\right)\right\},
\end{equation*}
and $\tilde{F}^{-1}(0) \subset [0,\frac{\pi}{4}]^q$ has measure $0$ for the measure induced by $\eta'$. Thus, if one the relations \eqref{eq condition c ijk bis} is not satisfied, then almost surely, $\sum_{k \geq 1} \sum_{l=1}^{a_k} \frac{\sqrt{\lambda_k}}{2\pi} \zeta^{(i_{k,l})}\left(\theta_{k,j_{k,l}}\right) \neq 0$. This argument shows that, almost surely, for any $a \in \ell_0$, we have:
\begin{equation*}
N_a = \card \left\{(i_{k,l},j_{k,l})_{k \geq 1; 1 \leq l \leq a_k} \mvert \forall k\geq 1, \forall j \in \{1,\dots,n_k\}, (c_i(k,j))_{1 \leq i \leq 8} \text{ satisfies \eqref{eq condition c ijk bis}} \right\}.
\end{equation*}

Let us consider only the full-measure set on which the $(N_a)_{a \in \ell_0}$ are defined by the previous formula. Let $a = (a_k)_{k \geq 1} \in \ell_0$ and let $(i_{k,l},j_{k,l})_{(k,l)} \in \prod_{k \geq 1} \left(\{1,\dots,8\}\times \{1,\dots,n_k\}\right)^{a_k}$. For any $k \geq 1$ and $j \in \{1,\dots,n_k\}$, we set:
\begin{equation*}
b_{k,j} = \card\left\{l \in \{1,\dots,a_k\} \mvert j_{k,l} = j\right\}.
\end{equation*}
Then, $(i_{k,l},j_{k,l})_{k\geq 1;1 \leq l \leq a_k}$ satisfies Eq.~\eqref{eq condition c ijk bis} for any $(k,j)$ if and only if:
\begin{equation*}
\left\{ \begin{aligned}
&c_1(k,j)=c_3(k,j),\\
&c_2(k,j)=c_4(k,j),\\
&c_5(k,j)=c_7(k,j),\\
&c_6(k,j)=c_8(k,j),\\
&2\left(c_1(k,j)+c_2(k,j)+c_5(k,j)+c_6(k,j)\right)= b_{k,j}.
\end{aligned} \right.
\end{equation*}
Arguing as in the proof of Lem.~\ref{lem computation Na}, we deduce the following. If there exists $k \geq 1$ such that $a_k$ is odd, then $N_a = 0$. And if $a_k$ is even for any $k \geq 1$, then
\begin{align*}
N_a &= a! \sum_{\substack{(b_{k,j}) \in \prod_{k \geq 1} \N^{n_k} \\ 2 \sum_{j=1}^{n_k} b_{k,j}=a_k}} \quad \sum_{\substack{c_1(k,j),c_2(k,j),c_5(k,j),c_6(k,j) \in \N \\ \forall (k,j), \ \sum c_i(k,j)= b_{k,j}}} \quad \prod_{k \geq 1} \prod_{j=1}^{n_k}\left(\frac{1}{c_1(k,j)!c_2(k,j)!c_5(k,j)!c_6(k,j)!}\right)^2\\
&= a! \sum_{\substack{(b_{k,j}) \in \prod_{k \geq 1} \N^{n_k} \\ 2 \sum_{j=1}^{n_k} b_{k,j}=a_k}} \quad \prod_{k \geq 1}\prod_{j=1}^{n_k} \quad \sum_{\substack{c_1,c_2,c_5,c_6 \in \N \\ c_1+c_2+c_5+c_6 = b_{k,j}}}\left(\frac{1}{c_1!c_2!c_5!c_6!}\right)^2.
\end{align*}
Then, for any $b \in \N$, we have:
\begin{align*}
\sum_{c_1+c_2+c_5+c_6 = b}\left(\frac{1}{c_1!c_2!c_5!c_6!}\right)^2 &= \sum_{b' + b'' = b} \left(\sum_{c_1 + c_2 = b'} \left(\frac{1}{c_1!c_2!}\right)^2\right)\left(\sum_{c_5 + c_6 = b''} \left(\frac{1}{c_5!c_6!}\right)^2\right)\\
&= \sum_{b' + b'' = b} \left(\sum_{c=0}^{b'}\left(\frac{1}{c!(b'-c)!}\right)^2 \right)\left(\sum_{c=0}^{b''}\left(\frac{1}{c!(b''-c)!}\right)^2 \right).
\end{align*}
Finally, if $a_k$ is even for any $k \geq 1$, we have:
\begin{align*}
N_a &= a! \sum_{\substack{(b_{k,j}) \in \prod_{k \geq 1} \N^{n_k} \\ 2 \sum_{j=1}^{n_k} b_{k,j}=a_k}} \quad \prod_{k \geq 1}\prod_{j=1}^{n_k} \sum_{\substack{b',b'' \in \N \\ b' + b'' = b_{k,j}}} \left(\sum_{c=0}^{b'}\left(\frac{1}{c!(b'-c)!}\right)^2 \right)\left(\sum_{c=0}^{b''}\left(\frac{1}{c!(b''-c)!}\right)^2 \right)\\
&= a! \sum_{\substack{(b_{k,j}) \in \prod_{k \geq 1} \N^{n_k} \\ 2 \sum_{j=1}^{n_k} b_{k,j}=a_k}} \sum_{\substack{b'_{k,j}, b''_{k,j} \in \N \\ \forall (k,j), \ b'_{k,j}+b''_{k,j}=b_{k,j}}} \prod_{k \geq 1}\prod_{j=1}^{n_k} \left(\sum_{c=0}^{b_{k,j}'}\left(\frac{1}{c!(b_{k,j}'-c)!}\right)^2 \right)\left(\sum_{c=0}^{b_{k,j}''}\left(\frac{1}{c!(b_{k,j}''-c)!}\right)^2 \right)\\
&= a! \sum_{\substack{(b'_{k,j},b''_{k,j}) \in \prod_{k \geq 1} \N^{n_k} \times \N^{n_k} \\ 2 \sum_{j=1}^{n_k} (b'_{k,j} + b''_{k,j}) =a_k}} \quad \prod_{k \geq 1}\prod_{j=1}^{n_k}  \left(\sum_{c=0}^{b_{k,j}'}\left(\frac{1}{c!(b_{k,j}'-c)!}\right)^2 \right)\left(\sum_{c=0}^{b_{k,j}''}\left(\frac{1}{c!(b_{k,j}''-c)!}\right)^2 \right).
\end{align*}
Recalling that $m_k=2n_k$ for all $k \geq 1$, $N_a$ can be rewritten as:
\begin{equation*}
a! \sum_{\substack{(b_{k,j}) \in \prod_{k \geq 1} \N^{m_k} \\ 2 \sum_{j=1}^{m_k} b_{k,j}=a_k}} \prod_{k \geq 1}\prod_{j=1}^{m_k} \sum_{c=0}^{b_{k,j}}\left(\frac{1}{c!(b_{k,j}-c)!}\right)^2 = a! \prod_{k \geq 1} \left(\sum_{\substack{(b_j) \in \N^{m_k}\\ 2\sum_{j=1}^{m_k} b_j = a_k}} \prod_{j=1}^{m_k} \sum_{c=0}^{b_j}\left(\frac{1}{c!\left(b_j-c\right)!}\right)^2\right),
\end{equation*}
which concludes the proof.
\end{proof}


\subsection{Randomization of the spectrum of the Laplacian}
\label{subsec randomization of the spectrum of the Laplacian}

In this section, we finally replace the sequence of eigenvalues of the Laplacian on $\T_\alpha$ by the values of a Poisson point process. This corresponds to the last step in the definition of our random model (Step~3 in Sect.~\ref{subsec random model and results}). Then, we check that the resulting random sequence $(\lambda_k)_{k \geq 1}$ behaves nicely, and we define the corresponding randomized even moments.

Let us consider a multiplicity function $m:[0,+\infty) \to [1,+\infty)$ (cf.~Def.~\ref{def multiplicity function}) and denote by $\nu_m$ the corresponding intensity measure (see Def.~\ref{def nu m}). From now on, we denote by $(\lambda_k)_{k \geq 1}$ the increasing sequence of the values of a Poisson point process of intensity $\nu_m$ on $[0,+\infty)$, as in Def.~\ref{def random spectrum}. The fact that this definition makes sense is the content of Lem.~\ref{lem random spectrum}. Let us denote by $m_k = m(\lambda_k)$ for all $k\geq 1$, as in Def.~\ref{def random spectrum}. As explained in Sect.~\ref{subsec random model and results}, for any $\tau \in \R$, we define the randomized spectral sums $(S^p_\tau)_{p \geq 1}$ by Def.~\ref{def S p lambda}. Recalling, Prop.~\ref{prop as expression of M p tau}, we are now only interested in even order moments, and we can redefine the even order randomized moments $(M^{2p}_\tau)_{p \geq 1}$ as follows.

\begin{dfn}
\label{def random Mp final}
Let $m$ be a multiplicity function and let $(\lambda_k)_{k \geq 1}$ be the associated randomized spectrum (see Def.~\ref{def random spectrum}). For any $k \geq 1$, let $m_k = m(\lambda_k)$. For all $p \in \N^*$, for all $\tau \in \R$, we define the randomized even moment $M^{2p}_\tau$ by:
\begin{equation*}
M^{2p}_\tau = (2p)! \sum_{\alpha \in \mathcal{P}(p)} (-1)^{p - \norm{\alpha}}\frac{2^{\norm{\alpha}}}{\alpha!} \prod_{q = 1}^p \left(A_q S^q_\tau\right)^{\alpha_q},
\end{equation*}
where $\mathcal{P}(p)$, $(A_p)_{p \geq 1}$ and $(S^p_\lambda)_{p \geq 1}$ are defined by Def.~\ref{def partition}, Def.~\ref{def Ap} and Def.~\ref{def S p lambda} respectively.
\end{dfn}

We will prove, in Lem.~\ref{lem Ap positive}, that the coefficients $(A_p)_{p \geq 1}$ are positive. Hence, Def.~\ref{def random Mp final} defines $M^{2p}_\tau$ as a polynomial in $(S^1_\tau,\dots,S^p_\tau)$ with some negative coefficients. In this setting, the spectral sums $(S^p_\tau)_{p \geq 1}$ are random variables taking values in $[0,+\infty]$. In particular, Def.~\ref{def random Mp final} might not make sense if several of these variables take the value $+\infty$ simultaneously. Lem.~\ref{lem as def Sq} states that this happens with probability $0$. Let us now prove this result.

\begin{proof}[Proof of Lem.~\ref{lem as def Sq}]
Let $\tau \in \R$. For all $p \in \N^*$, let us define the function $g_p: t \mapsto \frac{m(t)}{(t-\tau)^{2p}}$. We have $S^p_\tau = \sum_{k \geq 1} g_p(\lambda_k)$. Since
\begin{align*}
\int_0^{+\infty} \min(g_p(t),1) \nu_m(\dx t) &= \frac{1}{16\pi} \int_0^{+\infty} \min\left(\frac{1}{(t-\tau)^{2p}},\frac{1}{m(t)}\right) \dx t\\
&\leq \int_0^{+\infty} \min\left(\frac{1}{(t-\tau)^{2p}},1\right) \dx t <+\infty,
\end{align*}
by Campbell's Theorem (Thm.~\ref{thm Campbell}), $S^p_\tau$ is almost surely finite. Then, almost surely, $S^p_\tau$ is finite for all $p \geq 1$. By Def.~\ref{def random Mp final}, this implies that almost surely, $M^{2p}_\tau$ is well-defined and finite for all $p \geq 1$.
\end{proof}

\begin{rem}
\label{rem expectation Sq}
By Campbell's Theorem, for any $p \geq 1$ and $\tau \geq 0$, $\esp{S^p_\tau} = \frac{1}{16\pi}\int_0^{+\infty} \frac{\dx t}{(t-\tau)^{2p}} = +\infty$.
\end{rem}


\subsection{Limit distribution of the spectral sums}
\label{subsec limit distribution of the spectral sums}

In this section, we study the asymptotic distribution of the randomized spectral sums $(S^p_\tau)_{p \geq 1}$ (cf.~Def.~\ref{def S p lambda}) in the setting described in the previous section (see also Sect.~\ref{subsec random model and results}, Step~3). More precisely, we prove that, after a good rescaling, the finite-dimensional marginals of $(S^p_\tau)_{p \geq 1}$ converge in distribution as $\tau \to +\infty$. We also describe the limit distribution.

First, we need the following basic result about multiplicity functions (see Def.~\ref{def multiplicity function}). It will be useful in the proof of Lem.~\ref{lem limit characteristic function} and the proof of the Random Weyl Law in Sect.~\ref{sec random Weyl Law}.

\begin{lem}
\label{lem multiplicity function}
Let $m$ be a multiplicity function. Then $m'$ is bounded. Moreover, there exists $\alpha <1$ such that $m(t) = O(t^\alpha)$ as $t \to +\infty$.
\end{lem}

\begin{proof}
The function $m'$ is continuous on $[0,+\infty)$ and $m'(t) \xrightarrow[t \to +\infty]{}0$. Hence $\norm{m'}$ admits a maximum on $[0,+\infty)$.

Let $\beta >0$ be such that $m'(t) = O(t^{-\beta})$ as $t \to +\infty$. Without loss of generality, we can assume that $0< \beta < 1$. Since $m$ is a $\mathcal{C}^1$ function, we have $m(t) = m(1) + \int_1^t m'(s) \dx s$ for all $t \geq 0$. Then, as $t \to +\infty$, we have:
\begin{equation*}
\int_1^t \norm{m'(s)} \dx s = O\left(\int_1^t s^{-\beta} \dx s\right) = O(t^{1-\beta}).
\end{equation*}
Setting $\alpha = 1-\beta \in (0,1)$, we have $m(t) = O(t^\alpha)$ as $t \to +\infty$.
\end{proof}

\begin{dfn}
\label{def psi}
Let $p \in \N^*$, we denote by $\psi$ be the function from $\R^p$ to $\R$ defined by:
\begin{equation*}
\psi : (x_1,\dots,x_p) \mapsto \exp\left(-\frac{1}{16\pi} \int_{-\infty}^{+\infty} 1-\exp\left(i\sum_{q=1}^p \frac{x_q}{t^{2q}}\right) \dx t\right).
\end{equation*}
\end{dfn}

\begin{lem}
\label{lem limit characteristic function}
Let $m$ be a multiplicity function, let $(\lambda_k)_{k \geq 1}$ denote the associated randomized spectrum (see Def.~\ref{def random spectrum}) and let $m_k = m(\lambda_k)$ for all $k \geq 1$. Let the randomized spectral sums $(S^p_\tau)_{p \geq 1}$ be defined by Def.~\ref{def S p lambda}. Let $p \in \N^*$, we denote by $S_\tau = \left(m(\tau)S^1_\tau,m(\tau)^3S^2_\tau,\dots,m(\tau)^{2p-1}S^p_\tau\right)$, for any $\tau \geq 0$.

Then $S_\tau\xrightarrow[\tau \to +\infty]{}S$ in distribution, where $S$ is a random vector in $\R^p$ whose characteristic function is $\psi$. In particular, the limit distribution does not depend on $m$.
\end{lem}

\begin{proof}
Let $p \in \N^*$ and let $\tau \geq 0$. Let us denote by $\psi_\tau$ the characteristic function of $S_\tau$ and let  $x=(x_1,\dots,x_p) \in \R^p$. For any $q \in \{1,\dots,p\}$, let $g_q:t \mapsto \frac{m(t)}{(t-\tau)^{2q}}$. We have already seen (cf.~the proof of Lem.~\ref{lem as def Sq}) that $g_q$ satisfies the hypothesis of Campbell's Theorem (Thm.~\ref{thm Campbell}) and that $S^q_\tau = \sum_{k \geq 1} g_q(\lambda_k)$.  Hence, we have:
\begin{equation*}
\psi_\tau(x) = \esp{\exp\left(i \sum_{q=1}^p x_qm(\tau)^{2q-1}S^q_\tau\right)} = \exp\left(-I_\tau(x)\right),
\end{equation*}
where
\begin{align*}
I_\tau(x) &= \int_0^{+\infty} \left(1-\exp\left(i\sum_{q=1}^p x_qm(\tau)^{2q-1}\frac{m(t)}{(t-\tau)^{2q}}\right)\right) \nu_m(\dx t)\\
&= \frac{1}{16\pi} \int_{-\tau}^{+\infty} \left(1-\exp\left(i\sum_{q=1}^p x_qm(\tau)^{2q-1}\frac{m(\tau+t)}{t^{2q}}\right)\right)\frac{\dx t}{m(\tau+t)}\\
&= \frac{1}{16\pi} \int_{-\frac{\tau}{m(\tau)}}^{+\infty} \left(1-\exp\left(i\frac{m(\tau+tm(\tau))}{m(\tau)}\sum_{q=1}^p \frac{x_q}{t^{2q}}\right)\right)\frac{m(\tau)}{m(\tau+tm(\tau))}\dx t.
\end{align*}

Given $\tau \in \R$ and $t \in \left[-\frac{\tau}{m(\tau)},+\infty\right)$, we have:
\begin{equation}
\label{eq simple convergence I tau}
\norm{\frac{m(\tau+tm(\tau))}{m(\tau)}-1} = \frac{1}{m(\tau)}\norm{\int_\tau^{\tau+tm(\tau)}m'(s) \dx s} \leq \norm{t} \max \left\{\norm{m'(s)} \rule{0pt}{10pt}\mvert s \geq \tau - \norm{t}m(\tau)\right\}.
\end{equation}
Let $t \in \R$, by Lem.~\ref{lem multiplicity function} we have $m(\tau) = o(\tau)$ so that $\tau - \norm{t}m(\tau) \xrightarrow[\tau \to +\infty]{} +\infty$. By Def.~\ref{def multiplicity function}, $m'(t) \xrightarrow[\tau \to +\infty]{} 0$. Hence, the maximum on the right-hand side of Eq.~\eqref{eq simple convergence I tau} goes to $0$ as $\tau \to +\infty$ and $\frac{m(\tau+tm(\tau))}{m(\tau)} \xrightarrow[\tau \to +\infty]{}1$. Thus, for all $t \in \R^*$,
\begin{equation*}
\left(1-\exp\left(i \frac{m(\tau+tm(\tau))}{m(\tau)}\sum_{q=1}^p \frac{x_q}{t^{2q}}\right)\right)\frac{m(\tau)}{m(\tau+tm(\tau))}\mathbf{1}_{\left[-\frac{\tau}{m(\tau)},+\infty\right)}(t) \xrightarrow[\tau \to +\infty]{} 1-\exp\left(i\sum_{q=1}^p \frac{x_q}{t^{2q}}\right).
\end{equation*}
We need an integrable dominating function, in order to apply Lebesgue's Dominated Convergence Theorem. By Lem.~\ref{lem multiplicity function}, $m'$ is bounded. Let us denote by $\Norm{m'} = \sup_{t \geq 0} \norm{m'(t)}$. For any $t \in \R^*$, if $\norm{t} \leq \frac{1}{2\Norm{m'}}$, we have $\norm{\frac{m(\tau+tm(\tau))}{m(\tau)}-1} \leq \frac{1}{2}$ by Eq.~\eqref{eq simple convergence I tau}. Then $\frac{m(\tau+tm(\tau))}{m(\tau)} \geq \frac{1}{2}$ and
\begin{equation*}
\norm{\left(1-\exp\left(i\frac{m(\tau+tm(\tau))}{m(\tau)} \sum_{q=1}^p \frac{x_q}{t^{2q}}\right)\right)}\frac{m(\tau)}{m(\tau+tm(\tau))} \leq 4.
\end{equation*}
If $\norm{t} \geq \frac{1}{2\Norm{m'}}$, we have:
\begin{equation*}
\norm{\left(1-\exp\left(i \frac{m(\tau+tm(\tau))}{m(\tau)}\sum_{q=1}^p \frac{x_q}{t^{2q}}\right)\right)}\frac{m(\tau)}{m(\tau+tm(\tau))} \leq \norm{\sum_{q=1}^p\frac{x_q}{t^{2q}}} \leq \Norm{x}\sum_{q=1}^p\frac{1}{t^{2q}},
\end{equation*}
where $\Norm{\cdot}$ stands for the Euclidean norm of $\R^p$. Thus, we can use the following dominating function:
\begin{equation*}
t \mapsto \left\{ \begin{aligned}
&\Norm{x}\sum_{q=1}^p\frac{1}{t^{2q}}, & &\text{if }\norm{t} \geq \frac{1}{2\Norm{m'}},\\
&4, & &\text{otherwise.}
\end{aligned} \right.
\end{equation*}
Hence, by Lebesgue's Theorem, for any $x \in \R^p$, we have:
\begin{equation*}
I_\tau(x) \xrightarrow[\tau \to +\infty]{} \frac{1}{16\pi} \int_{-\infty}^{+\infty} 1-\exp\left(i\sum_{q=1}^p \frac{x_q}{t^{2q}}\right) \dx t,
\end{equation*}
and $\psi_\tau$ converges pointwise to $\psi$ (see Def.~\ref{def psi}).

Let $x \in \R^p$ be such that $\Norm{x} \leq 1$. For any $t \in \R^*$, we have:
\begin{equation*}
\norm{1-\exp\left(i\sum_{q=1}^p \frac{x_q}{t^{2q}}\right)} \leq \min\left(2,\norm{\sum_{q=1}^p \frac{x_q}{t^{2q}}}\right) \leq \min\left(2, \Norm{x}\sum_{q=1}^p \frac{1}{t^{2q}}\right)\leq \min\left(2,\sum_{q=1}^p \frac{1}{t^{2q}}\right).
\end{equation*}
The function $t \mapsto \min\left(2,\sum_{q=1}^p \frac{1}{t^{2q}}\right)$ is integrable. Hence, by Lebesgue's Theorem,
\begin{equation*}
\int_{-\infty}^{+\infty} 1-\exp\left(i\sum_{q=1}^p \frac{x_q}{t^{2q}}\right) \dx t \xrightarrow[x \to 0]{} 0
\end{equation*}
and $\psi(x) \xrightarrow[x \to 0]{} 1 = \psi(0)$. Since $\psi$ is continuous at $0$, by Lévy's Continuity Theorem, $\psi$ is the characteristic function of a probability distribution. Moreover, denoting by $S$ a random vector in~$\R^p$ whose characteristic function is $\psi$, we have $S_\tau \xrightarrow[\tau \to +\infty]{}S$ in distribution.
\end{proof}

\begin{lem}
\label{lem limit distribution}
Let $p \in \N^*$ and let $S=(S^1,\dots,S^p)$ be a random vector in $\R^p$ whose characteristic function is $\psi$ (cf.~Def.~\ref{def psi}). Then, $S$ admits a smooth density $D$ with respect to the Lebesgue measure of $\R^p$. The map $D$ is the inverse Fourier transform of $\psi$. Moreover, $S$ is supported in:
\begin{equation*}
\bar{\Omega} = \left\{(x_1,\dots,x_p) \in \R^p \rule{0pt}{9pt}\mvert x_1 \geq 0 \text{ and }\forall q \in \{2,\dots,p\},\ 0 \leq x_q \leq x_1 x_{q-1} \right\}.
\end{equation*}
\end{lem}

\begin{proof}
By Lem.~\ref{lem psi fast decreasing}, the function $\psi$ is fast decreasing. Hence, its inverse Fourier transform is well-defined and is smooth. Denoting by $D$ this function, the distribution of $S$ admits the density~$D$ with respect to the Lebesgue measure.

By Lem.~\ref{lem limit characteristic function}, the random vector $S$ is the limit of $\left(m(\tau)S^1_\tau,m(\tau)^3S^2_\tau,\dots,m(\tau)^{2p-1}S^p_\tau\right)$ in distribution, where $m$ is any multiplicity function. Choosing $m$ to be the constant unit function, we have:
\begin{equation*}
(S^1_\tau,\dots,S^p_\tau)=\left(\sum_{k \geq 1} \frac{1}{(\lambda_k - \tau)^2},\dots,\sum_{k \geq 1} \frac{1}{(\lambda_k - \tau)^{2p}}\right) \xrightarrow[\tau \to +\infty]{} S,
\end{equation*}
in distribution, where $(\lambda_k)_{k \geq 1}$ is a Poisson point process on $[0,+\infty)$ with constant intensity $\frac{1}{16\pi}$. For any $\tau \in \R$ and any $q \in \{1,\dots,p\}$, we have $S^q_\tau >0$. Moreover, if $q \geq 2$, then
\begin{equation*}
S^{q-1}_\tau S^1_\tau = \sum_{k \neq l} \frac{1}{(\lambda_k-\tau)^{2q-2}}\frac{1}{(\lambda_l-\tau)^2} + \sum_{k \geq 1} \frac{1}{(\lambda_k-\tau)^{2q}} > S^q_\tau.
\end{equation*}
Since $S$ admits a density with respect to the Lebesgue measure, we have $\P\left(S \in \partial \bar{\Omega} = 0\right)$. Hence,
\begin{equation*}
\P\left(S \in \bar{\Omega}\right) = \lim_{\tau \to +\infty} \P\left((S^1_\tau,\dots,S^p_\tau) \in \bar{\Omega}\right) =1.
\end{equation*}
Finally the support of $S$ is included in the closed set $\bar{\Omega}$.
\end{proof}

\begin{rem}
\label{rem limit distribution}
Let $p \in \N^*$ and let $S=(S^1,\dots,S^p)$ be a random vector in $\R^p$ whose characteristic function is $\psi$. Then for any $q \in \{1,\dots,p\}$, the characteristic function of $S^q$ is
\begin{equation*}
\psi_q : x \mapsto  \exp\left(-\frac{1}{16\pi} \int_{-\infty}^{+\infty} 1-\exp\left( \frac{ix}{t^{2q}}\right) \dx t\right).
\end{equation*}
For any $x \in \R^*$, we have:
\begin{equation*}
\int_{-\infty}^{+\infty} 1-\exp\left( \frac{ix}{t^{2q}}\right) \dx t = 2 \int_0^{+\infty} \exp\left( \frac{ix}{t^{2q}}\right)-1 \dx t = \frac{\norm{x}^\frac{1}{2q}}{q}\int_0^{+\infty} \frac{e^{i\frac{x}{\norm{x}} u}-1}{u^{1+\frac{1}{2q}}} \dx u.
\end{equation*}
By~\cite[p.~90]{UZ1999}, the integral on the right-hand side equals $\exp\left(-i\frac{x}{\norm{x}}\frac{\pi}{4q}\right)\Gamma\left(-\frac{1}{2q}\right)$, where $\Gamma$ is Euler's Gamma Function. Hence, denoting by $c_q = \left(\frac{1}{8\pi}\cos\left(\frac{\pi}{4q}\right)\Gamma\left(1-\frac{1}{2q}\right)\right)^{2q}$, we have:
\begin{equation*}
\psi_q(x) = \exp\left(-\norm{c_q x}^\frac{1}{2q}\left(1-i\frac{x}{\norm{x}}\tan\left(\frac{\pi}{4q}\right)\right)\right).
\end{equation*}

This shows that $S^q$ follows a one-sided $\frac{1}{2q}$-stable law, more precisely the stable law of parameters $\left(\frac{1}{2q},1,c_q,0\right)$. This distribution is known to admit a smooth density with respect to the Lebesgue measure, and to be supported in $[0,+\infty)$. For $q \geq 2$, to the best of our knowledge, no analytical expression of the density is known. For $q =1$, this distribution coincides with the Lévy distribution of parameters $0$ and $\frac{1}{128\pi}$. In particular, its density is $t \mapsto \frac{1}{16\pi} t^{-\frac{3}{2}}\exp\left(-\frac{1}{256\pi t}\right)\mathbf{1}_{[0,+\infty)}(t)$. Note also that the expectation of $S^q$ is infinite, for all $q \geq 1$.
\end{rem}

\begin{lem}
\label{lem limit density quotient}
Let $p \in \N^*$ and let $S=(S^1,\dots,S^p)$ be a random vector in $\R^p$ whose characteristic function is $\psi$ (cf.~Def.~\ref{def psi}). Then the random vector $\mathcal{S}=\left(\frac{S^2}{(S^1)^2},\dots,\frac{S^p}{(S^1)^p}\right)$ admits a smooth density
\begin{equation}
\label{eq def D}
\mathcal{D}:(y_2,\dots,y_p) \longmapsto \int_0^{+\infty} (x_1)^{\frac{p(p+1)}{2}-1}D\left(x_1,y_2(x_1)^2,\dots,y_p(x_1)^p\right) \dx x_1,
\end{equation}
with respect to the Lebesgue of $\R^{p-1}$, where $D$ is the inverse Fourier transform of $\psi$. Moreover, the support of $\mathcal{D}$ is contained in the compact set:
\begin{equation*}
K = \left\{(y_2,\dots,y_p) \in \R^{p-1} \rule{0pt}{9pt}\mvert 0 \leq y_p \leq \dots \leq y_2 \leq 1 \right\}.
\end{equation*}
\end{lem}

\begin{proof}
Let us denote by $\tilde{\mathcal{S}} = \left(S^1,\frac{S^2}{(S^1)^2},\dots,\frac{S^p}{(S^1)^p}\right)$. By Lem.~\ref{lem limit distribution}, the random vector $S$ admits the density $D$. In particular, $\P\left(S^1=0\right)=0$, so that $\tilde{\mathcal{S}}$ and $\mathcal{S}$ are almost surely well-defined. Then, for any continuous bounded function $g:\R^p\to \R$, we have:
\begin{multline*}
\esp{g\left(\tilde{\mathcal{S}}\right)} = \int_{x \in \R^p} g\left(x_1,\frac{x_2}{(x_1)^2},\dots,\frac{x_p}{(x_1)^p}\right)D(x_1,\dots,x_p) \dx x\\
= \int_{(x_1,y_2,\dots,y_p) \in \R^p} g\left(x_1,y_2,\dots,y_p\right)D\left(x_1,y_2(x_1)^2,\dots,y_p(x_1)^p\right) \norm{x_1}^{\frac{p(p+1)}{2}-1} \dx x_1 \dx y_2 \dots \dx y_p.
\end{multline*}
Recalling that $D$ is smooth and supported in $[0,+\infty)^p$ (see Lem.~\ref{lem limit distribution}), this shows that $\tilde{\mathcal{S}}$ admits the smooth density $\tilde{D}:(x_1,y_2,\dots,y_p) \mapsto (x_1)^{\frac{p(p+1)}{2}-1}D\left(x_1,y_2(x_1)^2,\dots,y_p(x_1)^p\right)$, with respect to the Lebesgue measure of $\R^p$. Since $\tilde{\mathcal{S}}=(S^1,\mathcal{S})$, the random vector $\mathcal{S}$ admits the density $\mathcal{D}$ defined by Eq.~\eqref{eq def D}.

It is not clear that Eq.~\eqref{eq def D} defines a smooth function. In order to prove this, note that $\tilde{D}$ is smooth, so that the characteristic function $\tilde{\psi}$ of $\tilde{\mathcal{S}}$ is fast decreasing. Then, the characteristic function of $\mathcal{S}$ is $(y_2,\dots,y_p) \mapsto \tilde{\psi}(0,y_2,\dots,y_p)$, which is also fast decreasing. Thus, $\mathcal{D}$ is smooth.

Finally, by Lem.~\ref{lem limit distribution}, the support of $S$ is contained in $\bar{\Omega}$. Hence, we have almost surely $S^1 \geq 0$ (in fact $S^1 >0$) and $0 \leq S^q \leq S^{q-1}S^1$, for all $q \in \{2,\dots,p\}$. Thus, almost surely, we have:
\begin{equation*}
0 \leq \frac{S^p}{(S^1)^p} \leq \frac{S^{p-1}}{(S^1)^{p-1}} \leq \dots \leq \frac{S^2}{(S^1)^2} \leq 1,
\end{equation*}
which proves that the support of $\mathcal{D}$ is contained in $K$.
\end{proof}


\subsection{Proof of Thm.~\ref{thm main}}
\label{subsec proof of the main theorem}

This section is concerned with the proof of our main result (Thm.~\ref{thm main}). Let $m:[0,+\infty) \to [1,+\infty)$ be a multiplicity function and let $\nu_m$ be the associated intensity measure (see Def.~\ref{def multiplicity function} and~\ref{def nu m}). Let $(\lambda_k)_{k \geq 1}$ and $(m_k)_{k \geq 1}$ be as in Def.~\ref{def random spectrum}, that is $(\lambda_k)_{k \geq 1}$ is a Poisson point process of intensity~$\nu_m$ and, for all $k \geq 1$, $m_k = m(\lambda_k)$. Let the randomized spectral sums $(S^p_\tau)_{p \geq 1}$ be defined by Def.~\ref{def S p lambda}. Finally, let the even randomized moments $(M^{2p}_\tau)_{p \geq 1}$ be defined by Def.~\ref{def random Mp final}. We have seen in Lem.~\ref{lem as def Sq} that these random variables are almost surely well-defined for all $\tau \in \R$.

Since Thm.~\ref{thm main} is concerned with the normalized moments, let us start by deriving an expression of these normalized moments. Note that, by Def.~\ref{def P and Q} and Def.~\ref{def Ap}, $A_1 = Q_1(1) = 1$. Besides, by Def.~\ref{def random Mp final}, for all $\tau \in \R$, $M^2_\tau = 4S^1_\tau$. 

Recall that $(\mu_p)_{p \geq 1}$ denotes the sequence of moments of the standard Gaussian distribution (see Def.~\ref{def mu p}) and that the polynomials $(P_p)_{p \geq 1}$ are defined by Def.~\ref{def P and Q}. Let $p \in \N^*$ and let $\tau \in \R$, by Def.~\ref{def random Mp final}, we have:
\begin{align*}
\frac{M^{2p}_\tau}{(M^2_\tau)^p} &= \frac{(2p)!}{(4S^1_\tau)^p} \sum_{\alpha \in \mathcal{P}(p)} (-1)^{p - \norm{\alpha}}\frac{2^{\norm{\alpha}}}{\alpha!} \prod_{q = 1}^p \left(A_q S^q_\tau\right)^{\alpha_q}\\
&= \frac{(2p)!}{2^p p!} \sum_{\alpha \in \mathcal{P}(p)} (-1)^{p-\norm{\alpha}}\frac{p!}{\alpha!} \prod_{q \geq 1} \left(\frac{2A_qS^q_\tau}{(2S^1_\tau)^q}\right)^{\alpha_q}\\
&= \mu_{2p} P_p\left(1,\frac{2A_2S^2_\tau}{(2S^1_\tau)^2},\dots,\frac{2A_pS^p_\tau}{(2S^1_\tau)^p}\right).
\end{align*}
In view of Lem.~\ref{lem limit characteristic function}, we rewrite this expression in the following form:
\begin{equation}
\label{eq normalized moments}
\frac{M^{2p}_\tau}{(M^2_\tau)^p} = \mu_{2p} P_p\left(1,\frac{A_2}{2m(\tau)}\frac{m(\tau)^3
S^2_\tau}{(m(\tau) S^1_\tau)^2},\dots,\frac{A_p}{\left(2m(\tau)\right)^{p-1}}\frac{m(\tau)^{2p-1}S^p_\tau}{(m(\tau)S^1_\tau)^p}\right).
\end{equation}

\begin{dfn}
\label{def Rpl}
Let $p \in \N \setminus \{0,1\}$ and let $S = (S^1,\dots,S^p)$ be a random vector in $\R^p$ whose characteristic function is $\psi$ (see Def.~\ref{def psi}). Let $l \in (0,+\infty)$, we define the random vector $\left(R_2(l),\dots,R_p(l)\right)$ in $\R^{p-1}$ by:
\begin{equation*}
\left(R_q(l)\right)_{2\leq q \leq p} = \left(P_q\left(1,\frac{A_2}{2l}\frac{S^2}{(S^1)^2}, \dots,\frac{A_q}{(2l)^{q-1}}\frac{S^q}{(S^1)^q}\right)\right)_{2\leq q \leq p},
\end{equation*}
where $(P_p)_{p \geq 1}$ is defined by Def.~\ref{def P and Q} and $(A_p)_{p \geq 1}$ is defined by Def.~\ref{def Ap}. We also define the random vector $\left(R_2(+\infty),\dots,R_p(+\infty)\right)$ as being constant equal to $(1,\dots,1) \in \R^{p-1}$.
\end{dfn}

\begin{rem}
\label{rem Rpl}
The distributions of the random vectors $\left((S^1,\dots,S^p)\right)_{p \geq 1}$ are compatible, in the sense that, if $1 \leq q \leq p$, then $(S^1,\dots,S^q)$ is distributed as the first $q$ components of $(S^1,\dots,S^p)$. Thus, if $l \in (0,+\infty]$ and $2 \leq q \leq p$, then $\left(R_2(l),\dots,R_q(l)\right)$ is distributed as the first $(q-1)$ components of $\left(R_2(l),\dots,R_p(l)\right)$. In particular, the distribution of $\left(R_2(l),\dots,R_q(l)\right)$ is uniquely defined.
\end{rem}

The following lemma proves the convergence in distribution stated in Thm.~\ref{thm main}. Note that the same proof also works to prove the result stated in Rem.~\ref{rem main thm CV tau}, as $\tau \to +\infty$.

\begin{lem}
\label{lem proof of CV in distribution}
Let $p \in \N \setminus \{0,1\}$ and let $l \in (0,+\infty]$. Let $(\tau_n)_{n \in \N}$ be a sequence such that $\tau_n \xrightarrow[n \to +\infty]{} +\infty$ and $m(\tau_n) \xrightarrow[n \to +\infty]{}l$, then the following holds in distribution:
\begin{equation*}
\left(\frac{M^4_{\tau_n}}{(M^2_{\tau_n})^2},\dots,\frac{M^{2p}_{\tau_n}}{(M^2_{\tau_n})^p}\right) \xrightarrow[n \to +\infty]{} \left(\mu_4 R_2(l),\dots,\mu_{2p}R_p(l)\right).
\end{equation*}
\end{lem}

\begin{proof}
By Lem.~\ref{lem limit characteristic function}, we have $\left(m(\tau)S^1_\tau,m(\tau)^3S^2_\tau,\dots,m(\tau)^{2p-1}S^p_\tau\right) \xrightarrow[\tau \to +\infty]{} S$ in distribution, where $S=(S^1,\dots,S^p)$ is a random vector of characteristic function $\psi$ (see Def.~\ref{def psi}). Since $S$ admits a density (see Lem.~\ref{lem limit distribution}), we have $S^1 \neq 0$ almost surely. Hence, by the Continuous Mapping Theorem, the following holds in distribution:
\begin{equation*}
\left(\frac{m(\tau)^3
S^2_\tau}{(m(\tau) S^1_\tau)^2},\dots,\frac{m(\tau)^{2p-1}S^p_\tau}{(m(\tau)S^1_\tau)^p}\right) \xrightarrow[\tau \to +\infty]{} \left(\frac{S^2}{(S^1)^2},\dots,\frac{S^p}{(S^1)^p}\right).
\end{equation*}

If $l \in (0,+\infty)$, by Eq.~\eqref{eq normalized moments} and the Continuous  Mapping Theorem, we have:
\begin{equation*}
\left(\frac{M^{2q}_{\tau_n}}{(M^2_{\tau_n})^q}\right)_{2 \leq q \leq p} \xrightarrow[n \to +\infty]{} \left(\mu_{2q}P_q\left(1,\frac{A_2}{2l}\frac{S^2}{(S^1)^2}, \dots,\frac{A_q}{(2l)^{q-1}}\frac{S^q}{(S^1)^q}\right)\right)_{2\leq q \leq p}
\end{equation*}
in distribution. This yields the result, by definition of $(R_q(l))_{2 \leq q \leq p}$ (see Def.~\ref{def Rpl}).

If $l=+\infty$, the same arguments shows that:
\begin{equation*}
\left(\frac{M^{2q}_{\tau_n}}{(M^2_{\tau_n})^q}\right)_{2 \leq q \leq p} \xrightarrow[n \to +\infty]{} \left(\mu_{2q}P_q\left(1,0,\dots,0\right)\right)_{2\leq q \leq p}
\end{equation*}
in distribution. By Def.~\ref{def P and Q}, for all $q \geq 2$, $P_q(1,0,\dots,0)=1$. This concludes the proof, since $(R_2(+\infty),\dots,R_p(+\infty)) = (1,\dots,1)$ by definition.
\end{proof}

We have just proved the first part of Thm.~\ref{thm main}. In order to conclude the proof, we need to study the distribution of the random vectors $(R_q(l))_{2 \leq q \leq p}$ defined by Def.~\ref{def Rpl}. For $l=+\infty$, we have $(R_2(l),\dots,R_p(l)) = (1,\dots,1)$ by definition, for all $p \geq 2$. In this case, the limit distribution in Thm.~\ref{thm main} is a Dirac mass. Hence the convergence also holds in probability. When $l \in (0,+\infty)$, the following lemma proves that the limit distribution admit a smooth density. Note that this density is explicit.

\begin{lem}
\label{lem limit density moments}
Let $p \in \N \setminus \{0,1\}$, and let $\mathcal{D}$ denote the function defined by Eq.\eqref{eq def D}. Then, for any $l \in (0,+\infty)$, the distribution of the random vector $\left(R_2(l),\dots,R_p(l)\right)$ (see Def.~\ref{def Rpl}) admits the smooth density
\begin{equation}
\label{eq def Dl}
\mathcal{D}_l:(x_2,\dots,x_p) \longmapsto \left(2l\right)^\frac{p(p-1)}{2} \left(\prod_{q=2}^p \frac{1}{A_qq!}\right) \mathcal{D}\!\left(\frac{2l}{A_2}
Q_2(1,x_2),\dots,\frac{(2l)^{p-1}}{A_p}Q_p(1,x_2,\dots,x_p)\right)
\end{equation}
with respect to the Lebesgue measure of $\R^{p-1}$, where $(Q_q)_{q \geq 1}$ and $(A_p)_{p \geq 1}$ are defined by Def.~\ref{def P and Q} and Def.~\ref{def Ap} respectively.
\end{lem}

\begin{proof}
Let $(S^1,\dots,S^p) \in \R^p$ be a random vector whose characteristic function is $\psi$ (see Def.~\ref{def psi}). Let $g:\R^{p-1} \to \R$ be a continuous bounded function. By Def.~\ref{def Rpl}, we have:
\begin{align*}
\esp{g\left(R_2(l),\dots,R_p(l)\right)} &= \esp{g\left(P_2\left(1,\frac{A_2}{2l}\frac{S^2}{(S^1)^2}\right),\dots,P_p\left(1,\frac{A_2}{2l}\frac{S^2}{(S^1)^2},\dots,\frac{A_p}{(2l)^{p-1}}\frac{S^p}{(S^1)^p}\right)\right)}\\
&= \int_{y \in \R^{p-1}} g \circ \Phi'_p \left(\frac{A_2}{2l}y_2,\dots,\frac{A_p}{(2l)^{p-1}}y_p\right) \mathcal{D}(y) \dx y,
\end{align*}
where $\mathcal{D}$ is the density of $\left(\frac{S^2}{(S^1)^2},\dots,\frac{S^p}{(S^1)^p}\right)$ (cf.~Lem.~\ref{lem limit density quotient}) and $\Phi'_p$ is defined by Def.~\ref{def Phi p Psi p}. By Lem.~\ref{lem Ap positive}, we have $A_p >0$ for all $p \geq 1$. Then,
\begin{multline*}
\int_{y \in \R^{p-1}} g \circ \Phi'_p \left(\frac{A_2}{2l}y_2,\dots,\frac{A_p}{(2l)^{p-1}}y_p\right) \mathcal{D}(y) \dx y\\
\begin{aligned}
&= \int_{y \in \R^{p-1}} g\left(\Phi'_p(y)\right) \mathcal{D}\left(\frac{2l}{A_2}y_2,\dots,\frac{(2l)^{p-1}}{A_p} y_p\right) \left(2l\right)^\frac{p(p-1)}{2}\left(\prod_{q=2}^p\frac{1}{A_q}\right)\dx y\\
&= \int_{x \in \R^{p-1}} g(x) \mathcal{D}_l(x) \dx x,
\end{aligned}
\end{multline*}
where we obtained the last line by a change of variable $x=\Phi'_p(y)$. Indeed, by Cor.~\ref{cor Phi' and Psi' diffeo}, the map $\Phi'_p$ is a diffeomorphism of inverse $\Psi'_p$ (see Def.~\ref{def Phi p Psi p}), and its Jacobian determinant is $\prod_{q=2}^p q!$.
\end{proof}

The last part of Thm.~\ref{thm main} is the content of the following lemma.

\begin{lem}
\label{lem different limit density}
Let $p \in \N \setminus \{0,1\}$, then $\left\{(R_2(l),\dots,R_p(l))\mvert l\in (0,+\infty] \right\}$ is a one-parameter family of random vectors, whose distributions are pairwise distinct. That is, $(R_2(l),\dots,R_p(l))$ and $(R_2(l'),\dots,R_p(l'))$ have the same distribution if and only if $l=l'$.
\end{lem}

\begin{proof}
Since $(R_2(+\infty),\dots,R_p(+\infty))$ is constant and $(R_2(l),\dots,R_p(l))$ is absolutely continuous with respect to Lebesgue for $l \in (0,+\infty)$, they can not be equal in distribution. To conclude the proof, by Lem.~\ref{lem limit density moments}, it is enough to prove that if $l$ and $l'\in (0,+\infty)$ are different, then $\mathcal{D}_l \neq \mathcal{D}_{l'}$, where these smooth functions are defined by Eq.~\eqref{eq def Dl}.

By Lem.~\ref{lem limit density quotient}, the density $\mathcal{D}$ defined by Eq.~\eqref{eq def D} is smooth with compact support. Hence it admits a global maximum, that we denote by $M$. Since $\mathcal{D}$ is the density of a probability distribution, we have $M>0$.

For any $l >0$, the map:
\begin{equation*}
(x_2,\dots,x_p) \longmapsto \left(\frac{2l}{A_2}Q_2(1,x_2),\dots,\frac{(2l)^{p-1}}{A_p} Q_p(1,x_2,\dots,x_p)\right)
\end{equation*}
is a global diffeomorphism from $\R^{p-1}$ to itself. Indeed, this map is the compose of $\Psi'_p$ (cf.~Def.~\ref{def Phi p Psi p}) with $(x_2,\dots,x_p) \mapsto \left(\frac{2l}{A_2}x_2,\dots,\frac{(2l)^{p-1}}{A_p}x_p\right)$. The former is a diffeomorphism by Cor.~\ref{cor Phi' and Psi' diffeo}. Since $l >0$, the latter is also a well-defined diffeomorphism from $\R^{p-1}$ to itself, by Lem.~\ref{lem Ap positive}. Hence, the density $\mathcal{D}_l$ defined by Eq.~\eqref{eq def Dl} admits a positive maximum, equal to $\left(2l\right)^\frac{p(p-1)}{2} \left(\prod_{q=2}^p \frac{1}{A_qq!}\right)M$.

Finally, if $l$ and $l' \in (0,+\infty)$ are different, then $\mathcal{D}_l$ and $\mathcal{D}_{l'}$ do not have the same global maximum. In particular, $\mathcal{D}_l \neq \mathcal{D}_{l'}$.
\end{proof}


\section{Random Weyl Law}
\label{sec random Weyl Law}

In this section, we prove the Random Weyl Law (Prop.~\ref{prop random Weyl Law}). Let $m:[0,+\infty) \to [1,+\infty)$ be a multiplicity function, in the sense of Def.~\ref{def multiplicity function}. As in Sect.~\ref{subsec random model and results}, we denote by $\nu_m$ the measure with density $(16\pi m)^{-1}$ with respect to Lebesgue  on $[0,+\infty)$ (cf.~Def.~\ref{def nu m}). Let $(\lambda_k)_{k \geq 1}$ denote the increasing sequence of the values of a Poisson point process of intensity $\nu_m$. As in Def.~\ref{def random spectrum}, we denote $m_k=m(\lambda_k)$ for all $k \geq 1$.

The deterministic Weyl Law (Thm.~\ref{thm Weyl Law}) is concerned with the function $\mathcal{N}$ that counts the number of eigenvalues of $\Delta$ smaller than its argument, the eigenvalues being counted with multiplicity. With the notations of Sect.~\ref{subsec point scatterers on flat tori}, for all $\lambda \geq 0$, we have:
\begin{equation*}
\mathcal{N}(\lambda) = 1 + \sum_{\{ k \geq 1 \mid \lambda_k \leq \lambda\}} \card(\Lambda_k).
\end{equation*}
In our random model, recall that almost surely $\card(\Lambda_k) = 4m_k$ for all $k \geq 1$. Hence, the random variable $\mathcal{N}_m(\lambda)$ defined by Eq.~\eqref{eq def Nm lambda} is indeed the analogue  of $\mathcal{N}(\lambda)$ is this model. Prop.~\ref{prop random Weyl Law} is the main reason behind our definition of the intensity measure $\nu_m$.

\begin{proof}[Proof of Prop.~\ref{prop random Weyl Law}]
Let $\lambda \geq 0$, we first compute the first two moments of $\mathcal{N}_m(\lambda)$. Note that
\begin{equation*}
\mathcal{N}_m(\lambda) = 1 + \sum_{\{k \geq 1 \mid \lambda_k \leq \lambda\}} 4m_k = 1 + \sum_{k \geq 1}  4m(\lambda_k)\mathbf{1}_{[0,\lambda]}(\lambda_k),
\end{equation*}
where $\mathbf{1}_{[0,\lambda]}$ denotes the indicator function of $[0,\lambda]$. Applying Campbell's Theorem (Thm.~\ref{thm Campbell}) to the function $4m\mathbf{1}_{[0,\lambda]}$, we have:
\begin{equation*}
\esp{\mathcal{N}_m(\lambda)} = 1 + \esp{\sum_{k \geq 1} 4m(\lambda_k) \mathbf{1}_{[0,\lambda]}(\lambda_k)} = 1+\int_0^\lambda 4m(t) \nu_m(\dx t) = 1+\frac{\lambda}{4\pi},
\end{equation*}
and
\begin{equation*}
\var{\mathcal{N}_m(\lambda)} = \var{\sum_{k \geq 1} 4m(\lambda_k) \mathbf{1}_{[0,\lambda]}(\lambda_k)} = \int_0^\lambda 16m(t)^2 \nu_m(\dx t) = \frac{1}{\pi} \int_0^\lambda m(t) \dx t.
\end{equation*}

Let us now prove the Law of Large Numbers. By Lem.~\ref{lem multiplicity function}, there exists $\alpha <1$ such that $m(t) = O(t^\alpha)$. Since  $m(t) \geq 1$ for all $t \geq 0$, we must have $\alpha \geq 0$. Then
\begin{equation*}
\var{\mathcal{N}_m(\lambda)}= \frac{1}{\pi}\int_0^\lambda m(t) \dx t = O\left(\int_0^\lambda t^\alpha \dx t \right) = O\left(\lambda^{\alpha+1}\right).
\end{equation*}
For all $k \in \N^*$, we set $t_k = k^\frac{2}{1-\alpha}$. We have $t_k \to +\infty$ as $k \to +\infty$, and
\begin{equation*}
\var{\frac{1}{t_k}\mathcal{N}_m(t_k)} = O\left((t_k)^{\alpha-1}\right) = O(k^{-2}).
\end{equation*}
Then,
\begin{equation*}
\esp{\sum_{k \geq 1} \frac{1}{t_k^2}\left(\mathcal{N}_m(t_k)-\left(1+\frac{t_k}{4\pi}\right)\right)^2} = \sum_{k \geq 1} \var{\frac{1}{t_k}\mathcal{N}_m(t_k)} < +\infty.
\end{equation*}
Hence, $\sum_{k \geq 1} \left(\frac{1}{t_k}\mathcal{N}_m(t_k)-\frac{1}{t_k}-\frac{1}{4\pi}\right)^2< +\infty$ almost surely, which implies that almost surely $\frac{1}{t_k}\mathcal{N}_m(t_k) \xrightarrow[k \to +\infty]{}\frac{1}{4\pi}$.

Let $\lambda \geq 1$, there exists a unique $k \in \N^*$ such that $t_k \leq \lambda < t_{k+1}$. Since $\mathcal{N}_m$ is a non-decreasing function, we have:
\begin{equation*}
\frac{t_k}{t_{k+1}}\frac{\mathcal{N}_m(t_k)}{t_k} = \frac{\mathcal{N}_m(t_k)}{t_{k+1}} \leq \frac{\mathcal{N}_m(t_k)}{\lambda} \leq \frac{\mathcal{N}_m(\lambda)}{\lambda} \leq \frac{\mathcal{N}_m(t_{k+1})}{\lambda} \leq \frac{\mathcal{N}_m(t_{k+1})}{t_k} = \frac{t_{k+1}}{t_k} \frac{\mathcal{N}_m(t_{k+1})}{t_{k+1}}.
\end{equation*}
Then, since $\frac{t_{k+1}}{t_k} \xrightarrow[k \to +\infty]{}1$ and $k \sim \lambda^\frac{1-\alpha}{2}$ as $\lambda \to +\infty$, we have $\frac{1}{\lambda}\mathcal{N}_m(\lambda) \xrightarrow[\lambda \to +\infty]{}\frac{1}{4\pi}$ almost surely.

Let us denote by $\sigma(\lambda) = \left(\frac{1}{\pi} \int_0^\lambda m(t) \dx t\right)^\frac{1}{2}$. In order to prove the second claim, that is the Central Limit Theorem, it is enough to prove that the characteristic function $\chi_\lambda$ of
\begin{equation*}
\frac{\mathcal{N}_m(\lambda) - \esp{\mathcal{N}_m(\lambda)}}{\var{\mathcal{N}_m(\lambda)}^\frac{1}{2}} = \frac{\mathcal{N}_m(\lambda) - 1 - \frac{\lambda}{4\pi}}{\sigma(\lambda)}
\end{equation*}
converges pointwise toward $x \mapsto e^{-\frac{x^2}{2}}$, the characteristic function of a standard Gaussian variable. Since
\begin{equation*}
\int_0^{+\infty} \min(4m(t)\mathbf{1}_{[0,\lambda]}(t),1) \nu_m(\dx t) = \frac{1}{16\pi} \int_0^\lambda \frac{\dx t}{m(t)} \leq \frac{\lambda}{16\pi} <+\infty, 
\end{equation*}
the characteristic function of $\mathcal{N}_m(\lambda)-1$ is given by Campbell's Theorem. For any $x \in \R$, we have:
\begin{align*}
\chi_\lambda(x) &= \exp\left(-ix\frac{\lambda}{4\pi\sigma(\lambda)}\right) \esp{\exp\left(i\frac{x}{\sigma(\lambda)}\sum_{k \geq 1} 4m(\lambda_k)\mathbf{1}_{[0,\lambda]}(\lambda_k)\right)}\\
&= \exp\left(-\int_0^\lambda \left(1-e^{i\frac{4xm(t)}{\sigma(\lambda)}}\right) \nu_m(\dx t) -ix\frac{\lambda}{4\pi\sigma(\lambda)}\right)\\
&= \exp\left(- \int_0^\lambda \left(1+i\frac{4xm(t)}{\sigma(\lambda)}-e^{i\frac{4xm(t)}{\sigma(\lambda)}}\right) \nu_m(\dx t)\right).
\end{align*}
Hence, it enough to prove that, for all $x \in \R$,
\begin{equation*}
\int_0^\lambda \left(1+i\frac{4xm(t)}{\sigma(\lambda)}-e^{i\frac{4xm(t)}{\sigma(\lambda)}}\right) \nu_m(\dx t) \xrightarrow[\lambda \to +\infty]{} \frac{x^2}{2}.
\end{equation*}

Using the Taylor--Lagrange Formula for $\theta \mapsto e^{i\theta}$ at order $2$ around $0$, we obtain that, for all $t \in [0,\lambda]$,
\begin{equation*}
\norm{1+i\frac{4xm(t)}{\sigma(\lambda)}-\frac{1}{2}\left(\frac{4xm(t)}{\sigma(\lambda)}\right)^2-e^{i\frac{4xm(t)}{\sigma(\lambda)}}} \leq \norm{\frac{4xm(t)}{\sigma(\lambda)}}^3.
\end{equation*}
We have:
\begin{equation*}
\int_0^\lambda \norm{\frac{4xm(t)}{\sigma(\lambda)}}^3 \nu_m(\dx t) = \frac{4\norm{x}^3}{\pi} \frac{1}{\sigma(\lambda)^3} \int_0^\lambda m(t)^2 \dx t.
\end{equation*}
Since $m$ is a multiplicity function, there exists $\beta \in (0,1)$ and $C >0$ such that $\norm{m'(t)} \leq C t^{-\beta}$, for all $t \geq 0$. In particular, $m(t) =O(t^{1-\beta})$, as in the proof of Lem.~\ref{lem multiplicity function}.
Let $s$ and $t \in \left[\lambda^\frac{1}{2},\lambda\right]$, we have: $\norm{m(t)-m(s)} \leq C \norm{t-s} \lambda^{-\frac{\beta}{2}}$. Thus, if $\norm{t-s} \leq \frac{1}{2C}m(s)\lambda^\frac{\beta}{2}$, we have $m(t) \geq \frac{1}{2}m(s)$. Note that $m(s)\lambda^\frac{\beta}{2} = \lambda^\frac{\beta}{2}O(s^{1-\beta})=O(\lambda^{1-\frac{\beta}{2}})$. Hence, for $\lambda$ large enough, for any $s \in \left[\lambda^\frac{1}{2},\lambda\right]$, at least one half of the interval $\left[s-\frac{1}{2C}m(s)\lambda^\frac{\beta}{2}, s+\frac{1}{2C}m(s)\lambda^\frac{\beta}{2}\right]$ is contained in $\left[\lambda^\frac{1}{2},\lambda\right]$ and
\begin{equation*}
\int_0^\lambda m(t) \dx t \geq \int_{\left[s-\frac{1}{2C}m(s)\lambda^\frac{\beta}{2}, s+\frac{1}{2C}m(s)\lambda^\frac{\beta}{2}\right] \cap \left[\lambda^\frac{1}{2},\lambda\right]} m(t) \dx t \geq \frac{m(s)^2}{4C}\lambda^\frac{\beta}{2}.
\end{equation*}
Denoting by $M(\lambda)$ the maximum of $m$ on $\left[\lambda^\frac{1}{2},\lambda\right]$, this yields $M(\lambda)^2 \leq 4C \pi  \sigma(T)^2 \lambda^{-\frac{\beta}{2}}$, hence $M(\lambda) = o(\sigma(\lambda))$, as $\lambda \to +\infty$. Then,
\begin{equation*}
\int_{\lambda^\frac{1}{2}}^\lambda m(t)^2 \dx t \leq M(\lambda) \int_{\lambda^\frac{1}{2}}^\lambda m(t) \dx t \leq \pi M(\lambda)\sigma(\lambda)^2 = o(\sigma(\lambda)^3).
\end{equation*}
On the other hand, 
\begin{equation*}
\int_0^{\lambda^\frac{1}{2}} m(t)^2 \dx t = O\left(\int_0^{\lambda^\frac{1}{2}} t^{2-2\beta} \dx t\right) = O(\lambda^{\frac{3}{2}-\beta})= o(\sigma(\lambda)^3).
\end{equation*}
Indeed, since $m(t) \geq 1$ for all $t \geq 0$, we have $\sigma(\lambda) = \left(\frac{1}{\pi} \int_0^\lambda m(t) \dx t\right)^\frac{1}{2} \geq \sqrt{\frac{\lambda}{\pi}}$. Finally, we have:
\begin{equation*}
\frac{1}{\sigma(\lambda)^3} \int_0^\lambda m(t)^2 \dx t \xrightarrow[\lambda \to +\infty]{} 0.
\end{equation*}
Thus, as $\lambda \to +\infty$,
\begin{align*}
\int_0^\lambda \left(1+i\frac{4xm(t)}{\sigma(\lambda)}-e^{i\frac{4xm(t)}{\sigma(\lambda)}}\right) \nu_m(\dx t) &= \frac{1}{2}\int_0^\lambda  \left(\frac{4xm(t)}{\sigma(\lambda)}\right)^2 \nu_m(\dx t) + o(1)\\
&= \frac{x^2}{2} + o(1).\qedhere
\end{align*}
\end{proof}


\appendix

\section{Reminder on Poisson point processes}
\label{sec reminder on Poisson point processes}

In this section, we gather the few facts we need about Poisson point processes on $[0,+\infty)$. Apart from recalling the definition of these processes, we prove Lem.~\ref{lem random spectrum}, and we state Campbell's Theorem (Thm.~\ref{thm Campbell}), which is one of the key tools in the proof of our main results (Thm.~\ref{thm main} and Prop.~\ref{prop random Weyl Law}).

We start with the definition of a Poisson point process in the setting we are interested in.

\begin{dfn}
\label{def Poisson point process}
Let $\nu$ denote a Borel measure on $[0,+\infty)$, a \emph{Poisson point process} on $[0,+\infty)$ with \emph{intensity measure} $\nu$ is a random subset $P$ of~$[0,+\infty)$, whose distribution is characterized by the following properties.
\begin{itemize}
\item For any bounded Borel subset $B \subset [0,+\infty)$, the cardinality of $B \cap P$ is distributed as a Poisson variable of parameter $\nu(B)$.
\item For any $n \in \N^*$, for any family $B_1,\dots,B_n$ of disjoint bounded Borel subsets of $[0,+\infty)$, the random variables $\card(B_1 \cap P), \dots, \card(B_n \cap P)$ are globally independent.
\end{itemize}
\end{dfn}

In the case, where the intensity measure is of the form $\nu_m$ (see Def.~\ref{def nu m}) for some multiplicity function $m:[0,+\infty) \to [1,+\infty)$ (see Def.~\ref{def multiplicity function}), the elements of the Poisson point process $P$ form a nice sequence.

\begin{lem}
\label{lem random spectrum}
Let $m$ be a multiplicity function and let $\nu_m$ denote the associated intensity measure. Let $P$ be a Poisson point process on $[0,+\infty)$ of intensity $\nu_m$. Then, almost surely, the elements of $P$ can be ordered into an increasing sequence of positive numbers $(\lambda_k)_{k \geq 1}$ such that $\lambda_k \xrightarrow[k \to +\infty]{}+\infty$.
\end{lem}

\begin{proof}
Let us denote by $P \subset [0,+\infty)$ the Poisson point process of intensity $\nu_m$. For any $\lambda \geq 0$, we have:
\begin{equation*}
\nu_m([0,\lambda]) = \frac{1}{16\pi} \int_0^\lambda \frac{\dx t}{m(t)} < +\infty.
\end{equation*}
By Lem.~\ref{lem multiplicity function}, $m(t)\leq t$ for all $t$ large enough. Hence $t \mapsto \frac{1}{m(t)}$ is not integrable at $+\infty$ and $\nu_m([0,\lambda]) \xrightarrow[\lambda \to +\infty]{} +\infty$. Let $n \in \N$, for all $\lambda \geq 0$ we have:
\begin{equation*}
\P\left(\card(P) \leq n\right) \leq \P\left(\card(P\cap [0,\lambda]) \leq n\right) = e^{-\nu_m([0,\lambda])} \sum_{k=0}^n \frac{\nu_m([0,\lambda])^k}{k!} \xrightarrow[\lambda \to +\infty]{} 0.
\end{equation*}
Hence, for all $n \in \N$, $\P\left(\card(P) \leq n\right)=0$, and $P$ is almost surely infinite. Besides, almost surely, for any $\lambda \geq 0$, $\card(P\cap [0,\lambda]) < +\infty$. Indeed,
\begin{equation*}
\left\{\exists \lambda >0 \mvert \card(P\cap [0,\lambda]) = +\infty\right\} = \bigcup_{\lambda >0}\hspace{-1mm}\uparrow \left\{\card(P\cap [0,\lambda]) = +\infty\right\},
\end{equation*}
so that the probability of this event equals $\lim_{\lambda \to +\infty} \P\left(\card(P\cap [0,\lambda]) = +\infty\right) = 0$. Thus, almost surely, $P$ is an infinite set whose elements can be ordered into an increasing sequence $(\lambda_k)_{k \geq 1}$ that goes to infinity as $k \to +\infty$. Finally, $\P\left(\lambda_1=0\right) \leq \P\left(P \cap \{0\} >0\right) = 0$. Hence, almost surely $\lambda_1 >0$, which ensures that $(\lambda_k)_{k \geq 1}$ is a sequence of positive numbers.
\end{proof}

We conclude this section by recalling Campbell's Theorem. The form of this theorem is the reason why we defined the sequence $(m_k)_{k \geq 1}$ as $(m(\lambda_k))_{k \geq 1}$ in our model. We state it in the special case we are interested in (compare~\cite[Sect.~3.2, pp.~28--31]{Kin1993}).

\begin{thm}[Campbell]
\label{thm Campbell}
Let $(\lambda_k)_{k \geq 1}$ be the values of a Poisson point process on $[0,+\infty)$ with intensity measure $\nu$. Let $g: [0,+\infty) \to [0,+\infty)$ be a measurable function and let $S = \sum_{k \geq 1} g(\lambda_k)$. Then we have:
\begin{equation*}
\esp{S} = \int_0^{+\infty} g(t)\nu(\dx t)
\end{equation*}
in $[0,+\infty]$. If this quantity is finite, then the following holds in $[0,+\infty]$:
\begin{equation*}
\var{S} = \int_0^{+\infty} g(t)^2\nu(\dx t).
\end{equation*}
Besides, $S$ is almost surely finite if and only if $\displaystyle\int_0^{+\infty} \min\left(g(t),1\right) \nu(\dx t) < +\infty$.

Moreover, let $p \in \N^*$. For any $q \in \{1,\dots,p\}$, let $g_q:[0,+\infty) \to [0,+\infty)$ be a measurable map such that $\displaystyle\int_0^{+\infty} \min\left(g_q(t),1\right) \nu(\dx t) < +\infty$ and let $S^q = \sum_{k \geq 1} g_q(\lambda_k)$. Then, the characteristic function of $S=(S^1,\dots,S^p)$ is given by:
\begin{equation*}
\forall (x_1,\dots,x_p) \in \R^p, \ \esp{\exp\left(i \sum_{q=1}^p x_q S^q\right)} = \exp\left(-\int_0^{+\infty} \left(1-\exp\left(i\sum_{q=1}^p x_q g_q(t)\right)\right) \nu(\dx t) \right).
\end{equation*}
In particular, the integral on the right-hand side is well-defined for all $(x_1,\dots,x_p) \in \R^p$.
\end{thm}


\section{Technical lemmas}
\label{sec technical lemmas}

In this appendix, we gather the statements and proofs of several technical results that we used in the paper. In Sect.~\ref{subsec fast decreasing characteristic function}, we prove that the function $\psi$ defined by Def.~\ref{def psi} is fast decreasing. Then, in Sect.~\ref{subsec polynomial maps}, we prove several results concerning the polynomials $(P_p)_{p \geq 1}$ and $(Q_q)_{q \geq 1}$ (see Def.~\ref{def P and Q}). Finally, in Sect.~\ref{subsec coefficients}, we prove that the coefficients $(A_p)_{p \geq 1}$ are positive (see Def.~\ref{def Ap}).


\subsection{Fast decreasing characteristic function}
\label{subsec fast decreasing characteristic function}

\begin{lem}
\label{lem psi fast decreasing}
Let $\psi :\R^p \to \R$ be the function defined by Def.~\ref{def psi}. Their exists $C >0$ such that $\norm{\psi(x)} = O\left(\exp\left(-C\Norm{x}^\frac{1}{2p}\right)\right)$, where $\Norm{\cdot}$ stands for the canonical Euclidean norm of $\R^p$.
\end{lem}

\begin{proof}
Let $x\in \R^p$, we have:
\begin{equation*}
\norm{\psi(x)} = \exp\left(-\frac{1}{16\pi} \int_{-\infty}^{+\infty} \Re \left( 1-\exp\left(i\sum_{q=1}^p \frac{x_q}{t^{2q}}\right)\right) \dx t\right) = \exp\left(-\frac{1}{8\pi} I(x)\right),
\end{equation*}
where
\begin{equation*}
I(x) = \int_0^{+\infty} 1-\cos \left(\sum_{q=1}^p \frac{x_q}{t^{2q}} \right) \dx t.
\end{equation*}

Let $T(x) = \max \left\{ \left(\frac{2p}{\pi}\norm{x_q}\right)^\frac{1}{2q} \mvert 1\leq q \leq p \right\}$ and let $q_0 \in \{1,\dots,p\}$ be an index such that $\left(\frac{2p}{\pi}\norm{x_{q_0}}\right)^\frac{1}{2q_0} = T(x)$. For all $t \geq T(x)$, we have, for all $q \in \{1,\dots,p\}$, $t^{2q} \geq T(x)^{2q} \geq \frac{2p}{\pi}\norm{x_q}$. Hence,
\begin{equation*}
\norm{\sum_{q=1}^p \frac{x_q}{t^{2q}}} \leq \sum_{q=1}^p \frac{\norm{x_q}}{t^{2q}} \leq \frac{\pi}{2}.
\end{equation*}
Recall that, for any $u \in \R$, we have $1-\cos(u) \geq 0$. Moreover, for any $u \in [0,\frac{\pi}{2}]$, we have:
\begin{equation*}
1-\cos(u) = \int_0^u \sin(\theta) \dx \theta \geq \int_0^u \frac{2\theta}{\pi} \dx \theta = \frac{1}{\pi} u^2.
\end{equation*}
Thus,
\begin{equation*}
I(x) \geq \int_{T(x)}^{+\infty} 1-\cos\left(\sum_{q=1}^p \frac{x_q}{t^{2q}}\right) \dx t \geq \frac{1}{\pi} \int_{T(x)}^{+\infty} \left(\sum_{q=1}^p \frac{x_q}{t^{2q}}\right)^{\!2}\!\dx t \geq \frac{1}{\pi} \sum_{1\leq i,j \leq p} \frac{x_i x_j}{(2i+2j-1)T(x)^{2i+2j-1}}.
\end{equation*}

Let $B:\R^p \times \R^p \to \R$ be defined by $B:(y,z) \mapsto \sum_{1\leq i,j,\leq p} \frac{y_iz_j}{2i+2j-1}$. Then $B$ is a bilinear symmetric form. Moreover, for any $y \in \R^p$,
\begin{equation*}
B(y,y) = \sum_{1\leq i,j,\leq p} \frac{y_iy_j}{2i+2j-1} = \sum_{1\leq i,j,\leq p} y_iy_j \int_1^{+\infty} \frac{\dx t}{t^{2i+2j}} = \int_1^{+\infty} \left(\sum_{i=1}^p \frac{y_i}{t^{2i}}\right)^2 \dx t \geq 0.
\end{equation*}
If $B(y,y)=0$, then $\frac{1}{t^2}$ is a root of the polynomial $\sum_{i=1}^p y_i X^i$ for all $t \geq 1$, hence $y=0$. Thus $B$ is positive-definite and is a scalar product on $\R^p$. The associated norm is then equivalent to the standard Euclidean norm on $\R^p$. Therefore, there exists $C'>0$ such that, for all $y \in \R^p$, $B(y,y) \geq C' \Norm{y}^2$. Denoting by $\tilde{x} = \left(\frac{x_1}{T(x)^2},\dots,\frac{x_p}{T(x)^{2p}}\right)$, we have:
\begin{align*}
I(x) \geq \frac{T(x)}{\pi} B\left(\tilde{x},\tilde{x}\right) \geq \frac{C'}{\pi} T(x) \sum_{i=1}^p \left(\frac{x_i}{T(x)^{2i}}\right)^2 \geq \frac{C'}{\pi} T(x) \left(\frac{\norm{x_{q_0}}}{T(x)^{2q_0}}\right)^2 \geq \frac{C'\pi}{4p^2} T(x).
\end{align*}

Besides, since $\sum_{q=1}^p x_q^2 = \Norm{x}^2$, there exists $q \in \{1,\dots,p\}$ such that $\norm{x_q} \geq \frac{1}{p}\Norm{x}$. Then, we have $T(x) \geq \left(\frac{2p}{\pi}\norm{x_q}\right)^\frac{1}{2q} \geq \left(\frac{2}{\pi}\Norm{x}\right)^\frac{1}{2q}$. Finally, there exists $C>0$ such that, for any $x \in \R^p$, we have $\frac{1}{8\pi} I(x) \geq C \Norm{x}^\frac{1}{2p}$. Hence the result.
\end{proof}


\subsection{Polynomial maps}
\label{subsec polynomial maps}

In this section, we prove several results related with the polynomials $(P_p)_{p \geq 1}$ and $(Q_p)_{p \geq 1}$ defined by Def.~\ref{def P and Q}.

\begin{lem}
\label{lem relation P and Q}
For any $p \in \N^*$, we have the following relations in $\R[X_1,\dots,X_p]$:
\begin{align*}
X_p &= Q_p\left(P_1(X_1),\dots,P_p(X_1,\dots,X_p)\right),\\
X_p &= P_p\left(Q_1(X_1),\dots,Q_p(X_1,\dots,X_p)\right).
\end{align*}
\end{lem}

\begin{proof}
Let us consider the following power series, that are related to the generating functions of the sequences $(P_p)_{p \geq 1}$ and $(Q_p)_{p \geq 1}$. These are formal power series in infinitely many variables: $T$ and $X=(X_k)_{k \geq 1}$ (resp.~$T$ and $Y =(Y_k)_{k \geq 1}$). We have:
\begin{equation}
\label{eq formal series P}
\begin{aligned}
P(X,T) &= 1 + \sum_{p \geq 1} P_p(X_1,\dots,X_p) \frac{T^p}{p!}\\
&= \sum_{p \geq 0} \sum_{\alpha \in \mathcal{P}(p)} (-1)^{p-\norm{\alpha}} T^p \frac{1}{\alpha!}\prod_{q \geq 1}\left(X_q\right)^{\alpha_q}\\
&= \sum_{\alpha \in \ell_0} \frac{1}{\alpha!} \prod_{q \geq 1}\left((-1)^{q-1}X_qT^q\right)^{\alpha_q}\\
&= \sum_{d \geq 0} \frac{1}{d!} \sum_{\alpha \in \ell_0; \norm{\alpha}=d} \frac{d!}{\alpha!} \prod_{q\geq 1}\left((-1)^{q-1}X_qT^q\right)^{\alpha_q}\\
&= \exp\left(\sum_{q \geq 1}(-1)^{q-1}X_qT^q\right).
\end{aligned}
\end{equation}
Similarly, we have:
\begin{equation}
\label{eq formal series Q}
\begin{aligned}
Q(Y,T) &= \sum_{q \geq 1} (-1)^{q-1}Q_q(Y_1,\dots,Y_q) T^q\\
&= \sum_{q \geq 1} \sum_{\alpha \in \mathcal{P}(q)} \frac{(-1)^{\norm{\alpha}-1}}{\norm{\alpha}} \frac{\norm{\alpha}!}{\alpha!} \prod_{p\geq 1} \left(\frac{Y_pT^p}{p!}\right)^{\alpha_p}\\
&= \sum_{\alpha \in \ell_0 \setminus \{0\}} \frac{(-1)^{\norm{\alpha}-1}}{\norm{\alpha}} \frac{\norm{\alpha}!}{\alpha!} \prod_{p\geq 1} \left(\frac{Y_pT^p}{p!}\right)^{\alpha_p}\\
&= \sum_{d\geq 1} \frac{(-1)^{d-1}}{d} \sum_{\alpha \in \ell_0; \norm{\alpha}=d} \frac{d!}{\alpha!} \prod_{p \geq 1}\left(\frac{Y_pT^p}{p!}\right)^{\alpha_p}\\
&= \ln\left(1+ \sum_{p \geq 1}Y_p\frac{T^p}{p!}\right).
\end{aligned}
\end{equation}

Starting from Eq.~\eqref{eq formal series P} above, we apply the logarithm and obtain:
\begin{equation*}
\ln \left(1+\sum_{p \geq 1} P_p(X_1,\dots,X_p)\frac{T^p}{p!}\right) = \sum_{q \geq 1}(-1)^{q-1}X_qT^q.
\end{equation*}
The left-hand side has the same form as the last term in Eq.~\eqref{eq formal series Q}, with $Y_p = P_p(X_1,\dots,X_p)$ for all $p \geq 1$. Hence, we get:
\begin{equation*}
\sum_{q \geq 1}(-1)^{q-1}X_qT^q = \sum_{q \geq 1} (-1)^{q-1}Q_q\left(P_1(X_1),\dots,P_q(X_1,\dots,X_q)\right) T^q.
\end{equation*}
Identifying the coefficients according to their degree in the variable $T$ proves that for any $p \geq 1$, $X_p = Q_p\left(P_1(X_1),\dots,P_p(X_1,\dots,X_p)\right)$. The proof of the second claim is similar. Starting from Eq.~\eqref{eq formal series Q}, we apply the exponential and obtain:
\begin{equation*}
\exp\left(\sum_{q \geq 1} (-1)^{q-1}Q_q(Y_1,\dots,Y_q) T^q\right) = 1+ \sum_{p \geq 1}Y_p\frac{T^p}{p!}.
\end{equation*}
Using Eq.~\eqref{eq formal series P} with $X_p = Q_p(X_1,\dots,X_p)$ for all $p \geq 1$, and identifying the coefficients according to their degree in $T$ yields the result.
\end{proof}

\begin{dfn}
\label{def Phi p Psi p}
Let $p \in \N \setminus \{0,1\}$, we denote by $\Phi_p$ and $\Psi_p$ the polynomial maps from $\R^p$ to~$\R^p$ defined by:
\begin{align*}
&\Phi_p : (x_1,x_2,\dots,x_p) \longmapsto \left(P_1(x_1),P_2\left(x_1,x_2\right), \dots, P_p\left(x_1,x_2,\dots,x_p\right)\right),\\
&\Psi_p:(x_1,x_2,\dots,x_p) \longmapsto \left(Q_1(x_1),Q_2\left(x_1,x_2\right), \dots,Q_p\left(x_1,x_2,\dots,x_p\right)\right),
\end{align*}
where $(P_p)_{p \geq 1}$ and $(Q_p)_{p \geq 1}$ were defined by Def.~\ref{def P and Q}. We also denote by $\Phi'_p$ and $\Psi'_p$ the polynomial maps from $\R^{p-1}$ to itself defined by:
\begin{align*}
&\Phi'_p : (y_2,\dots,y_p) \longmapsto \left(P_2\left(1,y_2\right), \dots, P_p\left(1,y_2,\dots,y_p\right)\right),\\
&\Psi'_p:(y_2,\dots,y_p) \left(Q_2\left(1,y_2\right), \dots, Q_p\left(1,y_2,\dots,y_p\right)\right).
\end{align*}
\end{dfn}

\begin{lem}
\label{lem Phi and Psi diffeo}
Let $p \in \N \setminus \{0,1\}$, then $\Phi_p$ and $\Psi_p$ are smooth diffeomorphisms from $\R^p$ onto itself, such that $\Psi_p = (\Phi_p)^{-1}$. Moreover, for any $x \in \R^p$, the absolute value of the determinant of the differential of $\Phi_p$ at the point $x$ is $\norm{\det\left(d_x\Phi_p\right)} = \prod_{q=1}^p q!$.
\end{lem}

\begin{proof}
Since $\Phi_p$ and $\Psi_p$ are polynomial maps, it is enough to check that $\Psi_p \circ \Phi_p = \Id = \Phi_p \circ \Psi_p$ to prove the first claim. By Lem.~\ref{lem relation P and Q}, for any $(x_1,x_2,\dots,x_p) \in \R^p$ and any $q \in \{1,\dots,p\}$, we have:
\begin{equation*}
Q_q\left(P_1(x_1),P_2\left(x_1,x_2\right), \dots, P_q\left(x_1,x_2,\dots,x_q\right)\right) = x_q.
\end{equation*}
Hence, $\Psi_p \circ \Phi_p = \Id$. Similarly, by Lem.~\ref{lem relation P and Q}, for any $(x_1,x_2,\dots,x_p) \in \R^p$ and any $q \in \{1,\dots,p\}$, we have $P_q\left(P_1(x_1),Q_2\left(x_1,x_2\right), \dots, Q_q\left(x_1,x_2,\dots,x_q\right)\right)=x_q$, so that $\Phi_p \circ \Psi_p = \Id$.

Let $x \in \R^p$, the matrix of $d_x\Phi_p$ in the canonical basis of $\R^p$ equals $\begin{pmatrix}
\deron{P_i}{x_j}(x)
\end{pmatrix}_{1 \leq i,j \leq p}$. By Def.~\ref{def P and Q}, this matrix is lower triangular, and its diagonal coefficients are $\left((-1)^{q-1}q!\right)_{1 \leq q \leq p}$. Hence $\norm{\det\left(d_x\Phi_p\right)} = \prod_{q=1}^p q!$.
\end{proof}

\begin{cor}
\label{cor Phi' and Psi' diffeo}
Let $p \in \N \setminus \{0,1\}$, then $\Phi'_p$ and $\Psi'_p$ are smooth diffeomorphisms from $\R^{p-1}$ onto itself, such that $\Psi'_p = (\Phi'_p)^{-1}$. Moreover, for any $y \in \R^p$, we have $\norm{\det\left(d_y\Phi'_p\right)} = \prod_{q=2}^p q!$.
\end{cor}

\begin{proof}
By Lem.~\ref{lem Phi and Psi diffeo}, for any $y \in \R^{p-1}$, we have:
\begin{equation*}
(1,y) = \Phi_p \circ \Psi_p(1,y) = \Phi_p\left(1,\Psi'_p(y)\right)) = (1,\Phi'_p\circ \Psi'_p(y)).
\end{equation*}
Hence, $\Phi'_p\circ \Psi'_p=\Id$. Similarly, since $\Psi_p \circ \Phi_p = \Id$, we have $\Phi'_p \circ \Psi'_p=\Id$. Thus $\Psi'_p = (\Phi'_p)^{-1}$.
 
As in the proof of Lem.~\ref{lem Phi and Psi diffeo} above, by Def.~\ref{def P and Q} the matrix of $d_y\Phi'_p$ is lower triangular and its diagonal coefficients are $(-2!,\dots,(-1)^{p-1}p!)$. Hence $\norm{\det\left(d_y\Phi'_p\right)} = \prod_{q=2}^p q!$.
\end{proof}


\subsection{Positivity of the coefficients \texorpdfstring{$A_p$}{}}
\label{subsec coefficients}

In this final section, we prove the positivity of the coefficients $(A_p)_{p \geq 1}$, defined by Def.~\ref{def Ap}. The proof relies on the study of some modified Bessel functions. Let us denote by $I_0$ and $I_1$ the modified Bessel functions of the first kind of index $0$ and $1$ respectively. Recall that $I_0$ and $I_1$ are the entire functions defined by:
\begin{align*}
I_0(X) &= \sum_{p \geq 0} \frac{1}{(p!)^2} \left(\frac{X}{2}\right)^{2p} & &\text{and} & I_1(X) &= \sum_{p \geq 1} \frac{1}{(p-1)!p!} \left(\frac{X}{2}\right)^{2p-1} = I_0'(X).
\end{align*}
By Eq.~\eqref{eq generating function Ap}, we have:
\begin{equation}
\label{eq ln Bessel and Ap}
\ln(I_0(2X)) = \sum_{p \geq 1} (-1)^{p-1}A_pX^{2p}.
\end{equation}
Our strategy is to study the power series $\ln(I_0(X))$ and deduce the positivity of the $(A_p)_{p \geq 1}$ from Eq.~\eqref{eq ln Bessel and Ap}.

By~\cite[Eq.~(2.2)]{GS1978}, we have the following continued fraction expansion:
\begin{equation}
\label{eq continued fraction 1}
\frac{2}{X}\frac{I_1(X)}{I_0(X)}= \cfrac{1}{1+\cfrac{a_1 X^2}{1+\cfrac{a_2 X^2}{1+\cfrac{a_3 X^2}{1+\cdots}}}},
\end{equation}
where $a_k=\frac{1}{4k(k+1)}$ for all $k\geq 1$. This type of continued fraction can be written as a so-called S-fraction, which was first studied by Stieltjes. By~\cite[Eq.~(1.1)]{vAss1996}, we have:
\begin{equation}
\label{eq continued fraction 2}
\cfrac{1}{1+\cfrac{a_1 X^2}{1+\cfrac{a_2 X^2}{1+\cfrac{a_3 X^2}{1+\cdots}}}} =Z \cfrac{1}{Z+\cfrac{a_1}{1+\cfrac{a_2}{Z+\cfrac{a_3}{1+\cdots}}}},
\end{equation}
where $Z = X^{-2}$ and we added a missing factor $Z$ on the right-hand side. Following Stieltjes, van Assche explains in~\cite[Thm.~1.1 and Eq.~(1.5)]{vAss1996} that the S-fraction converges and that we have:
\begin{equation}
\label{eq continued fraction 3}
\cfrac{1}{Z+\cfrac{a_1}{1+\cfrac{a_2}{Z+\cfrac{a_3}{1+\cdots}}}} = \frac{s_0}{Z}+\sum_{k \geq 1}\frac{s_k}{Z+r_k}.
\end{equation}
Here $s_0 \geq 0$ and $(s_k)_{k \geq 1}$ is a sequence of positive coefficients such that $\sum_{k \geq 0} s_k =1$. Besides, $(r_k)_{k \geq 1}$ is a sequence of non-zero complex numbers such that $r_k \xrightarrow[k \to +\infty]{} 0$.

\begin{lem}
\label{lem expression Ap}
Let $(A_p)_{p \geq 1}$ be defined by Def.~\ref{def Ap} and the sequences $(s_k)_{k \geq 1}$ and $(r_k)_{k \geq 1}$ be those appearing in Eq.~\eqref{eq continued fraction 3}. We have $A_1=1$ and, for all $p \geq 2$, $A_p= \frac{4^{p-1}}{p}\sum_{k \geq 1} s_k r_k^{p-1}$.
\end{lem}

\begin{proof}
Recalling that $Z=X^{-2}$, by Eq.~\eqref{eq continued fraction 1}, \eqref{eq continued fraction 2} and \eqref{eq continued fraction 3}, we have:
\begin{equation*}
\frac{2}{X}\frac{I_1(X)}{I_0(X)} = s_0+\sum_{k \geq 1}\frac{s_k}{1+r_k X^2}.
\end{equation*}
Expanding the right-hand side, and using $\sum_{k \geq 0}s_k =1$, we get:
\begin{equation*}
\frac{2}{X}\frac{I_1(X)}{I_0(X)} = s_0+\sum_{k \geq 1}s_k \sum_{p \geq 0}(-1)^p r_k^p X^{2p}= 1 + \sum_{p\geq 1}(-1)^p \left(\sum_{k \geq 1} s_k r_k^p\right) X^{2p}.
\end{equation*}
Then, taking the derivative of Eq.~\eqref{eq ln Bessel and Ap}, we obtain:
\begin{equation*}
\sum_{p \geq 1} (-1)^{p-1}(2p)A_pX^{2p-1} =2\frac{I_1(2X)}{I_0(2X)}=2X + 2\sum_{p \geq 2}(-1)^{p-1} \left(\sum_{k \geq 1} s_k r_k^{p-1}\right) 4^{p-1} X^{2p-1}.
\end{equation*}
The result follows by identifying the coefficients of these power series.
\end{proof}

Since the $(s_k)_{k \geq 1}$ are positive, using Lem.~\ref{lem expression Ap} it is enough to prove that the $(r_k)_{k \geq 1}$ are positive to obtain the positivity of the $(A_p)_{p \geq 1}$. This is what we do in the following.

As van Assche points out in~\cite{vAss1996}, the convergents of the continued fraction on the left-hand side of Eq.~\eqref{eq continued fraction 3} are of the form:
\begin{equation*}
\frac{u_n\left(\frac{1}{Z}\right)}{v_n\left(\frac{1}{Z}\right)},
\end{equation*}
where $(u_n)_{n \geq 1}$ and $(v_n)_{n \geq 1}$ are sequences of polynomials such that, for all $n \geq 1$, the zeros of $v_n$ lie on the negative real half-line. Moreover, the continued fraction converges in the sense that $(u_n)$ and $(v_n)$ converge to entire functions $u$ and $v$ respectively, uniformly on compact sets. We have:
\begin{equation*}
\frac{u\left(\frac{1}{Z}\right)}{v\left(\frac{1}{Z}\right)} = \frac{s_0}{Z} + \sum_{k \geq 1} \frac{s_k}{Z+r_k},
\end{equation*}
so that the zeros of $v$ in $\C^*$ are the complex numbers of the form $-\frac{1}{r_k}$ with $k \geq 1$. In particular, $v \neq 0$, and its zeros are isolated points.

\begin{lem}
\label{lem zeros v}
Let $z \in \C$ be such that $v(z) = 0$, then $z$ is real and non-positive.
\end{lem}

\begin{proof}
First, let us prove that $v_n' \to v'$ uniformly on compact subsets of $\C$. Let $K$ be a compact subset of $\C$. Let $C$ be a closed contour such that $K$ is contained in the interior of $C$, and the distance from $C$ to $K$ is larger than $1$. For all $z \in K$, we have:
\begin{equation*}
\norm{v_n'(z)-v'(z)} = \norm{\frac{1}{2i\pi} \int_C \frac{v_n(w)-v(w)}{(w-z)^2} \dx w} \leq \frac{1}{2\pi} \int_C \norm{v_n(w)-v(w)} \dx w.
\end{equation*}
Since $C$ is compact, the right-hand side converges to $0$ as $n \to +\infty$.

Now, let us assume that there exists $z \in \C \setminus (-\infty,0]$ such that $v(z) = 0$. For any $R >0$, we denote by $C(R)$ the circle of center $z$ and radius $R$. Since $v$ is analytic and non-zero, its zeros are isolated. Then, there exists $R>0$ such that $C(R) \subset \C \setminus (-\infty,0]$ and, $\forall w \in C(R)$, $v(w) \neq 0$.

The number of zeros of $v_n$ lying inside $C(R)$, counted with multiplicity, is given by the following contour integral:
\begin{equation*}
\frac{1}{2i\pi}\int_{C(R)} \frac{v_n'(z)}{v_n(z)}\dx z.
\end{equation*}
This integral is well-defined and vanishes, because the zeros of $v_n$ are real and negative. On the compact set $C(R)$ we have $v_n\to v$ and $v_n' \to v'$ uniformly. Moreover, $\norm{v}$ is bounded from below. Hence, $\frac{v_n'}{v_n}$ converges uniformly to $\frac{v'}{v}$ on $C(R)$ and:
\begin{equation*}
0 = \frac{1}{2 i \pi}\int_C \frac{v_n'(z)}{v_n(z)}\dx z \xrightarrow[n \to +\infty]{} \frac{1}{2i \pi}\int_C \frac{v'(z)}{v(z)}\dx z >0.
\end{equation*}
This is absurd, which proves that the zeros of $v$ must be real and non-positive.
\end{proof}

\begin{lem}
\label{lem Ap positive}
For all $p \geq 1$, $A_p >0$.
\end{lem}

\begin{proof}
Let $(s_k)_{k \geq 1}$ and $(r_k)_{k \geq 1}$ be the sequences appearing in Eq.~\eqref{eq continued fraction 3}. For all $k \geq 1$, we have $v(-\frac{1}{r_k})=0$, hence $-\frac{1}{r_k}$ is negative by Lem.~\ref{lem zeros v}. Thus, both $(s_k)_{k \geq 1}$ and $(r_k)_{k \geq 1}$ are sequences of positive numbers. The result follows from the expression of the $(A_p)_{p \geq 1}$ derived in Lem.~\ref{lem expression Ap}.
\end{proof}


\bibliographystyle{amsplain}
\bibliography{RandomMomentsForTheNewEigenfunctionsOfPointScatterersOnRectangularFlatTori}

\end{document}